\documentclass[11pt]{article}

\usepackage[total={6.5in,9.0in}]{geometry}

\usepackage{setspace}

\newcommand{\remove}[1]{}

\usepackage{microtype}

\title{\textbf{Fast and Scalable Group Mutual Exclusion}}

\author{Shreyas Gokhale, Neeraj Mittal}

\date{}

\remove{
\titlerunning{Fast and Scalable GME} 

\subjclass{Computing methodologies $\rightarrow$ Concurrent computing methodologies $\rightarrow$ Concurrent algorithms}

}

\remove{

\funding{This work was supported, in part, by the National Science Foundation (NSF) under grants numbered CNS-1115733 and CNS-1619197.}

}

\usepackage[inline]{enumitem}
\usepackage{multicol}
\usepackage{pifont}
\usepackage{moresize}
\usepackage{comment}

\usepackage{tabu}
\usepackage{makecell}
\usepackage{multirow}

\usepackage[algoruled,longend,linesnumbered,noresetcount]{algorithm2e}
\SetAlFnt{\fontsize{8pt}{8pt}\selectfont}

\newlength\algowd

\newcounter{propcounter}
\newcommand*{\propcount}{\refstepcounter{propcounter}\thepropcounter}

\usepackage{amsmath}
\usepackage{amsthm}
\usepackage[nosort]{cleveref}
\crefname{figure}{figure}{figures}

\crefname{algocf}{Algorithm}{Algorithms}
\Crefname{algocf}{Algorithm}{Algorithms}
\crefname{AlgoLine}{Line}{Lines}
\Crefname{AlgoLine}{Line}{Lines}

\renewcommand{\eqref}[1]{(\ref{eq:#1})}

\newcommand{\Tabref}[1]{Table~\ref{tab:#1}}

\newcommand{\n}{n}
\newcommand{\m}{m}
\newcommand{\g}{\ell}
\newcommand{\s}{s}
\newcommand{\pc}{\dot{c}}
\newcommand{\ic}{\bar{c}}

\newcommand{\cp}{\dot{c}}
\newcommand{\ci}{\bar{c}}

\newcommand{\Ok}{$O(\p)$}
\newcommand{\ON}{$O(\n)$}

\newcommand{\OkLogNmin}{$O\big(\min \{\log \n, \p\}\big)$}

\newcommand{\GME}{GME}  

\newcommand{\cmark}{\ding{51}}%
\newcommand{\xmark}{\ding{55}}%

\newcommand{\LL}{\textsf{LL{}}\xspace}
\newcommand{\SC}{\textsf{SC{}}\xspace}
\newcommand{\LLSC}{\LL/\SC\xspace}
\newcommand{\FAI}{\textsf{FAI}}
\newcommand{\FAD}{\textsf{FAD}}
\newcommand{\FAA}{\textsf{FAA}}
\newcommand{\CAS}{\textsf{CAS}}
\newcommand{\FAS}{\textsf{FAS}}

\newcommand{\entry}{entry\xspace}
\newcommand{\exit}{exit\xspace}

\newcommand{\Entry}{Entry\xspace}

\let\oldnl\nl
\newcommand{\nonl}{\renewcommand{\nl}{\let\nl\oldnl}}
\newcommand{\RNum}[1]{\uppercase\expandafter{\romannumeral #1\relax}}

\newcommand{\size}{size\xspace}
\newcommand{\state}{state\xspace}

\newcommand{\open}{open\xspace}
\newcommand{\closed}{closed\xspace}
\newcommand{\adjourned}{adjourned\xspace}

\newcommand{\retired}{retired\xspace}
\newcommand{\close}[1][d]{close{#1}\xspace}
\newcommand{\adjourn}[1][ed]{adjourn{#1}\xspace}
\newcommand{\retire}[1][d]{retire{#1}\xspace}

\SetKw{True}{true}
\SetKw{False}{false}
\SetKw{LAnd}{and}
\SetKw{LOr}{or}
\SetKw{LNot}{not}
\SetKw{Continue}{continue}
\SetKw{Break}{break}
\SetKw{Enum}{enum}
\SetKw{Null}{null}
\SetKw{Integer}{integer}
\SetKw{Boolean}{bool}
\SetKw{Struct}{struct}
\SetKw{Shared}{shared variables}
\SetKw{Private}{private variables}
\SetKw{Array}{array}
\SetKw{Initialization}{initialization}
\SetKw{Initially}{initially}
\SetKw{New}{new}
\SetKwRepeat{Repeat}{repeat}{until}
\SetKw{Additional}{additional}

\newcommand{\pack}[1]{\{{#1}\}}

\newcommand{\instance}{instance}

\newcommand{\session}{session}
\newcommand{\gate}{state}
\newcommand{\nextinlist}{next}
\newcommand{\previousinlist}{prev}

\newcommand{\pointer}{\rightarrow}

\newcommand{\head}{head\xspace}
\newcommand{\headArray}{head\xspace}
\newcommand{\announceArray}{announce\xspace}
\newcommand{\snapshotArray}{snapshot\xspace}
\newcommand{\hpArray}{hp\xspace}
\newcommand{\whichArray}{which\xspace}
\newcommand{\markerArray}{marker\xspace}

\newcommand{\other}{other}
\newcommand{\counter}{number}

\newcommand{\node}{node}
\newcommand{\current}{current}

\newcommand{\myid}{me}
\newcommand{\myinstance}{myinstance}
\newcommand{\mysession}{mysession}
\newcommand{\mynode}{mynode}

\newcommand{\successor}{successor}

\newcommand{\homogeneous}{homogeneous}
\newcommand{\homogeneity}{homogeneity}

\newcommand{\mine}{mine}
\newcommand{\helpee}{helpee}

\newcommand{\guard}{guard\xspace}
\newcommand{\vacant}{vacant\xspace}

\newcommand{\LEADERFLAG}{\texttt{LEADERLESS}\xspace}
\newcommand{\CONFLICTFLAG}{\texttt{CONFLICT}\xspace}
\newcommand{\VACANTFLAG}{\texttt{VACANT}\xspace}
\newcommand{\RETIREDFLAG}{\texttt{RETIRED}\xspace}

\newcommand{\flag}{\mathit{flag}}

\newcommand{\poolArray}{pool}

\newcommand{\owner}{owner}
\newcommand{\condition}{condition\xspace}

\newcommand{\CLEANFLAG}{\texttt{SAFE}\xspace}
\newcommand{\DIRTYFLAG}{\texttt{UNSAFE}\xspace}
\newcommand{\UNKNOWNFLAG}{\texttt{UNKNOWN}\xspace}

\newcommand{\IsClosed}{\textsc{IsClosed}}
\newcommand{\IsAdjourned}{\textsc{IsAdjourned}}
\newcommand{\IsRetired}{\textsc{IsRetired}}

\newcommand{\SetLeaderOrConflictFlag}{\textsc{SetGuardFlag}}
\newcommand{\SetVacantFlag}{\textsc{SetVacantFlag}}
\newcommand{\MarkAsRetired}{\textsc{MarkAsRetired}}

\newcommand{\ReadHead}{\textsc{ReadHead}}
\newcommand{\UpdateHead}{\textsc{AdvanceHead}}
\newcommand{\TestHead}{\textsc{TestHead}}

\newcommand{\InitializeNode}{\textsc{GetNewNode}}
\newcommand{\FetchNode}{\textsc{SelectNextNode}}
\newcommand{\AppendNode}{\textsc{AppendNextNode}}
\newcommand{\ReclaimNode}{\textsc{RetireNode}}

\newcommand{\Cleanup}{\textsc{Cleanup}}

\newcommand{\Notify}{\textsc{Notify}}
\newcommand{\NotifyAll}{\textsc{NotifyAll}}

\newcommand{\TryToEnter}{\textsc{Enter}}
\newcommand{\TryToLeave}{\textsc{Exit}}

\newcommand{\bitor}{\mathrel{|}}
\newcommand{\bitand}{\mathrel{\&}}

\newcommand{\na}{--}

\newcommand{\clean}{safe\xspace}
\newcommand{\dirty}{unsafe\xspace}
\newcommand{\unknown}{unknown\xspace}

\newcommand{\apool}{active\xspace}
\newcommand{\ppool}{passive\xspace}

\newcommand{\skiplistgme}{40\%}
\newcommand{\gainnextbest}{189\%}

\newcommand{\GLBGME}{GLB-GME}
\newcommand{\FSGME}{FS-FME}
\newcommand{\BHGME}{BH-GME}

\newcommand{\myparagraph}[1]{\paragraph{#1:}}

\newcommand{\sessionOf}[1]{s(#1)}
\newcommand{\safeSet}{\mathcal{S}}

\newcommand{\readyArray}{ready}

\usepackage{graphicx}
\usepackage{rotating}

\usepackage{tikz}
\usepackage{pgfplots}
\usetikzlibrary{calc,shapes.geometric,arrows,fit,matrix,automata,positioning,shapes.symbols,shapes.misc}
\usetikzlibrary{fit,decorations.pathmorphing}
\usetikzlibrary{pgfplots.groupplots}
\usetikzlibrary{shadings,patterns}
\pgfplotsset{compat=1.13}

\newcommand{\STAB}[1]{\begin{tabular}{@{}c@{}}#1\end{tabular}}

\newtheorem{lemma}{Lemma}
\newtheorem{theorem}{Theorem}

\usepackage{subcaption}

\begin{document}

\maketitle

\begin{abstract}

The \emph{group mutual exclusion (GME)} problem is a generalization of the classical mutual exclusion problem in which every critical section is associated with a \emph{type} or \emph{session}. Critical sections belonging to the same session can execute concurrently, whereas critical sections belonging to different sessions must be executed serially. The well-known read-write mutual exclusion problem is a special case of the group mutual exclusion problem.

In this work, we present a new GME algorithm for an asynchronous shared-memory system that, in addition to satisfying lockout freedom, bounded exit and concurrent entering properties, has $O(1)$ step-complexity 
when the system contains no conflicting requests \emph{as well as} $O(1)$ space-complexity per GME object when the system contains sufficient number of GME objects. 
To the best of our knowledge, no existing GME algorithm has $O(1)$ step-complexity for concurrent entering. 
The RMR-complexity of a request is only $O(\pc)$ in the amortized case, where $\pc$ denotes the point contention of the request.

In our experimental results, our GME algorithm vastly outperformed two of the existing GME algorithms especially for higher thread count values by as much as \gainnextbest{} in some cases.

\end{abstract}

\section{Introduction}

The \emph{group mutual exclusion (GME)} problem is a generalization of the classical mutual exclusion (ME) problem in which every critical section is associated with a \emph{type} or \emph{session}~\cite{Jou:2000:DC}. Critical sections belonging to the same session can execute concurrently, whereas critical sections belonging to different sessions must be executed serially. The GME problem models situations in which a resource may be accessed at the same time by processes of the same group, but not by processes of different groups. As an example, suppose data is stored on multiple discs in a shared CD-jukebox. When a disc is loaded into the player, users that need data on that disc can access the disc concurrently, whereas users that need data on a different disc have to wait until the current disc is unloaded~\cite{Jou:2000:DC}. Another example includes a meeting room for philosophers interested in different forums or topics~\cite{Jou:2002:DC,TakIga:2003:COCOON}.

The well-known readers/writers problem is a special case of the group mutual exclusion problem in which all read critical sections belong to the same session but every write critical section belongs to a separate session.

Note that any algorithm that solves the mutual exclusion problem also solves the group mutual exclusion problem. However, the solution is inefficient since critical sections are executed in a serial manner and thus the solution does not permit any concurrency. To rule out such inefficient solutions, a group mutual exclusion algorithm needs to satisfy concurrent entering property. Roughly speaking, the concurrent entering property states that if all processes are requesting the same session, then they \emph{must} be able to execute their critical sections concurrently.

The GME problem has been defined for both message-passing and shared-memory systems. The focus of this work is to develop an efficient GME algorithm for shared-memory systems. 
Recently, GME-based locks have been used to improve the performance of lock-based concurrent skip lists for multi-core systems using the notion of unrolling by storing multiple key-value pairs in a single node~\cite{Pla:2017:PhD}. Unlike in a traditional skip list, most update operations in an unrolled skip list do not need to make any structural changes to the list. This can be leveraged to allow multiple insert operations or multiple delete operations (but not both) to act on the same node simultaneously in most cases. To make structural changes to the list, an operation needs to acquire exclusive locks on the requisite nodes as before. Note that implementing this idea requires GME-based locks; read-write locks  do not suffice since a lock needs to support two distinct shared modes. Experimental evaluation showed that, using GME-based locks, can improve the performance of a concurrent (unrolled) skip list by more than \skiplistgme{}~\cite{Pla:2017:PhD}.

\subsection{Related Work} 
Since the GME problem was first introduced by Joung around two decades ago~\cite{Jou:2000:DC}, several algorithms  have been proposed to solve the problem for shared-memory systems~\cite{Jou:2000:DC,KeaMoi:1999:PODC,Had:2001:PODC,TakIga:2003:COCOON,JayPet+:2003:PODC,DanHad:2004:DISC,BhaHua:2010:PODC,HeGop+:2016:ICDCN}. These algorithms provide different trade-offs between fairness, concurrency, step complexity and space complexity. 
A detailed description of the related work is given in~\cref{sec:related}. 
To the best of our knowledge, all of the prior work suffers from at least one and possibly both of the following drawbacks. 

\medskip

\noindent
\textbf{Drawback 1 (high step complexity in the absence of any conflicting request):}
In a system using fine-gained locking, most of the lock acquisitions are likely to be uncontended 
in practice (\emph{i.e.}, at most one process is trying to acquire a given lock). Note that this is the primary motivation behind providing a fast-path mechanism for acquiring a lock~\cite{HerSha:2012:Book}. Moreover, in concurrent unrolled skip lists using GME-based locks~\cite{Pla:2017:PhD}, most of the lock acquisitions involve only two shared sessions. In many cases, all requests are likely to be for the same session. This necessiates the need for a GME algorithm that has low step-complexity when all requests for acquiring a given lock are for the same session, which we refer to as \emph{concurrent entry step complexity}. (Note that this includes the case where there is only one request for lock acquisition.)

To the best of our knowledge, except for two, all other existing GME algorithms have  concurrent entry step complexity of $\Omega(\n)$, where $\n$ denotes the number of processes in the system. The GME algorithm by Bhatt and Huang~\cite{BhaHua:2010:PODC} has concurrent entry step-complexity of $O(\min\{\log \n, \pc\})$, where $\pc$ denotes the point contention of the request. Also, one of the GME algorithms by Danik and Hadzilcos~\cite[Algorithm~3]{DanHad:2004:DISC} has concurrent entry step complexity of $O(\log \s \cdot \min\{\log \n, \pc\})$, where $\s$ denotes the number of different types of sessions.

\medskip

\noindent
\textbf{Drawback 2 (high space complexity with a large number of GME objects):}
All the existing work in this area has (implicitly) focused on a \emph{single} GME object. 
However, many systems use fine-grained locking to achieve increased scalability in multi-core/multi-processor systems.
For example, each node in a concurrent data structure is protected by a separate lock~\cite{HerSha:2012:Book,Pla:2017:PhD}.

All the existing GME algorithms that guarantee starvation freedom have a space-complexity of at least $\Theta(\n)$ for a single GME object. Note that this is expected because mutual exclusion is a special case of group mutual exclusion and any starvation-free mutual exclusion
algorithm requires $\Omega(\n)$ space even when powerful atomic instructions such as compare-and-swap are used~\cite{FicHen+:2006:DC}.
Some of these GME algorithms (\emph{e.g.}, \cite{KeaMoi:1999:PODC,DanHad:2004:DISC,HeGop+:2016:ICDCN}) can be modified relatively easily to share the bulk of this space among all GME objects and, as a result, the additional space usage for each new GME object is only $O(1)$. 
However, it is not clear how the other GME algorithms (\emph{e.g.}, \cite{Jou:2000:DC,Had:2001:PODC,TakIga:2003:COCOON,JayPet+:2003:PODC,DanHad:2004:DISC,BhaHua:2010:PODC,HeGop+:2016:ICDCN}) can be modified to achieve the same space savings. For these GME algorithms, to our understanding, the additional space usage for each new GME object is at least $\Theta(\n)$. We refer to the former set of GME algorithms as \emph{space-efficient} and the latter set of GME algorithms as \emph{space-inefficient}. 

Consider the example of a concurrent data structure using GME-based locks to improve performance~\cite{Pla:2017:PhD}. If $\n$ is relatively large, then the size of a node equipped with a lock based on a GME algorithm that is space-inefficient may be several factors more than its size otherwise. This will \emph{significantly} increase the memory footprint of the concurrent data structure, which, in turn, will adversely affect its performance and may even negate the benefit of increased concurrency resulting from using a GME-based lock.

\subsection{Our Contributions} 
In this work, we present a new GME algorithm that, in addition to satisfying the group mutual exclusion, lockout freedom, bounded exit, concurrent entering and bounded space variable properties, has the following desirable features.
First, it has $O(1)$ concurrent entry step complexity. Note that, as a corollary, a process can enter its critical section within a constant number of its own steps in the absence of any other request, which is typically referred to as contention-free step complexity. To the best of our knowledge, no existing GME algorithm has $O(1)$ concurrent entry step-complexity.
Second, it uses only $O(\m + \n^2)$ space for managing $\m$ GME objects, where $O(\n^2)$ space is shared among all $\m$ GME objects. In addition, each process needs only $O(\g)$ space, where $\g$ denotes the maximum number of GME objects (or locks) a process needs to hold at the same time, which is space-optimal. 
Third, the number of remote references made by a request under the cache-coherent model, which is referred to as RMR complexity, is $O(\min\{\ic,\n\})$ in the worst case and $O(\pc)$ in the amortized case, where $\ic$ denotes the interval contention of the request.
Finally, it uses bounded space variables.
As in~\cite{DanHad:2004:DISC}, our algorithm uses two read-modify-write (RMW) instructions, namely compare-and-swap (\CAS) and fetch-and-add (\FAA), both of which are commonly available on modern processors including x86\_64 and AMD64.

In our experimental results, our GME algorithm vastly outperformed two of the well-known existing GME algorithms~\cite{BhaHua:2010:PODC,HeGop+:2016:ICDCN} especially for higher thread count values by as much as \gainnextbest{} in some cases.

We show elsewhere that our algorithm can be easily adapted to achieve optimal RMR complexity of $O(\n)$ under the distributed shared memory (DSM) model, while maintaining all the other aforementioned desirable properties.

\subsection{Roadmap}
The rest of the text is organized as follows. 
We present the system model and describe the problem in \cref{sec:model|definition}.
In \cref{sec:algorithm}, we describe our GME algorithm, prove its correctness and analyze its complexity.
\Cref{sec:related} describes the related work. 
Finally, \cref{sec:conclusion|future} concludes the text and outlines directions for future work.

\section{System Model and Problem Specification}
\label{sec:model|definition}

\subsection{System Model}

We consider an asynchronous shared-memory system consisting of $\n$ processes labeled $p_1, p_2, \ldots, p_\n$. Each process also has its own private variables. Processes can only communicate by performing read, write and read-modify-write (RMW) instructions on shared variables.
A system execution is modeled as a sequence of process steps. In each step, a process  either performs  some local computation affecting only its private variables or executes one of the available instructions (read, write or RMW) on a shared variable. Processes take steps asynchronously. 
This means that in any execution, between two successive steps of a process, there can be an unbounded but finite number of steps performed by other processes. 


\subsection{Synchronization Instructions}

We assume the availability of two RMW instructions, namely \emph{compare-and-swap (\CAS)} and \emph{fetch-and-add (\FAA)}.

A compare-and-swap instruction takes a shared variable $x$ and two values $u$ and $v$ as inputs. If the current value of $x$ matches $u$, then it writes $v$ to $x$ and returns true. Otherwise, it returns false.

A fetch-and-add instruction takes a shared variable $x$ and a value $v$ as inputs, returns the current value of $x$ as output, and, at the same time, increments the value of $x$ by $v$.

\subsection{Problem Specification}

\begin{algorithm}[t]
\DontPrintSemicolon
\While{true}{
\textsc{Non-Critical Section (NCS)} \;
\textsc{Entry Section} 
\tcp*{try to enter critical section}
\textsc{Critical Section (CS)} 
\tcp*{execute critical section}
\textsc{Exit Section} 
\tcp*{exit critical section}
}

\caption{Structure of a GME Algorithm}
\label{algo:GAS}
\end{algorithm}

In the GME problem, each process repeatedly executes four sections of code, namely \emph{non-critical section (NCS)}, \emph{\entry section}, \emph{critical section (CS)}, and \emph{\exit section}, as shown in \cref{algo:GAS}.
Each critical section is associated with a \emph{type} or a \emph{session}. 
Critical sections belonging to the \emph{same session} can execute \emph{concurrently}, whereas critical sections belonging to \emph{different sessions} must be executed \emph{serially}. 
We refer to the code executed by a process from the beginning of its \entry section until the end of its \exit section as an \emph{passage}. Note that the session associated with a critical section may be different in different passages (and is selected based on the needs of the underlying application).
We say that a process has an \emph{outstanding} request if it is in one of its \emph{passages}. 

We assume that every process is \emph{live} meaning that, if it is not executing its non-critical section, then it will eventually execute its next step.

\subsubsection{Correctness Properties}

Solving the GME problem involves designing code for \entry and \exit sections in order to ensure the following four properties are satisfied in each passage:


\begin{description}
\item[(P\propcount) Group mutual exclusion] If two processes are in their critical sections at the same time, then they have requested the same session. 
\label{prop:gme}
\item[(P\propcount) Lockout freedom] If a process is trying to enter its critical section, then it is able to do so eventually (\entry section is finite).
\item[(P\propcount) Bounded exit] If a process is trying to leave its critical section, then it is able to do so eventually within a bounded number of its own steps (\exit section is bounded).
\item[(P\propcount) Concurrent entering] If a process is trying to enter its critical section and all current and future requests are for the same session, then the (former) process is able to enter its critical session eventually within a bounded number of its own steps (\entry section is bounded in the absence of a request for a different session). 
\end{description}

\subsubsection{Complexity Measures}
We say that two requests \emph{conflict} if they involve the same GME object but belong to different sessions. 
We say that a request is \emph{outstanding} until its process has finished executing the \exit section.
We use the following metrics to evaluate the performance of our GME algorithm:

\begin{description}

\item[Context-switch complexity]
It is defined as the maximum number of sessions that can be established while a process is waiting to enter its critical section. 
It is also referred to as session switch complexity elsewhere~\cite{KeaMoi:1999:PODC,Had:2001:PODC}.

\item[Concurrent entrering step complexity]
It is defined as the maximum number of steps a process has to execute in its \entry and \exit sections provided  no other process in the system  has an outstanding conflicting request during that period.

\item[Remote memory reference (RMR) complexity] 
It is defined as the maximum number of remote memory references required by a process in its \entry and \exit sections.

\end{description}

In addition, we also consider the memory footprint of the GME algorithm when the system contains multiple GME objects. 

\begin{description}

\item[Multi-object space complexity] 
It is defined as the maximum amount of space needed to instantiate and maintain a certain number of GME objects.

\end{description}

We analyze the RMR complexity of a GME algorithm is under the \emph{cache-coherent (CC)} model, which is the most common model used for RMR complexity analysis.

In the CC model, all shared variables are stored in a central location or global store. Each processor has a private cache. When a process accesses a shared variable, a copy of the contents of the variable is saved in the private cache of the process. Thereafter, every time the process reads that shared variable, it does so using its cached (local) copy until the cached copy is invalidated. Also, every time a process writes to a shared variable, it writes to the global store, which also invalidates all cached copies of the variable. In the CC model, spinning on a memory location generates at most two RMRs---one when the variable is cached and the other when the cached copy is invalidated.

In the DSM model, instead of having the shared memory in a central location or a global store, each process ``owns'' a part of the shared memory and keeps it in its own local memory. Every shared variable is stored in the local memory of some process. 
Accessing a shared variable stored in the local memory of a different process causes the process to make a remote memory reference. A reference to a variable stored in a non-local memory requires traversing the processor-to-memory interconnect, which takes much longer to access than to access a locally stored variable. In the DSM model, spinning on a variable that is stored in remote memory may generate an unbounded number of RMRs.

An algorithm is called \emph{local-spinning} if the maximum number of RMRs made in \entry and \exit sections is bounded. 
It is desirable to design algorithms that minimize the number of remote memory references because this factor can critically affect the performance of these algorithms~\cite{MelSco:1991:trcs}.

We express the RMR complexity of our GME algorithm using the following measures of contention.
The \emph{interval contention} of a passage $\pi$, denoted by $\ci(\pi)$, is defined
as the total number of passages involving the same GME object as $\pi$ that overlap with $\pi$.
The \emph{point contention} of a passage $\pi$, denoted by $\cp(\pi)$, is defined
as the maximum number of passages involving the same GME object as $\pi$  that are simultaneously in progress in the system at any point
during $\pi$.

\section{The Group Mutual Exclusion Algorithm}
\label{sec:algorithm}

\subsection{The Main Idea}

Our GME algorithm is inspired by Herlihy's universal construction for deriving a wait-free linearizable implementation of a concurrent object from its sequential specification using consensus objects~\cite{Her:1991:TOPLAS,Her:1993:TOPLAS}. Roughly speaking, the universal construction works as follows. The state of the concurrent object is represented using
\begin{enumerate*}[label=(\roman*)]
\item its initial state and 
\item the sequence of operations that have applied to the object so far. 
\end{enumerate*}
The two attributes of the object are maintained using a singly linked list in which the first node represents the initial state and 
the remaining nodes represent the operations. To perform an operation, a process first creates a new node and initializes it with all the relevant details of the operation, namely its type and its input arguments. It then tries to append the node at the end of the list. To manage conflicts in case multiple processes are trying to append their own node to the list, a consensus object  is used to determine which of several nodes is chosen to be appended to the list. Specifically, every node stores a consensus object and the consensus object of the current last node is used to decide its successor (\emph{i.e.}, the next operation to be applied to the object). A process whose node is not selected simply tries again. A helping mechanism is used to guarantee that every process trying to perform an operation eventually succeeds in appending its node to the list.

We modify the aforementioned universal construction to derive a GME algorithm that satisfies several desirable properties. Intuitively, an operation in the universal construction corresponds to a critical section request in our GME algorithm. \emph{Appending a new node} to the list thus corresponds to \emph{establishing a new session}. However, unlike in the universal construction, a single session in our GME algorithm can be used to satisfy multiple critical section requests. This basically means that every critical section request does not cause a new node to be appended to the list. This requires some careful bookkeeping so that no ``useless'' sessions are established. Further, a simple consensus algorithm, implemented using \CAS{} instruction, is used to determine the next session to be established.

We describe our GME algorithm in an incremental manner. First, we describe a basic GME algorithm that is only deadlock-free (some session request is eventually satisfied but a given request may be starved), and uses unbounded space. Next, we enhance the basic algorithm to achieve starvation freedom (every session request is eventually satisfied) using a helping mechanism.  
Finally, we  enhance the algorithm to make it space-efficient by reusing nodes using a memory reclamation algorithm.
Note that all our algorithms are safe in the sense that they satisfy the group mutual exclusion property.

For ease of exposition, we describe the first two variants in the next section, \cref{sec:starvation|free:algorithm} along with a correctness proof and complexity analysis. We then describe the third (and the final) variant in the section thereafter, \cref{sec:space|efficient:algorithm}.


\subsection{A Starvation-Free Algorithm}
\label{sec:starvation|free:algorithm}

\begin{algorithm}[!t]
\tcp{Node of a list}
\Struct Node \{ \\
\label{line:node}
\Indp 
\nonl
\Integer $\session$\tcp*[r]{session associated with the node}
\nonl
\Integer $\instance$\tcp*[r]{instance identifier of the GME object}
\nonl
\Integer $\counter$\tcp*[r]{the next process to be helped}
\nonl
\pack{\Boolean,\Boolean,\Boolean,\Boolean} $\gate$\tcp*[r]{four flags representing state}
\nonl
\Integer $\size$\tcp*[r]{number of processes currently in the session}
\nonl
NodePtr $\previousinlist$, $\nextinlist$\tcp*[r]{address of the previous and next nodes}
\nonl
\Integer $\owner$\tcp*[r]{the last process to own the node}
\Indm
\nonl
\}\;
\BlankLine\BlankLine
\Shared \\
\Indp
$\headArray$: \Array[$1\ldots\m$] of NodePtr\tcp*[r]{to store references to head nodes of lists}
\label{line:head|array}
$\announceArray$: \Array[$1\ldots\n$] of NodePtr, \Initially [\Null,$\ldots$,\Null{}]\tcp*[r]{to announce CS requests}
\label{line:announce|array}
\Indm
\BlankLine\BlankLine
\Private \\
\Indp
$\snapshotArray$: \Array[$1\ldots\n$] of NodePtr\tcp*[r]{to store snapshots of the head nodes} 
\label{line:snapshot|array}
\tcp{$\snapshotArray[i]$ is a private variable of process $p_i$}
\Indm
\BlankLine\BlankLine
\nonl
\Initialization \\
\label{line:initialize|begin}
\Begin{
	\tcp{initialize variables}
	\ForEach{$i \in [1\ldots\m]$}
	{
	   $\head[i]$ := \New Node\tcp*[r]{create a new node}
	   $\head[i] \pointer \gate$ := \LEADERFLAG\tcp*[r]{session has no leader}
	   $\head[i] \pointer \size$ := 0\tcp*[r]{session has no processes}
	   $\head[i] \pointer \nextinlist$ := \Null{}\tcp*[r]{node has no successor}
	   \tcp{all other fields can be initialized arbitrarily}
	}
	\ForEach{$i \in [1\ldots\n]$}
	{
		$\announceArray[i]$ := \Null\tcp*[r]{process has no outstanding request}
	}
    \label{line:initialize|end}	
}

\caption{Data types and variables used.}
\label{algo:types|variables}
\end{algorithm}


\begin{algorithm}[!t]
\tcp{returns true if the session is \closed and false otherwise}
\Boolean \IsClosed(\Integer $\gate$) 
   \quad \{ \Return ($\gate$ $\bitand$ \LEADERFLAG)  \LAnd{} ($\gate$ $\bitand$ \CONFLICTFLAG); \} \\
   \label{line:isclosed}
\BlankLine\BlankLine
\tcp{returns true if the session is \adjourned and false otherwise}
\Boolean \IsAdjourned(\Integer $\gate$) 
   \quad\{ \Return ($\gate$ $\bitand$ \VACANTFLAG); \} \\
   \label{line:isadjourned}
\BlankLine\BlankLine
\tcp{returns true if the node is \retired and false otherwise}
\Boolean \IsRetired(\Integer $\gate$) 
   \quad \{ \Return ($\gate$ $\bitand$ \RETIREDFLAG); \} \\
   \label{line:isretired}
\BlankLine\BlankLine
\tcp{sets a given \guard flag (\LEADERFLAG or \CONFLICTFLAG) in the session state}
\SetLeaderOrConflictFlag(NodePtr $\node$, \Boolean $\flag$ ) \\
\Begin{
   \While{\True}{
      \label{line:guard:while|begin}
      \Integer $\gate$ := $\node \pointer \gate$ \tcp*[r]{read the current state}
      \lIf(\tcp*[f]{flag already set}){($\gate$ $\bitand$ $\flag$)}{\Return}
      \BlankLine
      \lIf(\tcp*[f]{successfully set the flag}){\CAS{}($\node \pointer \gate$, $\gate$, $\gate$ $\bitor$ $\flag$)}{ \Return}
     
   }
   \label{line:guard:while|end}

}
\BlankLine\BlankLine
\tcp{sets the \vacant flag in the state if possible}
\Boolean \SetVacantFlag(NodePtr $\node$) \\
\Begin{
   \Integer $\gate$ := $\node \pointer \gate$\tcp*[r]{read the current state}
   \label{line:vacant:state|read}
   \BlankLine
   \lIf(\tcp*[f]{session is still open}){\LNot{}(\IsClosed($\gate$))}{\Return}
   \lIf(\tcp*[f]{session still has participants}){($\node \pointer \size \neq 0$)}{\Return}
   \BlankLine
   \Return \CAS{}($\node \pointer \gate$, $\gate$, $\gate$ $\bitor$ \VACANTFLAG)\;
   \label{line:vacant:state|write}
}
\BlankLine\BlankLine
\tcp{mark the node as retired}
\MarkAsRetired(NodePtr $\node$) 
   \quad \{ $\node \pointer \gate$ := \LEADERFLAG $\bitor$ \CONFLICTFLAG $\bitor$ \VACANTFLAG $\bitor$ \RETIREDFLAG; \} \\
   \label{line:mark:retired}
\caption{Functions operating on session state.}
\label{algo:methods|state}
\end{algorithm}


\begin{algorithm}[!t]
\tcp{reads the current head pointer of the list}
\ReadHead(\Integer $\instance$) \\
\Begin{
   $\snapshotArray[\myid]$ := $\headArray[\instance]$\;
   \label{line:readhead:read}
}
\BlankLine\BlankLine
\tcp{returns true if the head of the list has not moved and false otherwise}
\Boolean \TestHead(\Integer $\instance$) \\
\Begin{
	 \lIf(\tcp*[f]{head has advanced}){($\headArray[\instance] \neq \snapshotArray[\myid]$)} 
	 {
	    \label{line:testhead:not|match}
      \Return \False
   } \lElse { \Return \True } 
   \label{line:testhead:match}

}
\BlankLine\BlankLine
\tcp{advances the head of a list to the given node if the head has not moved}
\UpdateHead(\Integer $\instance$, NodePtr $\successor$) \\
\Begin{
   
 	\CAS{}($\headArray[\instance]$, $\snapshotArray[\myid]$, $\successor$)\;
 	\label{line:updatehead:CAS}
}
\caption{Functions operating on list head.}
\label{algo:methods|metadata}
\end{algorithm}


\begin{algorithm}[!t]
\tcp{code for \entry section}
\TryToEnter(\Integer $\myinstance$, \Integer $\mysession$) \\
\Begin{
 
   \tcp{initialize a node and announce the request to other processes}
   \InitializeNode($\myinstance$, $\mysession$)\;
   \label{line:trytoenter|begin}
   NodePtr $\mynode$ := $\announceArray[\myid]$\;
   \BlankLine
   \While{\True}{
  	  \label{line:trytoenter:while|begin}	
      \ReadHead($\myinstance$)\tcp*[r]{read the head pointer of the list}
      \label{line:trytoenter:readhead}
      NodePtr $\current$ := $\snapshotArray[\myid]$\tcp*[r]{find the last node in the list}
      \BlankLine
      \If{($\current$ = $\mynode$)}{
         \tcp{join the session as a leader and \retire[] the predecessor node}
         \ReclaimNode($\mynode \pointer \previousinlist$)\;
      	 \label{line:trytoenter:retire|leader}
       	 \Return\;
      }
      \label{line:trytoenter:node|match}
      \BlankLine
      \uIf(\tcp*[f]{my request is compatible with the current session}){($\current \pointer \session$ = $\mysession$)}
      {
         \label{line:trytoenter:session|match}
         \If(\tcp*[f]{the session is open}){\LNot{}(\IsClosed($\current \pointer \gate$))}
         {
            \label{line:trytoenter:head:open}
		    \tcp{attempt to join the session as a follower}
            \FAA($\current \pointer \size$, 1)\tcp*[r]{increment the session \size{}}
            \label{line:trytoenter:size|increment}
            \uIf(\tcp*[f]{the session is still open}){\LNot{}(\IsClosed($\current \pointer \gate$))}
            {
               \label{line:trytoenter:session|still|open}
               \tcp{join the session as a follower and \retire[] own node}
               \ReclaimNode($\mynode$)\;
               \label{line:trytoenter:retire|follower}
               \Return\;
               \label{line:trytoenter:joined}
            } \Else(\tcp*[f]{the session is no longer open})
            {
               \FAA($\current \pointer \size$, -1)\tcp*[r]{abort the attempt and decrement the session \size{}}
               \label{line:trytoenter:size|decrement}
               \SetVacantFlag($\current$)\tcp*[r]{set \VACANTFLAG flag if applicable}
               \label{line:trytoenter:vacant|if}
            }
            
         }

      } \Else(\tcp*[f]{my request conflicts with the current session}) {
         
         \SetLeaderOrConflictFlag($\current$, \CONFLICTFLAG)\tcp*[r]{set \CONFLICTFLAG flag}
         \label{line:trytoenter:conflict}
         \SetVacantFlag($\current$)\tcp*[r]{set \VACANTFLAG flag if applicable}
         \label{line:trytoenter:vacant|else}
      }
      \BlankLine
	  \While(\tcp*[f]{spin}){\LNot{}(\IsAdjourned($\current \pointer \gate$))}           
      { 
		   \label{line:trytoenter:spin:adjourned|begin} 
		   \tcp*[r]{do nothing}
		 
	  }
	  \label{line:trytoenter:spin:adjourned|end} 	
      \BlankLine
	  \lIf(\tcp*[f]{establish a new session}){\TestHead($\myinstance$)}{\AppendNode($\myinstance$)}
	  \label{line:trytoenter:append}
	
   }
   \label{line:trytoenter:while|end}	
   \label{line:trytoenter|end}	
}
\BlankLine\BlankLine
\tcp{code for \exit section}
\TryToLeave(\Integer $\myinstance$) \\
\Begin{
   
   \ReadHead($\myinstance$)\tcp*[r]{find the head node of the list}
   \label{line:trytoleave:readhead}
   NodePtr $\current$ := $\snapshotArray[\myid]$\;
   \label{line:trytoleave|begin}
   \BlankLine
   \If(\tcp*[f]{joined the session as a leader}){($\current \pointer \owner = \myid$) }{ 
      \SetLeaderOrConflictFlag($\current$, \LEADERFLAG)\tcp*[r]{set the \LEADERFLAG flag}
      \label{line:trytoleave:leader}
   } 
   \BlankLine
   \FAA($\current \pointer \size$, -1)\tcp*[r]{decrement the session \size{}}
   \label{line:trytoleave:size|decrement}
   \SetVacantFlag($\current$)\tcp*[r]{set \VACANTFLAG flag if applicable}
   \label{line:trytoleave:vacant}
   \label{line:trytoleave|end}
}
\caption{\Entry and \exit sections.}
\label{algo:entry|exit}
\end{algorithm}


\begin{algorithm}[!t]
\tcp{get a new node, initialize it and announce it to other processes}
\InitializeNode(\Integer $\instance$, \Integer $\session$) \\
\Begin{
   
   \remove{
   NodePtr $\node$ := $\poolArray[\myid][\whichArray[\myid]][\markerArray[\myid]]$\tcp*[r]{get a \clean node from the \apool pool}
   }
   NodePtr $\node$ := get a new node\tcp*[r]{invoke dynamic memory manager}
   \label{line:initializenode|begin}	
   \label{line:initializenode:new}
   $\node \pointer \owner$ := $\myid$\tcp*[r]{set the owner as myself}
   \label{line:initializenode:owner}
   $\node \pointer \instance$ := $\instance$\tcp*[r]{initialize node's instance}
   \label{line:initializenode:instance}
   $\node \pointer \session$ := $\session$\tcp*[r]{initialize node's session}
   $\node \pointer \size$ := 1\tcp*[r]{initialize session \size{}}
   $\node \pointer \nextinlist$ := \Null\tcp*[r]{node has no successor}
   $\node \pointer \previousinlist$ := \Null\tcp*[r]{node has no predecessor}
   $\node \pointer \gate$ := 0\tcp*[r]{session is open with no \guard{} flag set}
   $\node \pointer \counter$ := 0\tcp*[r]{set the sequence number to a sentinel value}
   \label{line:initializenode:counter}
   $\announceArray[\myid]$ := $\node$\tcp*[r]{make the node visible to other processes}
   \label{line:initializenode:announce}
   \label{line:initializenode|end}	
}
\BlankLine\BlankLine
\tcp{get the next node to be appended to the list}
NodePtr \FetchNode(\Integer $\instance$) \\
\Begin{
   NodePtr $\mine$ := $\announceArray[\myid]$\tcp*[r]{my node}
   \label{line:fetchnode|begin}
   \label{line:fetchnode:mine}
   NodePtr $\helpee$ := $\announceArray[\snapshotArray[\myid] \pointer \counter]$\tcp*[r]{helpee's node}
   \label{line:fetchnode:helpee}
   \tcp{ascertain that the helpee's node is usable}
   \lIf(\tcp*[f]{no outstanding request}){($\helpee$ = \Null)}{\Return $\mine$}
	 \label{line:fetchnode:outstanding}	
   \lIf(\tcp*[f]{request is for a different GME object}){($\helpee \pointer \instance \neq \instance$)}{ \Return $\mine$}
   \label{line:fetchnode:same}
   \lIf(\tcp*[f]{node has been \retired}){\IsRetired($\helpee$)}{ \Return $\mine$}
   \label{line:fetchnode:retired} 
   \BlankLine
   \Return $\helpee$\tcp*[r]{helpee's node passed all the tests}
   \label{line:fetchnode:return|helpee}
   \label{line:fetchnode|end}
}
\BlankLine\BlankLine
\tcp{append a new node to the list}
\AppendNode(\Integer $\instance$) \\
\Begin{
   NodePtr $\current$ := $\snapshotArray[\myid]$\tcp*[r]{get the last node in the list}
   \label{line:appendnode|begin}
   \label{line:appendnode:snapshot}
   NodePtr $\successor$ := \FetchNode($\instance$)\tcp*[r]{choose a node to append}
   \label{line:appendnode:fetchnode}
   \BlankLine
   \CAS{}($\current \pointer \nextinlist$, \Null, $\successor$)\tcp*[r]{set the next field of the current last node}
   \label{line:appendnode:CAS}
   NodePtr $\successor$ := $\current \pointer \nextinlist$ \tcp*[r]{read the next field}
   \label{line:appendnode:successor}
   \lIf(\tcp*[f]{append operation already complete}){\LNot{}(\TestHead($\instance$))} { \Return}
   \label{line:appendnode:complete}
   $\successor \pointer \previousinlist$ := $\current$\tcp*[r]{set the previous field of the successor}
   \label{line:appendnode:previous}
   $\successor \pointer \counter$ := $(\current \pointer \counter + 1) \mod \n + 1$\tcp*[r]{set the sequence number used in helping}
   \label{line:appendnode:counter}
   \BlankLine
   \UpdateHead($\instance$,  $\successor$)\tcp*[r]{advance the head}
   \label{line:appendnode:advance|head}
   \label{line:appendnode|end}
}
\BlankLine\BlankLine
\tcp{retire the node}
\ReclaimNode(NodePtr $\node$) \\
\Begin{
   \label{line:reclaimnode|begin}
   $\announceArray[\myid]$ := \Null{}\tcp*[r]{help is no longer needed}
   \label{line:reclaimnode:revoke}
   \MarkAsRetired($\node$)\tcp*[r]{mark the node as retired}
   \label{line:reclaimnode:state}
   \label{line:reclaimnode|end}
}
\caption{Functions operating on a list node.}
\label{algo:methods|node}
\end{algorithm}


In this section, assume that nodes are \emph{never} reused. 
A pseudocode of our GME algorithm is given in \crefrange{algo:types|variables}{algo:methods|node}. In the pseudocode, $\myid$ refers to the identifier of the process (\emph{e.g.}, $\myid$ for process $p_i$ will evaluate to $i$).

\subsubsection{Data Structures Used}

\myparagraph{List node}
Central to our GME algorithm is a (list) node; it is used to maintain information about a session. As opposed to the linked list in the wait-free construction, which is a singly linked list, we maintain a doubly linked list. 
A node stores the following information (\cref{line:node}):  
\begin{enumerate*}[label=(\alph*)]
\item the session represented by the node, 
\item the instance identifier of the GME object to which the session belongs,
\item the state of the session,
\item the \size of the session, 
\item the address of the  previous and next nodes in the (doubly linked) list, and 
\item the owner of the node.
\end{enumerate*}

A session (or node) has four possible states: 
\begin{enumerate*}[label=(\roman*)]
\item \emph{\open:} it means that the session is currently in progress and new processes can join in,
\item \emph{\closed:} it means the session is currently in progress but no new processes can join in,
\item \emph{\adjourned:} it means that the session is no longer in progress and has no participating processes, and
\item \emph{\retired:} it means that  the node is no longer needed to either establish or maintain an already established session.
\end{enumerate*}
When a session is first established, it is in \open state. It stays  \open  as long as one of the following conditions still holds:
\begin{enumerate*}[label=(\arabic*)]
\item there is no conflicting request in the system, or
\item the request that established the session is still outstanding, \emph{i.e.}, executing its critical section.
\end{enumerate*}
Once both the conditions become false, the session moves to \closed state. Note that, in \closed state, the session may still have participants executing their critical sections. Once all such participants have left the session, the session moves to \adjourned state. Finally, the node associated with a request is \retired once either the session established by the node has \adjourned and a new session has been established or 
the node is no longer needed to establish a session.

We use four flags to represent session state: 
\begin{enumerate*}[label=(\arabic*)]
\item \LEADERFLAG flag to indicate that the session leader has left its critical section, 
\item \CONFLICTFLAG flag to indicate that some process has made a conflicting request, 
\item \VACANTFLAG flag to indicate that the session is empty or vacant, and 
\item \RETIREDFLAG flag to indicate that the node has been \retired.
\end{enumerate*}
For convenience, we refer to the first two flags as \emph{\guard} flags, the third flag as \emph{\vacant} flag and the fourth flag as \emph{\retired} flag.

The \vacant flag is set only after \emph{both} the \guard flags have been set.
Thus a session is \closed if both its \guard flags are set. It is \adjourned if its \vacant flag is also set. Finally, a node is considered retired if its \retired flag is set.
For convenience, when the \retired flag is set, we set the remaining three flags as well to simplify the algorithm. Thus if the \vacant flag is set, then both the \guard flags are also set; if the \retired flag is set, then the \vacant flag as well as both the \guard flags are also set. 
All the four flags are stored in a single word hence the value of session state can be easily read and updated atomically.

The \size of a session refers to the number of processes that have joined or trying to join the session, \emph{i.e.}, still executing their critical sections.

\myparagraph{Shared variables}
Each GME object has a separate linked list associated with it. Each list has a \emph{head}, which points to the last node in the list. Initially, the head of each list points to a ``dummy'' node representing an \adjourned session. For ease of exposition, we assume that pointers to all head nodes are stored in an array with one entry for each GME object, denoted by $\headArray$ (\cref{line:head|array}).

To enable helping, each process \emph{announces} its request by storing address of the node associated with its request in an array with one entry for each process, denoted by $\announceArray$ (\cref{line:announce|array}).

\remove{

To enable memory reclamation, each process maintains a small number of (specifically, two) hazard pointers in an array with one entry for each process, denoted by $\hpArray$; each entry is itself an array of size two (\cref{line:hp|array}).

}

\myparagraph{Private variables}
\remove{

In addition, each process uses private variables to maintain the following information:
\begin{enumerate*}[label=(\alph*)]
\item snapshot of the pointer to the head node of the list associated with the current request (\cref{line:snapshot|array}),
\item two disjoint pools of nodes (\cref{line:pool|array}),
\item which of the two pools is \apool, \emph{i.e.}, currently used to service requests (\cref{line:which|array}), and
\item the index of the first \clean node in the \apool pool (all \clean nodes are guaranteed to be toward the end of the pool) (\cref{line:marker|array}).
\end{enumerate*}
Note that private variables are modeled as arrays in our algorithm; the $i^{th}$ entry of each array is private to process $p_i$.

}

In addition, each process uses a private variable to maintain a snapshot of the pointer to the head node of the list associated with the current request (\cref{line:snapshot|array}).
Note that a private variable is modeled as an array in our algorithm; the $i^{th}$ entry of each array is private to process $p_i$.

\myparagraph{Managing session state}
\Cref{algo:methods|state} shows the pseudocode for accessing and manipulating session state. 
The methods for reading session state \IsClosed{} (\cref{line:isclosed}), \IsAdjourned{} (\cref{line:isadjourned}) and \IsRetired{} (\cref{line:isretired}) follow from the discussion earlier and are self-explanatory. 
The method \SetLeaderOrConflictFlag{} repeatedly attempts to set the given  \guard{} flag in the session state, if not already set, using a \CAS{} instruction until it succeeds (\crefrange{line:guard:while|begin}{line:guard:while|end}).
The method \SetVacantFlag{}  attempts to set the \vacant flag in the session state using a \CAS{} instruction provided the session has \closed and has no participants (\crefrange{line:vacant:state|read}{line:vacant:state|write}).
The method \MarkAsRetired{}  sets all the four flags in the session state  (\cref{line:mark:retired}).

The following lemma limits the number of times the loop in \SetLeaderOrConflictFlag{} method is executed:

\begin{lemma}
\label{lem:guard|flag:constant}
The while-do loop in \SetLeaderOrConflictFlag{} method (\crefrange{line:guard:while|begin}{line:guard:while|end}) is executed only $O(1)$ times per invocation of the method.
\end{lemma}
\begin{proof}
A new iteration of the while-do loop is executed only if the \CAS{} instruction performed on the session \state fails. The 
failure occurs only if one of the two \guard{} flags in the session state has been set by another \CAS{} instruction. This can only happen at most two times. 
\end{proof}

\begin{lemma}
\label{lem:adjourn:closed}
A session can \adjourn[] only after it has \closed.
\end{lemma}
\begin{proof}
For a session to \adjourn[], both the \guard flags (\LEADERFLAG and \CONFLICTFLAG) must be set in the session state. This implies that the session must be \closed before it can be \adjourned.
\end{proof}

\myparagraph{Managing list head}
\Cref{algo:methods|metadata} shows the pseudocode for accessing and manipulating list head.
\remove{

The method \ReadHead{} repeatedly reads the pointer to the current head of the list, declares it as a hazard pointer and then verifies that the head of the list is still the same until the verification succeeds (\crefrange{line:readhead:repeat|begin}{line:readhead:repeat|end}).

}
The method \ReadHead{} reads the pointer to the current head of the list and stores it in its private variable (\cref{line:readhead:read}).
The method \TestHead{} checks whether the head of the list is still the same since it was declared to be a hazard pointer (\crefrange{line:testhead:not|match}{line:testhead:match}). 
The method \UpdateHead{} advances the head of the list to its successor (\cref{line:updatehead:CAS}).

\subsubsection{Achieving Deadlock-Freedom}
\label{sec:deadlock|freedom}

\myparagraph{Entering critical section}
Whenever a process generates a critical section request, it obtains a new node  and initializes it appropriately (\crefrange{line:initializenode:instance}{line:initializenode:counter}). Specifically, all flags in the session \state are cleared, the number of processes in the session is set to \emph{one}, and the address of the previous and next nodes are set to \emph{null}. The process then repeatedly performs the following steps until it is able to enter its critical section (\crefrange{line:trytoenter:while|begin}{line:trytoenter:while|end} in \TryToEnter{}  method):  
\begin{enumerate}[label=(\arabic*)] 
\item It locates the current head of the linked list associated with the GME object (\cref{line:trytoenter:readhead}). 
\item If the head node (of the list) matches its own node (may happen because of helping described in \cref{sec:starvation|freedom}), 
it \retire[s] its predecessor node (\cref{line:trytoenter:retire|leader}) and enters its critical section (\cref{line:trytoenter:node|match}).  Otherwise, if
\begin{enumerate*}[label=(\roman*)]
\item the session is compatible with its own request, and
\item the session is \open (\cref{line:trytoenter:session|match,line:trytoenter:head:open}), 
\end{enumerate*}
it attempts to join the session by incrementing the session \size using an \FAA{} instruction (\cref{line:trytoenter:size|increment}). It then ascertains that the session is still in \open state (\cref{line:trytoenter:session|still|open}). If so, it \retire[s] its own node (\cref{line:trytoenter:retire|follower}) and enters its critical section (\cref{line:trytoenter:joined}). If not, it aborts the attempt, decrements the session \size using an \FAA{} instruction (\cref{line:trytoenter:size|decrement}) and attempts to \adjourn[] the session if possible (\cref{line:trytoenter:vacant|if}). Finally, if the session is not compatible with its own request, it sets the \CONFLICTFLAG{} flag in the session \state (\cref{line:trytoenter:conflict}) and attempts to \adjourn[] the session if applicable (\cref{line:trytoenter:vacant|else}).
\item If it is unable to join the session in the previous step for any reason (\emph{e.g.}, the session was not compatible with its own request or was not \open or was \closed before it could join), it busy waits for the session \state to change to \adjourned 
(\crefrange{line:trytoenter:spin:adjourned|begin}{line:trytoenter:spin:adjourned|end}). 
\item If the head of the list has not yet moved, then it attempts to establish a new session by appending a new node to the list (\cref{line:trytoenter:append}).

\item To append a new node to the list (\crefrange{line:appendnode|begin}{line:appendnode|end}), it first obtains a node to be used for appending (for now assume its own node) (\cref{line:appendnode:fetchnode}) and attempts to set the next pointer of the current head to that node using a \CAS{} instruction (\cref{line:appendnode:CAS}). Note that, irrespective of the outcome (of the \CAS{} instruction), a new node is guaranteed to be appended to the list. It then sets the previous pointer and the sequence number of the newly appended node (\cref{line:appendnode:previous,line:appendnode:counter}). Finally, it attempts to advance the head of the list to the newly appended node using a \CAS{} instruction (\cref{line:appendnode:advance|head}). 
\end{enumerate}

The following lemmas characterize the working of the \entry section:

\begin{lemma}
\label{lem:increment|open:joins}
A process starts executing its critical section as a follower only if the session it joins is compatible with its request and the session is still \open after it incremented the session \size.
\end{lemma}
\begin{proof}
After incrementing the session \size, a process joins the session  (and starts executing its critical section) only after ascertaining that 
the session is still \open.
\end{proof}

\begin{lemma}
\label{lem:not|adjourned:no|append}
No new node can be appended to a list until the session associated with the current head of the list has \adjourned.
\end{lemma}
\begin{proof}
Only a process that is unable to join a session tries to append a new node to the list, but only after it has detected that the session has \adjourned. 
\end{proof}

\myparagraph{Leaving critical section}
We say that a process enters its critical section as a \emph{leader} if its node is used to establish a new session. Otherwise, we say that it enters as a \emph{follower}.
On leaving the critical section, a process performs the following steps (\crefrange{line:trytoleave|begin}{line:trytoleave|end} in \TryToLeave{} method): 
\begin{enumerate}[label=(\arabic*)]
\item If it owns the head node, then it sets the \LEADERFLAG{} flag in the session \state{}. 
\item It then decrements the session \size using an \FAA{} instruction (\cref{line:trytoleave:size|decrement}).
\item It finally attempts to \adjourn[] the session if applicable (\cref{line:trytoleave:vacant}). 
\end{enumerate}

\begin{lemma}
\label{lem:joined|adjourn:exit}
If a process has successfully joined a session, then the session cannot \adjourn[] until after it starts executing its \exit{} section.
\end{lemma}
\begin{proof}
If a process enters its critical section as a leader, then the session \size is incremented even before its node is appended to the list. 
If a process enters its critical section as a follower, then, from \cref{lem:increment|open:joins}, the session was \open after the process incrementing the session \size. 

Clearly,  when the session \close[s], the value of the session \size is greater than or equal to the number of processes in the session that are executing their critical sections. And, no process sets the \vacant flag in the session \state until the session \size reaches zero. 
\end{proof}

The algorithm is not starvation free since there is no guarantee that a session compatible with the request of the process is ever established.

\subsubsection{Achieving Starvation-Freedom}
\label{sec:starvation|freedom}

To achieve starvation-freedom,  when selecting a node to append to the list, we use the \emph{helping} mechanism used in many wait-free algorithms. This requires making changes to  \InitializeNode{}, \FetchNode{}, \AppendNode{} and \ReclaimNode{} methods.

After obtaining a new node and initializing it  
(\crefrange{line:initializenode:new}{line:initializenode:counter} in \InitializeNode{} method), 
the process \emph{announces} its request to other processes by storing the node's address in a shared array, which has one entry for each process, denoted by $\announceArray$ (\cref{line:initializenode:announce}). 

When selecting a node to establish a new session (\FetchNode{} method), instead of always choosing its own node (\cref{line:fetchnode:mine}), it selects another process to help and chooses its node if 
the helpee process has an outstanding request (\cref{line:fetchnode:outstanding}) for the same GME object (\cref{line:fetchnode:same}) and the node has not been retired yet (\cref{line:fetchnode:retired}).

We use a simple \emph{round-robin} scheme to determine which process to help by storing a sequence number in every node. Every time a new node is appended to the list (\crefrange{line:appendnode:snapshot}{line:appendnode:advance|head} in \AppendNode{}), the sequence number of the (appended) node is set to one more than that of its predecessor using modulo $\n$ arithmetic (\cref{line:appendnode:counter}).

Finally, in \ReclaimNode{}, the process also revokes its announcement by clearing its entry in $\announceArray$ array (\cref{line:reclaimnode:revoke}).

\begin{lemma}
\label{lem:append:same|pending}
At the time a node is appended to the list, the request associated with the node is 
\begin{enumerate*}[label=(\alph*)]
\item for the GME object that owns the list and 
\item still outstanding. 
\end{enumerate*}
\end{lemma}

\begin{lemma}
\label{lem:atmost|n+1}
After a process has announced its request, at most $\n+1$ new sessions can be established until its request is fulfilled.
\end{lemma}
\begin{proof}
Every time a new node is appended to the list and its head pointer updated, the sequence number in the new head node of the list is incremented by one using module $\n$ arithmetic. Let the sequence number of the head node when a process, say $p_i$ with $i \in [1\ldots\n]$, announces its request be $x$. Among the next $\n$ values, given by $\{(x+1) \mod n + 1\}$, $\{(x+2) \mod n + 1\}$, $\ldots$, $\{(x+\n) \mod n + 1\}$, at least one value matches $i$. Clearly, when the sequence number of the head node reaches $i$, every process that tries to append a new node to the list chooses the node for $p_i$ as the one to append unless it is already \retired.
\end{proof}


Note that the algorithm is still space-inefficient since a new node is allocated for every request.


\subsubsection{Correctness Proof}
\label{sec:correctness|proof}

In this section, unless explicity mentioned, we focus on a single GME object. Our correctness proof easily carries over to multiple GME objects.
We first prove the group mutual exclusion property.

\begin{theorem}[group mutual exclusion]
The GME algorithm satisfies the group mutual exclusion property.
\end{theorem}
\begin{proof}
\Cref{lem:increment|open:joins} implies that only those processes whose request is compatible with the session can join the session and execute their critical sections within the session. 
\Cref{lem:joined|adjourn:exit} implies that, as long as a process is executing its critical section within a session, the session cannot \adjourn[].
Finally, \cref{lem:not|adjourned:no|append}  implies that no new session can be established until the current session has \adjourned. 
\end{proof}


We next prove the bounded exit property.

\begin{theorem}[bounded exit]
The GME algorithm satisfies the bounded exit property.
\end{theorem}
\begin{proof}
The body of \exit section (\TryToLeave{} method) includes up to one invocation of \ReadHead{}, \SetLeaderOrConflictFlag{} and \SetVacantFlag{} methods. Only the second method contains a loop; \cref{lem:guard|flag:constant} implies that the loop is only 
executed $O(1)$ times.
\end{proof}


We now prove the concurrent entering property. To that end, we start by establishing some properties of our GME algorithm.

\begin{lemma}
\label{lem:close:conflicting}
An \open session can \close[] only if the system contains a conflicting request.
\end{lemma}
\begin{proof}
For a session to \close[], \CONFLICTFLAG flag in the session \state must be set. The flag can only be set by a process whose request conflicts
with the current session. 
\end{proof}


As part of joining a session as a follower, a process first increments the session \size and then rechecks if the session is still \open. If not, it decrements the session \size immediately without executing its critical section. We refer to such an increment as \emph{spurious}. Note that spurious increments may prevent a session from moving to \adjourned state. 

\begin{lemma}
A process spuriously increments the \size of a session at most once. 
\end{lemma}
\begin{proof}
After performing a spurious increment followed by a matching decrement, a process busy waits until the session has \adjourned.
\end{proof}

\begin{lemma}
A process spuriously increments the \size of a session only if some other process in the system has a request that conflicts with its own request. 
\end{lemma}
\begin{proof}
Note that if the increment of session \size turns out to be spurious then it implies that the session \closed \emph{after} the increment step
but \emph{before} the decrement step. \Cref{lem:close:conflicting} implies that the system has a conflicting request at 
the point the session \closed.
\end{proof}

We consider an iteration of a while-do loop to start just after the boolean condition is evaluated and end just after the boolean condition is evaluated next or the loop is quit, whichever case applies.

\begin{lemma}
\label{lem:outer|while:new|session}
Consider an execution of one iteration of the outer while-do loop at \crefrange{line:trytoenter:while|begin}{line:trytoenter:while|end} in the \entry{} section. At the end of the iteration, either the process  joins the current session or a new session is established.
\end{lemma}

We say that a system state is \emph{\homogeneous} if no two requests, current or future, are for different sessions. 
Note that \homogeneity{} is a \emph{stable} property; once the system  enters a \homogeneous{} state, it stays in a \homogeneous{} state.

\begin{lemma}
\label{lem:homogeneous:atmost|one}
Once the system reaches a \homogeneous{} state, at most one new session can be established thereafter.
\end{lemma}

\begin{lemma}
\label{lem:homogeneous|outer:twice}
Assume that the system is in a \homogeneous{} state 
when a process starts executing an iteration of the  outer while-do loop 
at \crefrange{line:trytoenter:while|begin}{line:trytoenter:while|end}. Then the process executes the body of the while-do loop at most twice. 
\end{lemma}
\begin{proof}
Follows from \cref{lem:append:same|pending,lem:outer|while:new|session,lem:homogeneous:atmost|one}.
\end{proof}

\begin{lemma}
\label{lem:homogeneous|inner:once}
Assume that the system is in a \homogeneous{} state  at the beginning of an iteration 
of the inner while-do loop at \cref{line:trytoenter:spin:adjourned|begin}. Then the process executes the body of the while-do loop at 
most once. 
\end{lemma}
\begin{proof}
We first define some notation. Given a node $X$, let $\sessionOf{X}$ denote the session hosted by $X$.
Also, given a request $R$, let $\sessionOf{R}$ denote the session that $R$ wants to join.

Let $p$ denote the process executing the loop mentioned in the lemma statement, let $t$ denote the time at which it starts executing the current iteration and let $R$ denote the pending request of $p$ at $t$.
Also, let $H$ denote the head of the list when $p$ starts executing the iteration, and let $U$ denote the head of the list read by $p$ most recently (using the \ReadHead{} method).
Note that, by assumption, the system is in a \homogeneous{} state at $t$, which, in turn, implies that there is no pending 
request at $t$ that conflicts with $R$.
There are two cases to consider:
\begin{description}

\item[Case 1 ($H \neq U$):] This implies that $U$ is an \emph{old} head node of the list and $\sessionOf{U}$ is already \adjourned at $t$.

\item Case 2 ($H = U$): We claim that $\sessionOf{H} \neq \sessionOf{R}$. Otherwise, $\sessionOf{H}$ is already \closed at $t$. \Cref{lem:close:conflicting} implies that there exists a pending request at $t$ that conflicts with $R$---a contradiction. 

Now, let $q$ denote the last process to leave $\sessionOf{H}$. Note that both $p$ and $q$ invoke \SetVacantFlag{} method---$p$ after setting the \CONFLICTFLAG flag and $q$ after decrementing the \size field of $H$. Let $X \in \{P, Q\}$ denote the process that invoked the method \emph{later}. 
Note that, when $X$ invokes the \SetVacantFlag{} method, the following must hold:
\begin{enumerate*}[label=(\alph*)]
\item the \LEADERFLAG flag is already set, and
\item the \size field of $H$ is never incremented spuriously. 
\end{enumerate*}
The latter holds because, otherwise, it would imply that exists a pending request at $t$ that conflicts with $R$---a contradiction.
Clearly, when the \SetVacantFlag{} method invoked by $X$ returns, $\sessionOf{H}$ is guaranteed to be \adjourned{}. 

It now remains to argue that  the method returns before $t$. If $X = P$, then it follows trivially from the code. If $X = Q$, then, by definition,
the system cannot be in \homogeneous{} state until $Q$ has finished executing its \exit section.
\end{description}
In both cases, $\sessionOf{U}$ is guaranteed to be \adjourned at $t$ and thus $p$ quits the loop after finishing the current iteration.
\end{proof}

We are now ready to prove the concurrent entering property.

\begin{theorem}[concurrent entering]
The GME algorithm satisfies the concurrent entering property.
\end{theorem}
\begin{proof}
\Cref{lem:guard|flag:constant,lem:homogeneous|inner:once} imply that, once the system is in a \homogeneous{} state, a process finishes executing an iteration of the outer while-do loop of its \entry{} section within a bounded number of its own steps.  The property then follows from \cref{lem:homogeneous|outer:twice}.
\end{proof}


For the lockout freedom property, we need the following additional lemmas.

\begin{lemma}
\label{lem:conflicting:close}
If the system contains a conflicting request while a session is in progress, then the session eventually \close[s].
\end{lemma}
\begin{proof}
All processes with a conflicting request eventually invoke \SetLeaderOrConflictFlag{} method to set \CONFLICTFLAG flag in the session \state, which terminates only after the flag has been set. Further, when the leader of the session leaves its critical section, it invokes  \SetLeaderOrConflictFlag{} method to set \LEADERFLAG flag in the session \state, which terminates only after the flag has been set. 
\end{proof}

\begin{lemma}
\label{lem:close:spurious}
Once a session is \closed, its \size can be incremented spuriously at most $\n$ times.
\end{lemma}
\begin{proof}
Each process is responsible for at most one spurious increment to the session \size. 
\end{proof}

\begin{lemma}
\label{lem:close:zero}
Once a session is \closed, eventually the session \size becomes zero and stays zero thereafter.
\end{lemma}
\begin{proof}
After a session \close[s], no new process can join the session. Every process that is in the session at the point the session \close[s] eventually leaves the session. The result then follows from \cref{lem:close:spurious}.
\end{proof}

\begin{lemma}
\label{lem:close:adjourn}
A \closed session is eventually \adjourned.
\end{lemma}
\begin{proof}
Whenever a process either sets one of the \guard flags in the session \state or decrements the session \size, it attempts to set the \vacant flag afterward.  The result then follows from \cref{lem:guard|flag:constant,lem:close:zero}.
\end{proof}

\begin{lemma}
\label{lem:adjourned|new}
Once a session is \adjourned, a new session is eventually established.
\end{lemma}
\begin{proof}
A session \close[s] (and hence \adjourn[s]) only if there is a conflicting request in the system. Clearly, this implies that, after a session is \adjourned, at least one process in the system  tries to append a new node to the list (and establish a new session).
\end{proof}


Finally, we have

\begin{theorem}[lockout freedom]
The GME algorithm satisfies the lockout freedom property.
\end{theorem}
\begin{proof}
As long as a process has an outstanding request, \cref{lem:conflicting:close,lem:close:adjourn,lem:adjourned|new} imply that eventually 
either the process is able to join the session or the current session is \adjourned and a new session is established. 
The lockout freedom then follows from \cref{lem:atmost|n+1}.
\end{proof}

\subsubsection{Complexity Analysis}
\label{sec:complexity|analysis}

In this section, as in the previous section, unless explicity mentioned, we focus on a single GME object. Our complexity analysis easily carries over to multiple GME objects.

\begin{theorem}[worst case context switch complexity]
\label{thm:worst:csc}
The context switch complexity of a passage $\pi$ is at most $\min\{\ci(\pi),\n\}+1$ in the worst case.
\end{theorem}
\begin{proof}
\Cref{lem:atmost|n+1} implies that at most $\n+1$ new sessions can be established after a process has announced its request and before it is able to enter its critical section. Moreover, if a new session is established while a process is waiting to enter its critical section, then, clearly, the leader of that session has a request whose passage overlaps with that of the given process. 
\end{proof}

The main result in~\cite{GibGra:2015:DISC} implies that

\begin{theorem}[amortized case context switch complexity]
\label{thm:amortized:csc}
The context switch complexity of a passage $\pi$ is at most $\cp(\pi) + 1$ in the amortized case.
\end{theorem}

\begin{lemma}
\label{lem:entry:sessions}
Let $\s$ denotes the number of sessions that overlap with the \entry section of a process. Then the process  performs only $O(\s)$ remote references in its \entry section. 
\end{lemma}
\begin{proof}
\Cref{lem:outer|while:new|session} implies that a process performs at most $\s$ iterations of the outer while-do loop at \crefrange{line:trytoenter:while|begin}{line:trytoenter:while|end} in its \entry section.
In every iteration, a process performs at most $O(1)$ instructions outside of the inner while-do loop  at \crefrange{line:trytoenter:spin:adjourned|begin}{line:trytoenter:spin:adjourned|end}. While spinning in the inner-while loop, it reads the contents of the session \state (of the node pointed 
to by the head pointer) repeatedly. 
The session \state consists of four flags which, once set, are never reset (assuming no memory reclamation). Thus reading the session \state repeatedly in the loop is also responsible for only $O(1)$ remote references per list node. 
\end{proof}

\begin{lemma}
\label{lem:exit:constant}
A process performs only $O(1)$ remote references in its \exit section.
\end{lemma}
\begin{proof}
The only loop in the \exit section is in \SetLeaderOrConflictFlag{} method, which is invoked only once. The result then follows from \cref{lem:guard|flag:constant}.
\end{proof}

\begin{theorem}[RMR complexity]
The RMR  complexity of \entry and \exit sections of a passage $\pi$ is $O(\min\{\ci(\pi),\n\})$ in the worst-case and $O(\cp(\pi))$ in the amortized case.
\end{theorem}
\begin{proof}
The number of sessions that overlap with the \entry section of a process is upper-bounded by one plus the context-switch complexity of the corresponding passage.
The result then follows from \cref{thm:worst:csc,thm:amortized:csc,lem:entry:sessions,lem:exit:constant}.
\end{proof}

\begin{theorem}[concurrent entering step complexity]
The maximum number of steps a process has to execute in its \entry and \exit sections provided all current and future requests are for the same session is $O(1)$.
\end{theorem}
\begin{proof}
Assume that the system is in a \homogeneous{} state. We first analyze the \entry{} section of the process. 
\Cref{lem:homogeneous|inner:once} implies that, during one iteration of the outer while-do loop,  the process executes only $O(1)$ iterations of the  inner while-do loop.  Further, \cref{lem:homogeneous|outer:twice} implies that, the process executes only $O(1)$ iterations of the outer while-do loop. Thus,  the process executes
only $O(1)$ steps in its \entry{} section after the system has entered a \homogeneous{} state. 
Clearly, a process executes only $O(1)$ in its \exit{} section.
\end{proof}


\subsection{Achieving Space Efficiency}
\label{sec:space|efficient:algorithm}

To achieve space efficiency, we describe a way to reuse/recycle nodes across \emph{all} GME objects in a safe manner while adding only $O(1)$ steps to each passage. 

Consider a node that is used to establish a new session. Note that the session may stay ``active'' long after the owner of the node (also the leader of the session) has left its critical section. To address this issue, the leader of a session, on leaving its critical section, relinquishes the ownership of its node and instead claims the ownership of its predecessor node. 
This is similar to the approach used in the well-known queue-based mutual exclusion algorithm presented in~\cite{Cra:1993:TR,MagLan+:1994:IPPS}. 
On the other hand, if a process joins a session as a follower, it retains the ownership of its node. 
In both cases, the node is considered to be retired and is not used to establish a new session.

\myparagraph{Claiming the ownership of the predecessor node}
Claiming the ownership of the predecessor node is relatively straightforward in~\cite{Cra:1993:TR,MagLan+:1994:IPPS} because, unlike in our algorithm, only one process is holding a reference to the predecessor node when it is reclaimed. In our algorithm, multiple processes may be holding a reference to the predecessor node because of the helping mechanism used to achieve starvation-freedom. Note that a node may be appended to the list by any process in the system and not necessarily by its owner. 

Note that, when a node is appended to the list, we store a pointer at the node to its predecessor. When the owner of the appended node (also the leader of the session associated with the node) leaves its critical section, it can use this pointer to access the predecessor node and claim it as its own node. 

\myparagraph{Reusing a retired node}

To determine when it is safe to reuse a node, we use a variant of the well known memory reclamation technique based on \emph{hazard pointers} first presented in~\cite{Mic:2004:TPDS}. The technique works as follows. Each process maintains information about the set of objects it is dereferencing currently or will dereference in the future, and hence it is ``hazardous'' to reclaim their memory. 
A process can reuse or recycle an object only if no process has declared it as a hazard pointer. 
To declare a hazard pointer, a process performs the following sequence of steps repeatedly until it succeeds: it first reads the address of the node it wishes to dereference, it then writes the address to a shared location (visible to other processes) and finally ascertains that the reference is still 
``valid'' for the task to be performed and that the task has not completed yet. If the validation succeeds in the last step, then it can be shown that, even if the node is \retired later, it cannot be reused as long as it is declared as a hazard pointer.

The above algorithm increases the step complexity of an operation by $O(1)$ in the amortized case but $O(\n)$ in the worst-case (assuming that each process holds only $O(1)$ hazard pointers). One disadvantage of the algorithm is that, when used to manage memory in a concurrent algorithm with wait-free operations, it weakens the progress guarantee of an operation from wait-freedom to lock-freedom. Aghazadeh \emph{et al.} improve upon the  above memory reclamation algorithm in two ways in~\cite{AghGol+:2014:PODC}. First, their algorithm increases the step complexity of an operation by only $O(1)$ in the \emph{worst-case}. Second, it does not degrade the progress guarantee of the underlying concurrent algorithm; wait-free operations remain wait-free.  

In our case, the mechanism used in the previous section to achieve starvation-freedom is not impacted by the memory reclamation algorithm based on hazard pointers. Thus we only focus on the lock-free version of Aghazadeh \emph{et al.}'s algorithm that guarantees the first property only, which works as follows. Each process maintains
a pool of $\Theta(\n)$ objects; the ownership of an object is fixed and does not change at run time. To identify which objects in its pool can be reused, a process scans the hazard pointers of processes in a \emph{lazy} manner; specifically, during each operation, it scans the hazard pointer(s) of only one process. An object can be reused if the following two conditions hold:
\begin{enumerate*}[label=(\alph*)]
\item it was \retired before the last $\n$ operations (\emph{i.e.}, the hazard pointer of each process was scanned \emph{after}
the node was \retired), and 
\item no process was found to hold a reference to it in its list of hazard pointers during the last $\n$ operations.
\end{enumerate*}

Note that Aghazadeh \emph{et al.}'s algorithm cannot be directly used to manage memory of nodes in our case because, in our GME algorithm, the ownership of a node may change over time, and some nodes, namely the head nodes of lists, are not owned by any process but by their respective objects. We adapt the lock-free version of their algorithm to work in our case, while adding only $O(1)$ cost to the complexity of each passage, as follows.

\subsubsection{A Lazy Memory Reclamation Algorithm}


\begin{algorithm}[!tp]
\tcp{add the following field to node structure}
\Indp
\Integer $\condition$\tcp*[r]{the \condition of the node with respect to memory reclamation}
\Indm
\label{line:node:condition}
\BlankLine\BlankLine
\Additional \Shared \\
\Indp
$\hpArray$: \Array[$1\ldots\n$][$1{\ldots}2$] of NodePtr\tcp*[r]{to store hazard pointers}
\label{line:hp|array}
\Indm
\BlankLine\BlankLine
\Additional \Private \\
\Indp 
$\poolArray$: \Array[$1\ldots\n$][$1{\ldots}2$][$1{\ldots}3\n$] of NodePtr\tcp*[r]{to store pools of nodes}
\label{line:pool|array}
$\whichArray$: \Array[$1\ldots\n$] of \Integer{}\tcp*[r]{to indicate which of the two pools is \apool{}}
\label{line:which|array}
$\markerArray$ \Array[$1\ldots\n$] of \Integer{}\tcp*[r]{pointer to the first \clean node in the \apool pool}
\label{line:marker|array}
\Indm
\BlankLine\BlankLine
\Initialization of additional variables \\
\Begin{
	\tcp{initialize shared variables}
	\lForEach{$i \in [1,\n], j \in [1,2]$}
	{  $\hpArray[i][j]$ := \Null}
	\BlankLine
	\tcp{initialize private variables}
	\ForEach{$i \in [1,\n]$}
	{
	   \ForEach{$j \in [1,2], k \in [1,3\n]$}
	   {
	      $\poolArray[i][j][k]$ := \New Node\tcp*[r]{create a new node}
	      $\poolArray[i][j][k] \pointer \owner$ := $\myid$\tcp*[r]{set the owner as myself}
	   }
	   $\whichArray[i]$ := 1\tcp*[r]{designate $\poolArray[i][1]$ as \apool{}}
	   $\markerArray[i]$ := 1\tcp*[r]{designate $\poolArray[i][1][1]$ as the first \clean node}
	}
}
\BlankLine\BlankLine
\tcp{changes to \ReadHead{} method - replace \cref{line:readhead:read} with \crefrange{line:readhead:repeat|begin}{line:readhead:repeat|end}}
   \Repeat{($\snapshotArray[\myid] = \headArray[\instance]$)}{
      \label{line:readhead:repeat|begin}
      $\snapshotArray[\myid]$ := $\headArray[\instance]$\tcp*[r]{read the current head pointer of the list}
      $\hpArray[\myid][1]$ := $\snapshotArray[\myid]$\tcp*[r]{declare it as a hazard pointer}
      \label{line:readhead:declare|hp}
      
   }
   \label{line:readhead:repeat|end}
\BlankLine\BlankLine
\tcp{changes to \InitializeNode{} method - replace \crefrange{line:initializenode:new}{line:initializenode:owner} with \crefrange{line:initializenode:reuse}{line:initializenode:dirty}}
NodePtr $\node$ := $\poolArray[\myid][\whichArray[\myid]][\markerArray[\myid]]$\tcp*[r]{get a \clean node from the \apool pool}
\label{line:initializenode:reuse}
$\node \pointer \condition$ := \DIRTYFLAG{}\tcp*[r]{mark it as \dirty{}}
\label{line:initializenode:dirty}
\BlankLine\BlankLine
\tcp{changes to \FetchNode{} method - insert \crefrange{line:fetchnode:declare|hp}{line:fetchnode:validate|hp} after \cref{line:fetchnode:helpee}}
$\hpArray[\myid][2]$ := $\helpee$\tcp*[r]{declare reference to the helpee's node as a hazard pointer}
\label{line:fetchnode:declare|hp}
\lIf(\tcp*[f]{request already fulfilled}){($\announceArray[\snapshotArray[\myid] \pointer \counter] \neq \helpee$)}{\Return $\mine$}
\label{line:fetchnode:validate|hp}
\BlankLine\BlankLine
\tcp{changes to \AppendNode{} method - insert \cref{line:appendnode:declare|hp} after \cref{line:appendnode:successor}}
$\hpArray[\myid][2]$ := $\successor$\tcp*[r]{declare reference to the successor node as a hazard pointer}
\label{line:appendnode:declare|hp}
\BlankLine\BlankLine
\tcp{changes to \ReclaimNode{} method - insert \crefrange{line:reclaimnode:owner}{line:reclaimnode:advance} before \cref{line:reclaimnode:state}}
$\node \pointer \owner$ := $\myid$\tcp*[r]{claim the ownership of the node}
\label{line:reclaimnode:owner}
$\poolArray[\myid][\whichArray[\myid]][\markerArray[\myid]]$ := $\node$\tcp*[r]{replace in case reclaiming the predecessor node}
\label{line:reclaimnode:replace}
$\markerArray[\myid]$ := $\markerArray[\myid] + 1$\tcp*[r]{advance the pointer for the \clean nodes}
\label{line:reclaimnode:advance}
\BlankLine\BlankLine
\tcp{used to identify safe nodes in the \ppool pool; executing the method once corresponds to one epoch}
\Cleanup(~) \\
\Begin{
  \Integer $\other$ := $3 - \whichArray[\myid]$\;
  \label{line:cleanup|begin}
  \label{line:cleanup:passive}
	\ForEach(\tcp*[f]{mark the \condition of all the nodes in the \ppool{} pool as \unknown{}}){$i \in [1,3\n]$}
	{
	   \label{line:cleanup:clean:for|begin} 
	   $\poolArray[\myid][\other][i] \pointer \condition$ := \UNKNOWNFLAG{}\;
	
	}
	\label{line:cleanup:clean:for|end} 

	\ForEach(\tcp*[f]{scan all the hazard pointers}){$i \in [1,\n], j \in [1,2]$}
	{
	   \label{line:cleanup:dirty:for|begin} 
	    NodePtr $\node$ := $\hpArray[i][j]$\;
	    \label{line:cleanup:dirty:scan}
		\If(\tcp*[f]{node is in a \ppool{} pool}){($\node \pointer \condition$ = \UNKNOWNFLAG{})}{
		    \label{line:cleanup:unknown:if|begin} 
		    \If(\tcp*[f]{I own the node}){($\node \pointer \owner$ = $\myid$)}
		    {
		 	   \label{line:cleanup:own:if|begin}      
		 	   	\If(\tcp*[f]{node must be in my \ppool{} pool}){($\node \pointer \condition$ = \UNKNOWNFLAG{})}{
		 	   	    \label{line:cleanup:unknown|re:if|begin}  
		       		$\node \pointer \condition$ := \DIRTYFLAG{}\tcp*[r]{mark the node as \dirty{}}
		       		\label{line:cleanup:own:dirty}
		       		
		       	}
		        \label{line:cleanup:unknown|re:if|end}  
		    }
		    \label{line:cleanup:own:if|end}
		}
		\label{line:cleanup:unknown:if|end} 
	
	}
    \label{line:cleanup:dirty:for|end} 
  \BlankLine
  let $\safeSet$ denote the subset of all nodes in the \ppool pool whose \condition is set to \UNKNOWNFLAG{}\;
  \label{line:cleanup:let}
  collect all nodes in $\safeSet$ toward the end of the \ppool pool using a method similar to the partition procedure used in   
  quick sort, which has linear running time, and also change their \condition to \CLEANFLAG{}\;
  \label{line:cleanup:collect}
    \BlankLine
  \tcp{start a new epoch}
  $\markerArray[\myid]$ := index of the first \clean node in the \ppool pool\;
  \label{line:cleanup:marker}
  $\whichArray[\myid]$ := $3 - \whichArray[\myid]$\tcp*[r]{switch the designations of the pools}
  \label{line:cleanup:switch}

  \label{line:cleanup|end}
}
\caption{Reusing retired nodes.}
\label{algo:method|cleanup}
\end{algorithm}


\Cref{algo:method|cleanup} shows the changes/additions we made to the pseudocode in \crefrange{algo:types|variables}{algo:methods|node} to reclaim memory of \retired nodes and, thus, achieve space efficiency. 

We say that a \retired{} node has become \emph{\clean{}} if no process was found to hold a reference to it as a hazard pointer (and thus can be 
reused to establish a new session); otherwise, we say that it is \emph{\dirty{}}. A node now contains an additional field, namely $\condition$, to indicate the status of the node with respect to memory reclamation---\CLEANFLAG{}, \DIRTYFLAG{} or \UNKNOWNFLAG{}  (\cref{line:node:condition}).

To enable memory reclamation, each process maintains a small number of (specifically, two) hazard pointers in an array with one entry for each process, denoted by $\hpArray$ (\cref{line:hp|array}).
Hazard pointers of a process contain the addresses of the following nodes associated with its current request:
\begin{enumerate*}[label=(\roman*)]
\item the last known head of the list (\cref{line:readhead:declare|hp}), and
\item its successor---potential (\cref{line:fetchnode:declare|hp}) or actual (\cref{line:appendnode:declare|hp}).
\end{enumerate*} 
Each process also maintains the following private variables:
\begin{enumerate*}[label=(\alph*)]
\item two disjoint pools of nodes, each consisting of $3\n$ nodes (\cref{line:pool|array}),
\item which of the two pools is \apool, \emph{i.e.}, currently used to service requests (\cref{line:which|array}), and
\item the index of the first \clean node in the \apool pool (all nodes that are \clean{} to resue are guaranteed to be stored toward the end in the pool) (\cref{line:marker|array}).
\end{enumerate*}

The method \ReadHead{} now works as follows: it repeatedly reads the pointer to the current head of the list, declares it as a hazard pointer and then  validates the reference, until the validation succeeds (\crefrange{line:readhead:repeat|begin}{line:readhead:repeat|end}).

The method \FetchNode{} now includes statements to declare the reference to the helpee node as a hazard pointer, followed by its validation \cref{line:fetchnode:declare|hp,line:fetchnode:validate|hp}.

The method \AppendNode{} now includes a statement to declare the reference to the successor node as a hazard pointer \cref{line:appendnode:declare|hp}, which is validated at \cref{line:appendnode:complete}.

The execution of a process is divided into \emph{epochs}. Each epoch consists of exactly $\n$ passages. During an epoch, one of the pools is designated as \emph{\apool}, while the other is designated as \emph{\ppool}.
Intuitively, during an epoch, the \apool pool is used to service critical section requests (\cref{line:initializenode:reuse}), whereas the \ppool pool is processed incrementally (in lazy manner) to identify at least $\n$ \clean{} nodes to service requests in the next epoch 
(\crefrange{line:cleanup|begin}{line:cleanup|end}).
The designation is switched at beginning of each epoch. 
To identify the subset of nodes in its \ppool pool that are reusable, a process first sets the \condition field of all the nodes in the \ppool pool to \UNKNOWNFLAG{} (\crefrange{line:cleanup:passive}{line:cleanup:clean:for|end}).  It then scans the hazard pointers of all processes (\crefrange{line:cleanup:dirty:for|begin}{line:cleanup:dirty:scan}) and changes the \condition field of 
any node whose \condition field is currently set to \UNKNOWNFLAG, that is owned by it and reference to which has been declared as a hazard pointer to \DIRTYFLAG{} (\crefrange{line:cleanup:retired:if|begin}{line:cleanup:retired:if|end}). It next collects all nodes in the \ppool pool whose \condition
field  is still set to \UNKNOWNFLAG{} towards the end of the pool and also changes their \condition field to \CLEANFLAG{} (\crefrange{line:cleanup:let}{line:cleanup:collect}). Finally, it switches the designation of the two pools (\crefrange{line:cleanup:marker}{line:cleanup:switch}).

To complete the memory reclamation algorithm, we make changes to two more methods.
In the \InitializeNode{} method, a node is obtained from the \apool pool and its \condition field is set to \DIRTYFLAG{} (\crefrange{line:initializenode:new}{line:initializenode:owner}).
Finally, in the \ReclaimNode{} method, when a process releases the ownership of a node and acquires the ownership of another node (the predecessor of the current head) (\cref{line:reclaimnode:owner}), it replaces the former node with the latter node in its \apool pool.

Our memory reclamation algorithm satisfies the following properties. First, a \retired{} node is deemed to be safe to reuse only after none of the processes has declared it as a hazard pointer \emph{after} the node was \retired. Second, a node belongs to \emph{at most one} pool. Third, 
the \condition of a node is set to \UNKNOWNFLAG{} \emph{only if} the node belongs to a \ppool pool. Fourth, if the node belongs to a pool, then its \emph{current owner} information is available in the node's \owner{} field. Fourth,  

The first property helps to guarantee that, once a process has validated a reference after declaring it as a hazard pointer, the node associated with the reference cannot be reused as long as it is declared to be hazard pointer and, thus, the  starvation-free 
GME algorithm does not interfere with the memory reclamation algorithm.
The last three properties help to guarantee that a process modifies the \condition field of a node at \cref{line:cleanup:own:dirty} only if the node belongs to its own \ppool pool. 
The relevant section of the pseudocode is from \crefrange{line:cleanup:unknown:if|begin}{line:cleanup:unknown:if|end}. Consider a process $p$ executing the \Cleanup{} method as part of some epoch. The first if-statement checks that the \condition{} field of the node is set to \UNKNOWNFLAG{}. This implies that the node is in the \ppool pool of some process, say $q$. Note that $q$ may be different from $p$. The second if-statement checks that the \owner{} field of the node is set to $p$. But both if-statements may also evaluate to true if the node has migrated from the \ppool pool of $q$ to the \apool pool of $p$ since the first if-statement was evaluated (recall that the \Cleanup{} method is executed incrementally).  In this case, however, the \condition{} field of the node is guaranteed to be set to \DIRTYFLAG{} because the node will stay in the \apool{} pool until the end of the epoch. Thus, the third if-statement ensures that the node is indeed in the \ppool pool of $p$.


\myparagraph{Step complexity analysis}

Note that the \Cleanup{} method can be executed in $O(\n)$ steps because a pool contains $3\n$ nodes and each process only holds two hazard pointers. By setting the size of each pool to $3\n$, we can ascertain that, by the end of an epoch, a process is able to identify at least $\n$ reusable nodes in its \ppool pool. Clearly, a \ppool pool can be processed in an incremental manner such that only $O(1)$ steps are added to each passage of a process contained in its epoch.

The only change to a method that may increase the step complexity (asymptotically) is the one made to the \ReadHead{} method since it now contains a loop; all other changes only add $O(1)$ steps to their respective methods. 
We first bound the total number of times the loop in the \ReadHead{} method is executed over \emph{all} invocations in the \entry section of a process.

\begin{lemma}
\label{lem:readhead:append}
The number of times the repeat-until loop at \crefrange{line:readhead:repeat|begin}{line:readhead:repeat|end}  is executed  by a process 
in one invocation of the \ReadHead{} method  is bounded by one plus the number of nodes appended to the list \emph{during} the loop execution.
\end{lemma}
\begin{proof}
A new iteration of the repeat-until loop is executed only if the process finds that the pointer to the head node of the list has changed from what it read at the beginning of the iteration.
\end{proof}

\begin{lemma}
\label{lem:readhead|loop|entry:sessions}
Let $\s$ denote the total number of sessions that overlap with the \entry section of a process. Further, let $\ell_i$ denote
the number of iterations of the repeat-until loop at \crefrange{line:readhead:repeat|begin}{line:readhead:repeat|end} 
executed by a process in the $i^{th}$ invocation of the \ReadHead{} method in its \entry section. Then, we have
\[
\left(\sum_i \ell_i \right) \quad \leq \quad 2\s
\]
\end{lemma}
\begin{proof}
The result follows from \cref{lem:outer|while:new|session,lem:readhead:append}.
\end{proof}

\begin{lemma}
\label{lem:readhead|loop|exit:constant}
In any invocation of the \ReadHead{} method at \cref{line:trytoleave:readhead} in the \exit section of a process, the repeat-until loop at \crefrange{line:readhead:repeat|begin}{line:readhead:repeat|end} is executed only once.
\end{lemma}
\begin{proof}
Note that the head of the list cannot advance until after all processes that joined the session have also left the session.
\end{proof}

All lemma and theorem statements in \cref{sec:correctness|proof} and \cref{sec:complexity|analysis} still hold. The only proof that needs to be modified is for \cref{lem:entry:sessions}; it particular, it needs to incorporate the result of \cref{lem:readhead|loop|entry:sessions} (but the statement still holds).

\myparagraph{Space complexity analysis}
Finally, we analyze the space complexity of our GME algorithm considering that the system may contain multiple GME objects and a process may hold locks on multiple GME objects at the same time. Note that our GME algorithm still works without any modification even if a process needs to hold lock on multiple GME objects at the same time.

\begin{theorem}[multi-object space complexity]
The space complexity of our GME algorithm is $O(\m + \n^2 + \n \g)$ space, where $\n$ denotes the number of processes, $\m$ denotes the number of GME objects and $\g$ denotes the maximum number of locks a process needs to hold at the same time. 
\end{theorem}
\begin{proof}
Our algorithm uses only $O(\m + \n^2)$ space for managing $\m$ GME objects, where $O(\n^2)$ space is shared among all $\m$ GME objects. In addition, each process needs only $O(\g)$ space, where $\g$ denotes the maximum number of GME objects (or locks) a process needs to hold at the same time.  
\end{proof}

\section{Experimental Evaluation}
\label{sec:evaluation}

\pgfplotsset
{
	select coords between index/.style 2 args=
	{
    x filter/.code=
		{
        \ifnum\coordindex<#1\def\pgfmathresult{}\fi
        \ifnum\coordindex>#2\def\pgfmathresult{}\fi
    }
	}
}

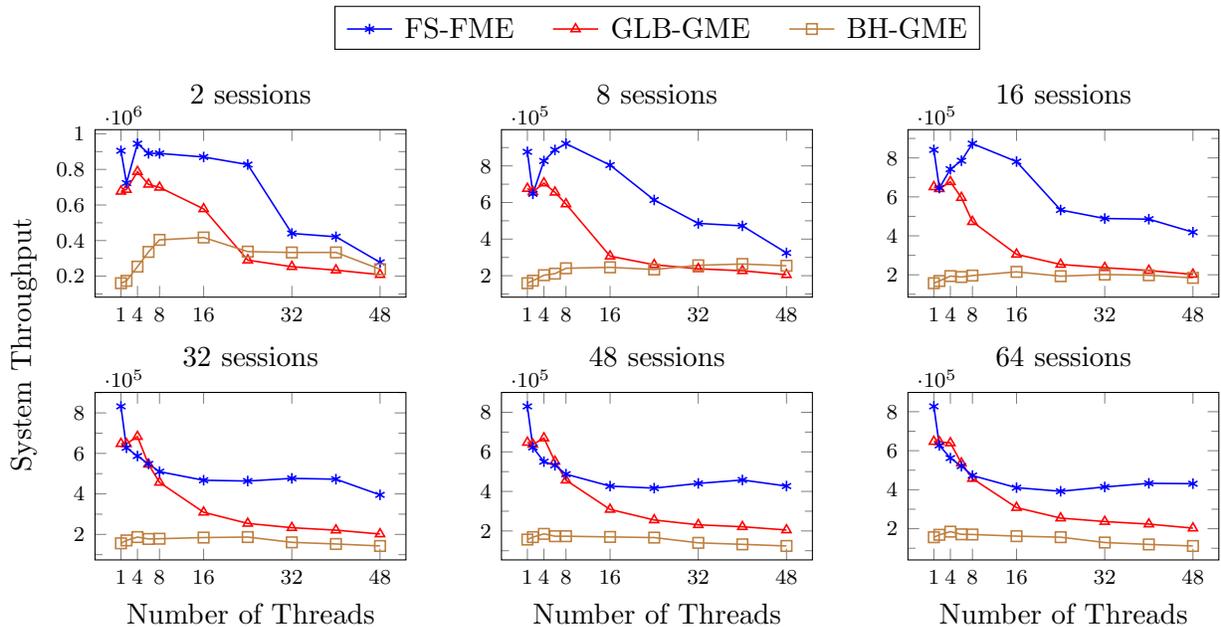
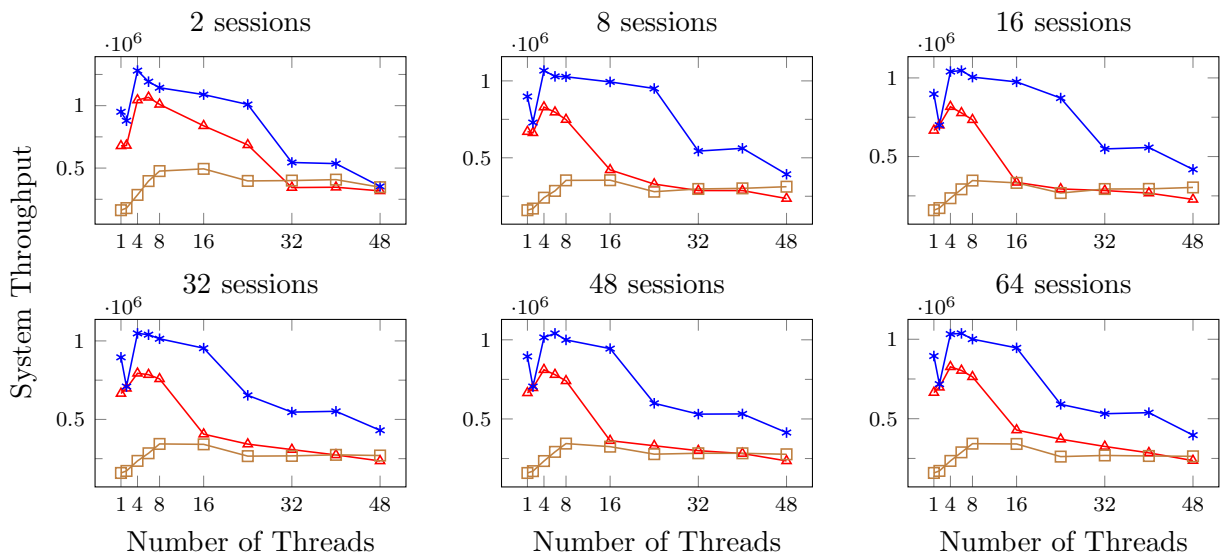
\begin{figure}[!t]
\centering
\begin{subfigure}[t]{\textwidth}
\begin{tikzpicture}
	\begin{groupplot}[group style={group size= 3 by 2, horizontal sep=0.5in, vertical sep=0.5in}, height=1.5in, width=2.25in, max space between ticks=20, minor tick num=1,tick label style={font=\scriptsize}]
		\nextgroupplot[title=2 sessions,xtick={1,4,8,16,32,48}]
				\addplot[red,semithick,mark=triangle] 	[select coords between index={0}{9}] table[x=ThreadCount,y=GLB-GME,col sep=space]{Data/lonestar/threadSweep-uniform-new.csv};  \label{plots:GLB-GME:1}
				\addplot[blue,semithick,mark=asterisk] 	[select coords between index={0}{9}] table[x=ThreadCount,y=FS-GME,col sep=space]{Data/lonestar/threadSweep-uniform-new.csv};  \label{plots:FS-GME:1}
				\addplot[brown, semithick,mark=square] 	[select coords between index={0}{9}] table[x=ThreadCount,y=BH-GME,col sep=space]{Data/lonestar/threadSweep-uniform-new.csv};  \label{plots:BH-GME:1}
				\coordinate (top) at (rel axis cs:0,1);
		\nextgroupplot[title=8 sessions,xtick={1,4,8,16,32,48}]
				\addplot[red,semithick,mark=triangle] 	[select coords between index={10}{19}] table[x=ThreadCount,y=GLB-GME,col sep=space]{Data/lonestar/threadSweep-uniform-new.csv};
				\addplot[blue,semithick,mark=asterisk] 	[select coords between index={10}{19}] table[x=ThreadCount,y=FS-GME,col sep=space]{Data/lonestar/threadSweep-uniform-new.csv};  
				\addplot[brown, semithick,mark=square] 	[select coords between index={10}{19}] table[x=ThreadCount,y=BH-GME,col sep=space]{Data/lonestar/threadSweep-uniform-new.csv};
		\nextgroupplot[title=16 sessions,xtick={1,4,8,16,32,48}]
				\addplot[red,semithick,mark=triangle] 	[select coords between index={20}{29}] table[x=ThreadCount,y=GLB-GME,col sep=space]{Data/lonestar/threadSweep-uniform-new.csv};
				\addplot[blue,semithick,mark=asterisk] 	[select coords between index={20}{29}] table[x=ThreadCount,y=FS-GME,col sep=space]{Data/lonestar/threadSweep-uniform-new.csv};  
				\addplot[brown, semithick,mark=square] 	[select coords between index={20}{29}] table[x=ThreadCount,y=BH-GME,col sep=space]{Data/lonestar/threadSweep-uniform-new.csv};
		\nextgroupplot[title=32 sessions,xlabel={Number of Threads}, xtick={1,4,8,16,32,48}]
				\addplot[red,semithick,mark=triangle] 	[select coords between index={30}{39}] table[x=ThreadCount,y=GLB-GME,col sep=space]{Data/lonestar/threadSweep-uniform-new.csv};
				\addplot[blue,semithick,mark=asterisk] 	[select coords between index={30}{39}] table[x=ThreadCount,y=FS-GME,col sep=space]{Data/lonestar/threadSweep-uniform-new.csv};  
				\addplot[brown, semithick,mark=square] 	[select coords between index={30}{39}] table[x=ThreadCount,y=BH-GME,col sep=space]{Data/lonestar/threadSweep-uniform-new.csv};
		\nextgroupplot[title=48 sessions,xlabel={Number of Threads}, xtick={1,4,8,16,32,48}]
				\addplot[red,semithick,mark=triangle] 	[select coords between index={40}{49}] table[x=ThreadCount,y=GLB-GME,col sep=space]{Data/lonestar/threadSweep-uniform-new.csv};
				\addplot[blue,semithick,mark=asterisk] 	[select coords between index={40}{49}] table[x=ThreadCount,y=FS-GME,col sep=space]{Data/lonestar/threadSweep-uniform-new.csv};  
				\addplot[brown, semithick,mark=square] 	[select coords between index={40}{49}] table[x=ThreadCount,y=BH-GME,col sep=space]{Data/lonestar/threadSweep-uniform-new.csv};
		\nextgroupplot[title=64 sessions,xlabel={Number of Threads},xtick={1,4,8,16,32,48}]
				\addplot[red,semithick,mark=triangle] 	[select coords between index={50}{59}] table[x=ThreadCount,y=GLB-GME,col sep=space]{Data/lonestar/threadSweep-uniform-new.csv};
				\addplot[blue,semithick,mark=asterisk] 	[select coords between index={50}{59}] table[x=ThreadCount,y=FS-GME,col sep=space]{Data/lonestar/threadSweep-uniform-new.csv};  
				\addplot[brown, semithick,mark=square] 	[select coords between index={50}{59}] table[x=ThreadCount,y=BH-GME,col sep=space]{Data/lonestar/threadSweep-uniform-new.csv};
				\coordinate (bot) at (rel axis cs:1,0);
	\end{groupplot}
	\path (top-|current bounding box.west)-- node[anchor=south,rotate=90] {System Throughput} (bot-|current bounding box.west);
	\path (top|-current bounding box.north)-- coordinate(legendpos) (bot|-current bounding box.north);
	\matrix[matrix of nodes, anchor=south, draw, inner sep=0.2em] at ([yshift=2ex]legendpos)
	{
		\ref{plots:FS-GME:1}   & \FSGME{} & [10pt]
		\ref{plots:GLB-GME:1}    & \GLBGME{}   & [10pt]
		\ref{plots:BH-GME:1}    & \BHGME{} \\
	};
\end{tikzpicture}
\caption{Uniform session distribution}
\end{subfigure}
\begin{subfigure}[t]{\textwidth}
\begin{tikzpicture}
	\begin{groupplot}[group style={group size= 3 by 2, horizontal sep=0.5in, vertical sep=0.5in},height=1.5in,width=2.25in, max space between ticks=20, minor tick num=1,tick label style={font=\scriptsize}]
		\nextgroupplot[title=2 sessions,xtick={1,4,8,16,32,48}]
				\addplot[red,semithick,mark=triangle] 	[select coords between index={0}{9}] table[x=ThreadCount,y=GLB-GME,col sep=space]{Data/lonestar/threadSweep-nonuniform-new.csv};  \label{plots:GLB-GME:2}
				\addplot[blue,semithick,mark=asterisk] 	[select coords between index={0}{9}] table[x=ThreadCount,y=FS-GME,col sep=space]{Data/lonestar/threadSweep-nonuniform-new.csv};  \label{plots:FS-GME:2}
				\addplot[brown, semithick,mark=square] 	[select coords between index={0}{9}] table[x=ThreadCount,y=BH-GME,col sep=space]{Data/lonestar/threadSweep-nonuniform-new.csv};  \label{plots:BH-GME:2}
				\coordinate (top) at (rel axis cs:0,1);
		\nextgroupplot[title=8 sessions,xtick={1,4,8,16,32,48}]
				\addplot[red,semithick,mark=triangle] 	[select coords between index={10}{19}] table[x=ThreadCount,y=GLB-GME,col sep=space]{Data/lonestar/threadSweep-nonuniform-new.csv};
				\addplot[blue,semithick,mark=asterisk] 	[select coords between index={10}{19}] table[x=ThreadCount,y=FS-GME,col sep=space]{Data/lonestar/threadSweep-nonuniform-new.csv};  
				\addplot[brown, semithick,mark=square] 	[select coords between index={10}{19}] table[x=ThreadCount,y=BH-GME,col sep=space]{Data/lonestar/threadSweep-nonuniform-new.csv};
		\nextgroupplot[title=16 sessions, xtick={1,4,8,16,32,48}]
				\addplot[red,semithick,mark=triangle] 	[select coords between index={20}{29}] table[x=ThreadCount,y=GLB-GME,col sep=space]{Data/lonestar/threadSweep-nonuniform-new.csv};
				\addplot[blue,semithick,mark=asterisk] 	[select coords between index={20}{29}] table[x=ThreadCount,y=FS-GME,col sep=space]{Data/lonestar/threadSweep-nonuniform-new.csv};  
				\addplot[brown, semithick,mark=square] 	[select coords between index={20}{29}] table[x=ThreadCount,y=BH-GME,col sep=space]{Data/lonestar/threadSweep-nonuniform-new.csv};
		\nextgroupplot[title=32 sessions,,xlabel={Number of Threads}, xtick={1,4,8,16,32,48}]
				\addplot[red,semithick,mark=triangle] 	[select coords between index={30}{39}] table[x=ThreadCount,y=GLB-GME,col sep=space]{Data/lonestar/threadSweep-nonuniform-new.csv};
				\addplot[blue,semithick,mark=asterisk] 	[select coords between index={30}{39}] table[x=ThreadCount,y=FS-GME,col sep=space]{Data/lonestar/threadSweep-nonuniform-new.csv};  
				\addplot[brown, semithick,mark=square] 	[select coords between index={30}{39}] table[x=ThreadCount,y=BH-GME,col sep=space]{Data/lonestar/threadSweep-nonuniform-new.csv};
		\nextgroupplot[title=48 sessions,xlabel={Number of Threads}, xtick={1,4,8,16,32,48}]
				\addplot[red,semithick,mark=triangle] 	[select coords between index={40}{49}] table[x=ThreadCount,y=GLB-GME,col sep=space]{Data/lonestar/threadSweep-nonuniform-new.csv};
				\addplot[blue,semithick,mark=asterisk] 	[select coords between index={40}{49}] table[x=ThreadCount,y=FS-GME,col sep=space]{Data/lonestar/threadSweep-nonuniform-new.csv};  
				\addplot[brown, semithick,mark=square] 	[select coords between index={40}{49}] table[x=ThreadCount,y=BH-GME,col sep=space]{Data/lonestar/threadSweep-nonuniform-new.csv};
		\nextgroupplot[title=64 sessions,xlabel={Number of Threads},xtick={1,4,8,16,32,48}]
				\addplot[red,semithick,mark=triangle] 	[select coords between index={50}{59}] table[x=ThreadCount,y=GLB-GME,col sep=space]{Data/lonestar/threadSweep-nonuniform-new.csv};
				\addplot[blue,semithick,mark=asterisk] 	[select coords between index={50}{59}] table[x=ThreadCount,y=FS-GME,col sep=space]{Data/lonestar/threadSweep-nonuniform-new.csv};  
				\addplot[brown, semithick,mark=square] 	[select coords between index={50}{59}] table[x=ThreadCount,y=BH-GME,col sep=space]{Data/lonestar/threadSweep-nonuniform-new.csv};
				\coordinate (bot) at (rel axis cs:1,0);
	\end{groupplot}
	\path (top-|current bounding box.west)-- node[anchor=south,rotate=90] {System Throughput} (bot-|current bounding box.west);
	\path (top|-current bounding box.north)-- coordinate(legendpos) (bot|-current bounding box.north);
	\end{tikzpicture}
\caption{Non-uniform session distribution}
\end{subfigure}
\caption{Comparison of system throughput of different algorithms. Higher the throughput, better the performance of the algorithm.}
\label{fig:lonestar:throughput}
\end{figure} 

\pgfplotsset
{
	select coords between index/.style 2 args=
	{
    x filter/.code=
		{
        \ifnum\coordindex<#1\def\pgfmathresult{}\fi
        \ifnum\coordindex>#2\def\pgfmathresult{}\fi
    }
	}
}

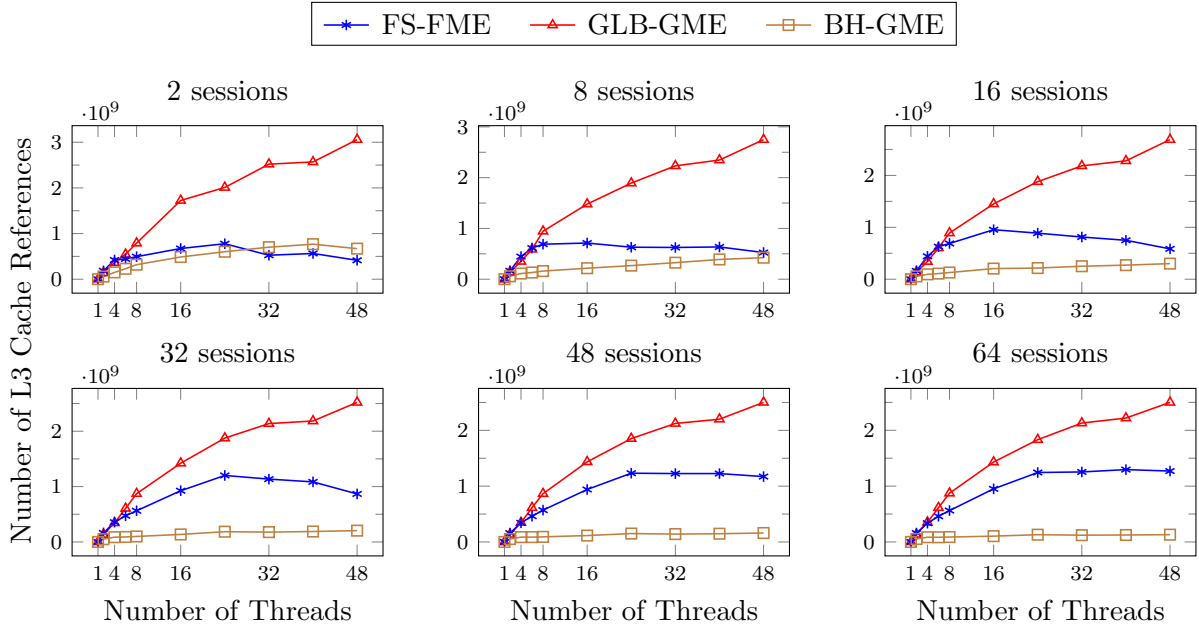
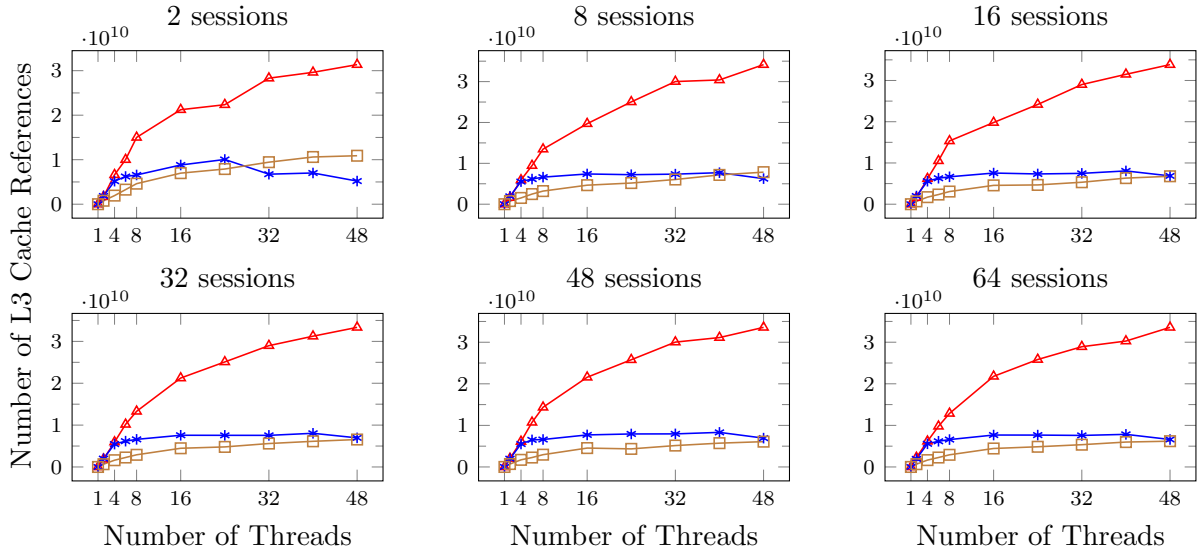
\begin{figure}[!t]
\centering
\begin{subfigure}{\textwidth}
\begin{tikzpicture}
	\begin{groupplot}[group style={group size= 3 by 2, horizontal sep=0.5in, vertical sep=0.5in},height=1.5in,width=2.25in, max space between ticks=20, minor tick num=1,tick label style={font=\scriptsize}]
		\nextgroupplot[title=2 sessions,xtick={1,4,8,16,32,48}]
				\addplot[red,semithick,mark=triangle] 	[select coords between index={0}{9}] table[x=ThreadCount,y=GLB-GME,col sep=space]{Data/lonestar/threadSweep-cache-ref.csv};  \label{plots:GLB-GME:5}
				\addplot[blue,semithick,mark=asterisk] 	[select coords between index={0}{9}] table[x=ThreadCount,y=FS-GME,col sep=space]{Data/lonestar/threadSweep-cache-ref.csv};  \label{plots:FS-GME:5}
				\addplot[brown, semithick,mark=square] 	[select coords between index={0}{9}] table[x=ThreadCount,y=BH-GME,col sep=space]{Data/lonestar/threadSweep-cache-ref.csv};  \label{plots:BH-GME:5}
				\coordinate (top) at (rel axis cs:0,1);
		\nextgroupplot[title=8 sessions,xtick={1,4,8,16,32,48}]
				\addplot[red,semithick,mark=triangle] 	[select coords between index={10}{19}] table[x=ThreadCount,y=GLB-GME,col sep=space]{Data/lonestar/threadSweep-cache-ref.csv};
				\addplot[blue,semithick,mark=asterisk] 	[select coords between index={10}{19}] table[x=ThreadCount,y=FS-GME,col sep=space]{Data/lonestar/threadSweep-cache-ref.csv};  
				\addplot[brown, semithick,mark=square] 	[select coords between index={10}{19}] table[x=ThreadCount,y=BH-GME,col sep=space]{Data/lonestar/threadSweep-cache-ref.csv};
		\nextgroupplot[title=16 sessions,xtick={1,4,8,16,32,48}]
				\addplot[red,semithick,mark=triangle] 	[select coords between index={20}{29}] table[x=ThreadCount,y=GLB-GME,col sep=space]{Data/lonestar/threadSweep-cache-ref.csv};
				\addplot[blue,semithick,mark=asterisk] 	[select coords between index={20}{29}] table[x=ThreadCount,y=FS-GME,col sep=space]{Data/lonestar/threadSweep-cache-ref.csv};  
				\addplot[brown, semithick,mark=square] 	[select coords between index={20}{29}] table[x=ThreadCount,y=BH-GME,col sep=space]{Data/lonestar/threadSweep-cache-ref.csv};
		\nextgroupplot[title=32 sessions,xlabel={Number of Threads},xtick={1,4,8,16,32,48}]
				\addplot[red,semithick,mark=triangle] 	[select coords between index={30}{39}] table[x=ThreadCount,y=GLB-GME,col sep=space]{Data/lonestar/threadSweep-cache-ref.csv};
				\addplot[blue,semithick,mark=asterisk] 	[select coords between index={30}{39}] table[x=ThreadCount,y=FS-GME,col sep=space]{Data/lonestar/threadSweep-cache-ref.csv};  
				\addplot[brown, semithick,mark=square] 	[select coords between index={30}{39}] table[x=ThreadCount,y=BH-GME,col sep=space]{Data/lonestar/threadSweep-cache-ref.csv};
		\nextgroupplot[title=48 sessions,xlabel={Number of Threads},xtick={1,4,8,16,32,48}]
				\addplot[red,semithick,mark=triangle] 	[select coords between index={40}{49}] table[x=ThreadCount,y=GLB-GME,col sep=space]{Data/lonestar/threadSweep-cache-ref.csv};
				\addplot[blue,semithick,mark=asterisk] 	[select coords between index={40}{49}] table[x=ThreadCount,y=FS-GME,col sep=space]{Data/lonestar/threadSweep-cache-ref.csv};  
				\addplot[brown, semithick,mark=square] 	[select coords between index={40}{49}] table[x=ThreadCount,y=BH-GME,col sep=space]{Data/lonestar/threadSweep-cache-ref.csv};
		\nextgroupplot[title=64 sessions,xlabel={Number of Threads},xtick={1,4,8,16,32,48}]
				\addplot[red,semithick,mark=triangle] 	[select coords between index={50}{59}] table[x=ThreadCount,y=GLB-GME,col sep=space]{Data/lonestar/threadSweep-cache-ref.csv};
				\addplot[blue,semithick,mark=asterisk] 	[select coords between index={50}{59}] table[x=ThreadCount,y=FS-GME,col sep=space]{Data/lonestar/threadSweep-cache-ref.csv};  
				\addplot[brown, semithick,mark=square] 	[select coords between index={50}{59}] table[x=ThreadCount,y=BH-GME,col sep=space]{Data/lonestar/threadSweep-cache-ref.csv};
				\coordinate (bot) at (rel axis cs:1,0);
	\end{groupplot}
	\path (top-|current bounding box.west)-- node[anchor=south,rotate=90] {Number of L3 Cache References} (bot-|current bounding box.west);
	\path (top|-current bounding box.north)-- coordinate(legendpos) (bot|-current bounding box.north);
	\matrix[matrix of nodes, anchor=south, draw, inner sep=0.2em] at ([yshift=2ex]legendpos)
    {
	   \ref{plots:FS-GME:5}   & \FSGME{} & [10pt]
	   \ref{plots:GLB-GME:5}    & \GLBGME{}   & [10pt]
	   \ref{plots:BH-GME:5}    & \BHGME{} \\
	};
\end{tikzpicture}
\caption{Uniform session distribution}
\end{subfigure}
\begin{subfigure}{\textwidth}
\begin{tikzpicture}
	\begin{groupplot}[group style={group size= 3 by 2, horizontal sep=0.5in, vertical sep=0.5in},height=1.5in,width=2.25in, max space between ticks=20, minor tick num=1,tick label style={font=\scriptsize}]
		\nextgroupplot[title=2 sessions,xtick={1,4,8,16,32,48}]
				\addplot[red,semithick,mark=triangle] 	[select coords between index={0}{9}] table[x=ThreadCount,y=GLB-GME,col sep=space]{Data/lonestar/threadSweep-cache-ref-nu.csv};  \label{plots:GLB-GME:6}
				\addplot[blue,semithick,mark=asterisk] 	[select coords between index={0}{9}] table[x=ThreadCount,y=FS-GME,col sep=space]{Data/lonestar/threadSweep-cache-ref-nu.csv};  \label{plots:FS-GME:6}
				\addplot[brown, semithick,mark=square] 	[select coords between index={0}{9}] table[x=ThreadCount,y=BH-GME,col sep=space]{Data/lonestar/threadSweep-cache-ref-nu.csv};  \label{plots:BH-GME:6}
				\coordinate (top) at (rel axis cs:0,1);
		\nextgroupplot[title=8 sessions,xtick={1,4,8,16,32,48}]
				\addplot[red,semithick,mark=triangle] 	[select coords between index={10}{19}] table[x=ThreadCount,y=GLB-GME,col sep=space]{Data/lonestar/threadSweep-cache-ref-nu.csv};
				\addplot[blue,semithick,mark=asterisk] 	[select coords between index={10}{19}] table[x=ThreadCount,y=FS-GME,col sep=space]{Data/lonestar/threadSweep-cache-ref-nu.csv};  
				\addplot[brown, semithick,mark=square] 	[select coords between index={10}{19}] table[x=ThreadCount,y=BH-GME,col sep=space]{Data/lonestar/threadSweep-cache-ref-nu.csv};
		\nextgroupplot[title=16 sessions,xtick={1,4,8,16,32,48}]
				\addplot[red,semithick,mark=triangle] 	[select coords between index={20}{29}] table[x=ThreadCount,y=GLB-GME,col sep=space]{Data/lonestar/threadSweep-cache-ref-nu.csv};
				\addplot[blue,semithick,mark=asterisk] 	[select coords between index={20}{29}] table[x=ThreadCount,y=FS-GME,col sep=space]{Data/lonestar/threadSweep-cache-ref-nu.csv};  
				\addplot[brown, semithick,mark=square] 	[select coords between index={20}{29}] table[x=ThreadCount,y=BH-GME,col sep=space]{Data/lonestar/threadSweep-cache-ref-nu.csv};
		\nextgroupplot[title=32 sessions,xlabel={Number of Threads},xtick={1,4,8,16,32,48}]
				\addplot[red,semithick,mark=triangle] 	[select coords between index={30}{39}] table[x=ThreadCount,y=GLB-GME,col sep=space]{Data/lonestar/threadSweep-cache-ref-nu.csv};
				\addplot[blue,semithick,mark=asterisk] 	[select coords between index={30}{39}] table[x=ThreadCount,y=FS-GME,col sep=space]{Data/lonestar/threadSweep-cache-ref-nu.csv};  
				\addplot[brown, semithick,mark=square] 	[select coords between index={30}{39}] table[x=ThreadCount,y=BH-GME,col sep=space]{Data/lonestar/threadSweep-cache-ref-nu.csv};
		\nextgroupplot[title=48 sessions,xlabel={Number of Threads},xtick={1,4,8,16,32,48}]
				\addplot[red,semithick,mark=triangle] 	[select coords between index={40}{49}] table[x=ThreadCount,y=GLB-GME,col sep=space]{Data/lonestar/threadSweep-cache-ref-nu.csv};
				\addplot[blue,semithick,mark=asterisk] 	[select coords between index={40}{49}] table[x=ThreadCount,y=FS-GME,col sep=space]{Data/lonestar/threadSweep-cache-ref-nu.csv};  
				\addplot[brown, semithick,mark=square] 	[select coords between index={40}{49}] table[x=ThreadCount,y=BH-GME,col sep=space]{Data/lonestar/threadSweep-cache-ref-nu.csv};
		\nextgroupplot[title=64 sessions,xlabel={Number of Threads},xtick={1,4,8,16,32,48}]
				\addplot[red,semithick,mark=triangle] 	[select coords between index={50}{59}] table[x=ThreadCount,y=GLB-GME,col sep=space]{Data/lonestar/threadSweep-cache-ref-nu.csv};
				\addplot[blue,semithick,mark=asterisk] 	[select coords between index={50}{59}] table[x=ThreadCount,y=FS-GME,col sep=space]{Data/lonestar/threadSweep-cache-ref-nu.csv};  
				\addplot[brown, semithick,mark=square] 	[select coords between index={50}{59}] table[x=ThreadCount,y=BH-GME,col sep=space]{Data/lonestar/threadSweep-cache-ref-nu.csv};
				\coordinate (bot) at (rel axis cs:1,0);
	\end{groupplot}
	\path (top-|current bounding box.west)-- node[anchor=south,rotate=90] {Number of L3 Cache References} (bot-|current bounding box.west);
	\path (top|-current bounding box.north)-- coordinate(legendpos) (bot|-current bounding box.north);
\end{tikzpicture}
\caption{Non-uniform session distribution}
\end{subfigure}
\caption{Comparison of number of L3 cache references of different algorithms.}
\label{fig:lonestar:l3|references}
\end{figure} 

\pgfplotsset
{
	select coords between index/.style 2 args=
	{
    x filter/.code=
		{
        \ifnum\coordindex<#1\def\pgfmathresult{}\fi
        \ifnum\coordindex>#2\def\pgfmathresult{}\fi
    }
	}
}


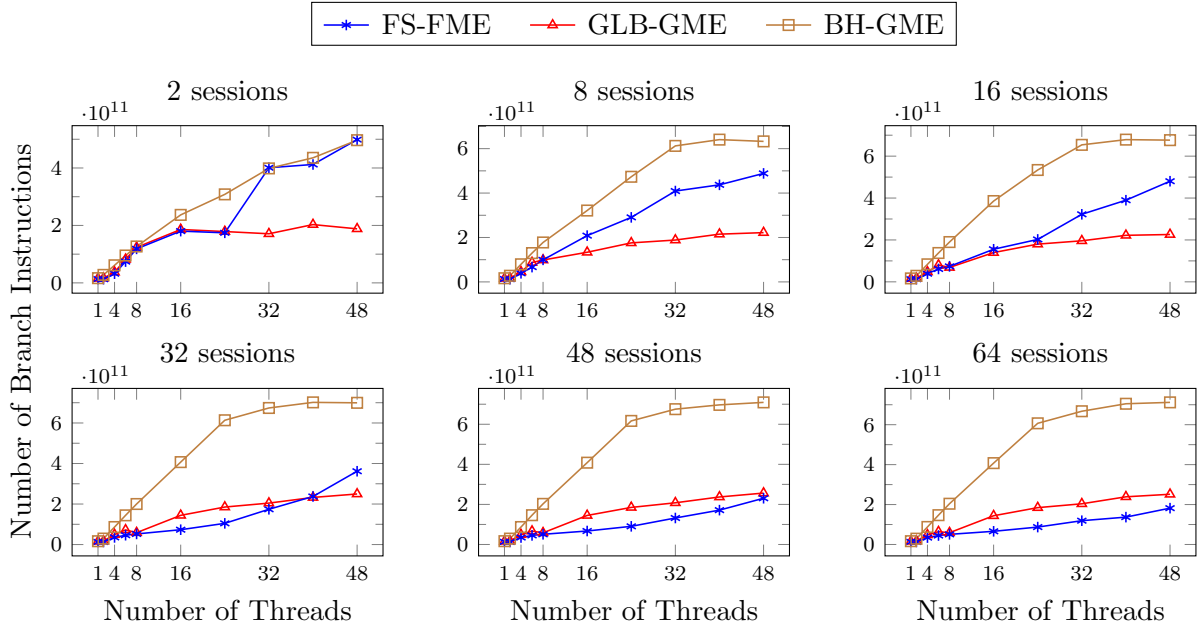
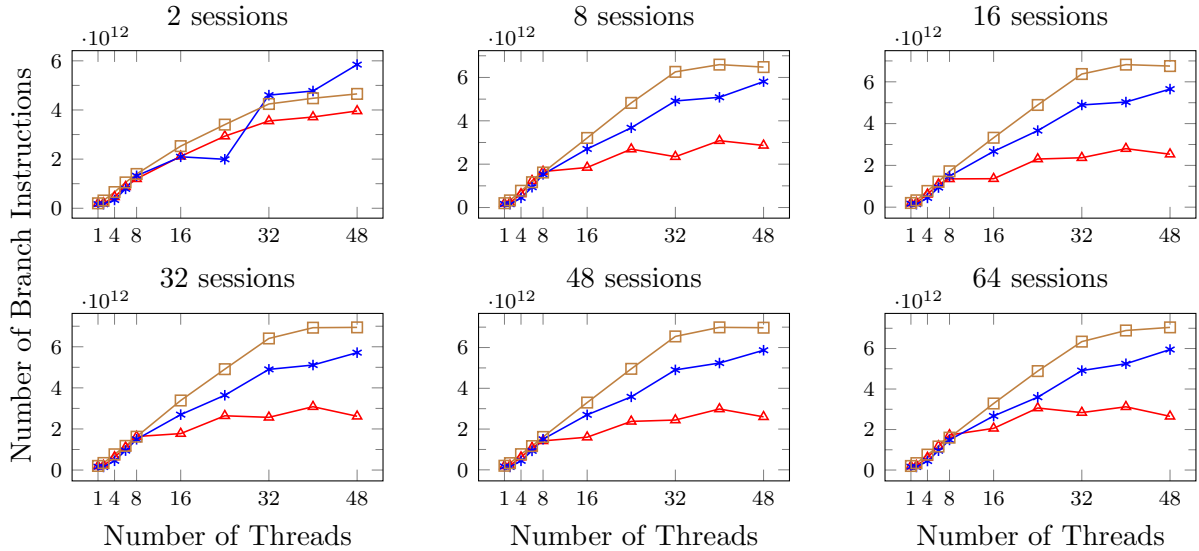
\begin{figure}[!t]
\centering
\begin{subfigure}{\textwidth}
\begin{tikzpicture}
	\begin{groupplot}[group style={group size= 3 by 2, horizontal sep=0.5in, vertical sep=0.5in},height=1.5in,width=2.25in, max space between ticks=20, minor tick num=1,tick label style={font=\scriptsize}]
		\nextgroupplot[title=2 sessions,xtick={1,4,8,16,32,48}]
				\addplot[red,semithick,mark=triangle] 	[select coords between index={0}{9}] table[x=ThreadCount,y=GLB-GME,col sep=space]{Data/lonestar/threadSweep-branches.csv};  \label{plots:GLB-GME:3}
				\addplot[blue,semithick,mark=asterisk] 	[select coords between index={0}{9}] table[x=ThreadCount,y=FS-GME,col sep=space]{Data/lonestar/threadSweep-branches.csv};  \label{plots:FS-GME:3}
				\addplot[brown, semithick,mark=square] 	[select coords between index={0}{9}] table[x=ThreadCount,y=BH-GME,col sep=space]{Data/lonestar/threadSweep-branches.csv};  \label{plots:BH-GME:3}
				\coordinate (top) at (rel axis cs:0,1);
		\nextgroupplot[title=8 sessions,xtick={1,4,8,16,32,48}]
				\addplot[red,semithick,mark=triangle] 	[select coords between index={10}{19}] table[x=ThreadCount,y=GLB-GME,col sep=space]{Data/lonestar/threadSweep-branches.csv};
				\addplot[blue,semithick,mark=asterisk] 	[select coords between index={10}{19}] table[x=ThreadCount,y=FS-GME,col sep=space]{Data/lonestar/threadSweep-branches.csv};  
				\addplot[brown, semithick,mark=square] 	[select coords between index={10}{19}] table[x=ThreadCount,y=BH-GME,col sep=space]{Data/lonestar/threadSweep-branches.csv};
		\nextgroupplot[title=16 sessions,xtick={1,4,8,16,32,48}]
				\addplot[red,semithick,mark=triangle] 	[select coords between index={20}{29}] table[x=ThreadCount,y=GLB-GME,col sep=space]{Data/lonestar/threadSweep-branches.csv};
				\addplot[blue,semithick,mark=asterisk] 	[select coords between index={20}{29}] table[x=ThreadCount,y=FS-GME,col sep=space]{Data/lonestar/threadSweep-branches.csv};  
				\addplot[brown, semithick,mark=square] 	[select coords between index={20}{29}] table[x=ThreadCount,y=BH-GME,col sep=space]{Data/lonestar/threadSweep-branches.csv};
		\nextgroupplot[title=32 sessions,xlabel={Number of Threads}, xtick={1,4,8,16,32,48}]
				\addplot[red,semithick,mark=triangle] 	[select coords between index={30}{39}] table[x=ThreadCount,y=GLB-GME,col sep=space]{Data/lonestar/threadSweep-branches.csv};
				\addplot[blue,semithick,mark=asterisk] 	[select coords between index={30}{39}] table[x=ThreadCount,y=FS-GME,col sep=space]{Data/lonestar/threadSweep-branches.csv};  
				\addplot[brown, semithick,mark=square] 	[select coords between index={30}{39}] table[x=ThreadCount,y=BH-GME,col sep=space]{Data/lonestar/threadSweep-branches.csv};
		\nextgroupplot[title=48 sessions,xlabel={Number of Threads}, xtick={1,4,8,16,32,48}]
				\addplot[red,semithick,mark=triangle] 	[select coords between index={40}{49}] table[x=ThreadCount,y=GLB-GME,col sep=space]{Data/lonestar/threadSweep-branches.csv};
				\addplot[blue,semithick,mark=asterisk] 	[select coords between index={40}{49}] table[x=ThreadCount,y=FS-GME,col sep=space]{Data/lonestar/threadSweep-branches.csv};  
				\addplot[brown, semithick,mark=square] 	[select coords between index={40}{49}] table[x=ThreadCount,y=BH-GME,col sep=space]{Data/lonestar/threadSweep-branches.csv};
		\nextgroupplot[title=64 sessions,xlabel={Number of Threads}, xtick={1,4,8,16,32,48}]
				\addplot[red,semithick,mark=triangle] 	[select coords between index={50}{59}] table[x=ThreadCount,y=GLB-GME,col sep=space]{Data/lonestar/threadSweep-branches.csv};
				\addplot[blue,semithick,mark=asterisk] 	[select coords between index={50}{59}] table[x=ThreadCount,y=FS-GME,col sep=space]{Data/lonestar/threadSweep-branches.csv};  
				\addplot[brown, semithick,mark=square] 	[select coords between index={50}{59}] table[x=ThreadCount,y=BH-GME,col sep=space]{Data/lonestar/threadSweep-branches.csv};
				\coordinate (bot) at (rel axis cs:1,0);
	\end{groupplot}
	\path (top-|current bounding box.west)-- node[anchor=south,rotate=90] {Number of Branch Instructions} (bot-|current bounding box.west);
	\path (top|-current bounding box.north)-- coordinate(legendpos) (bot|-current bounding box.north);
	\matrix[matrix of nodes, anchor=south, draw, inner sep=0.2em] at ([yshift=2ex]legendpos)
    {
	   \ref{plots:FS-GME:3}   & \FSGME{} & [10pt]
	   \ref{plots:GLB-GME:3}    & \GLBGME{}   & [10pt]
	   \ref{plots:BH-GME:3}    & \BHGME{} \\
	};

\end{tikzpicture}
\caption{Uniform session distribution}
\end{subfigure}
\begin{subfigure}{\textwidth}
\begin{tikzpicture}
	\begin{groupplot}[group style={group size= 3 by 2, horizontal sep=0.5in, vertical sep=0.5in},height=1.5in,width=2.25in, max space between ticks=20, minor tick num=1,tick label style={font=\scriptsize}]
		\nextgroupplot[title=2 sessions,xtick={1,4,8,16,32,48}]
				\addplot[red,semithick,mark=triangle] 	[select coords between index={0}{9}] table[x=ThreadCount,y=GLB-GME,col sep=space]{Data/lonestar/threadSweep-branches-nu.csv};  \label{plots:GLB-GME:4}
				\addplot[blue,semithick,mark=asterisk] 	[select coords between index={0}{9}] table[x=ThreadCount,y=FS-GME,col sep=space]{Data/lonestar/threadSweep-branches-nu.csv};  \label{plots:FS-GME:4}
				\addplot[brown, semithick,mark=square] 	[select coords between index={0}{9}] table[x=ThreadCount,y=BH-GME,col sep=space]{Data/lonestar/threadSweep-branches-nu.csv};  \label{plots:BH-GME:4}
				\coordinate (top) at (rel axis cs:0,1);
		\nextgroupplot[title=8 sessions,xtick={1,4,8,16,32,48}]
				\addplot[red,semithick,mark=triangle] 	[select coords between index={10}{19}] table[x=ThreadCount,y=GLB-GME,col sep=space]{Data/lonestar/threadSweep-branches-nu.csv};
				\addplot[blue,semithick,mark=asterisk] 	[select coords between index={10}{19}] table[x=ThreadCount,y=FS-GME,col sep=space]{Data/lonestar/threadSweep-branches-nu.csv};  
				\addplot[brown, semithick,mark=square] 	[select coords between index={10}{19}] table[x=ThreadCount,y=BH-GME,col sep=space]{Data/lonestar/threadSweep-branches-nu.csv};
		\nextgroupplot[title=16 sessions,xtick={1,4,8,16,32,48}]
				\addplot[red,semithick,mark=triangle] 	[select coords between index={20}{29}] table[x=ThreadCount,y=GLB-GME,col sep=space]{Data/lonestar/threadSweep-branches-nu.csv};
				\addplot[blue,semithick,mark=asterisk] 	[select coords between index={20}{29}] table[x=ThreadCount,y=FS-GME,col sep=space]{Data/lonestar/threadSweep-branches-nu.csv};  
				\addplot[brown, semithick,mark=square] 	[select coords between index={20}{29}] table[x=ThreadCount,y=BH-GME,col sep=space]{Data/lonestar/threadSweep-branches-nu.csv};
		\nextgroupplot[title=32 sessions,xlabel={Number of Threads},xtick={1,4,8,16,32,48}]
				\addplot[red,semithick,mark=triangle] 	[select coords between index={30}{39}] table[x=ThreadCount,y=GLB-GME,col sep=space]{Data/lonestar/threadSweep-branches-nu.csv};
				\addplot[blue,semithick,mark=asterisk] 	[select coords between index={30}{39}] table[x=ThreadCount,y=FS-GME,col sep=space]{Data/lonestar/threadSweep-branches-nu.csv};  
				\addplot[brown, semithick,mark=square] 	[select coords between index={30}{39}] table[x=ThreadCount,y=BH-GME,col sep=space]{Data/lonestar/threadSweep-branches-nu.csv};
		\nextgroupplot[title=48 sessions,xlabel={Number of Threads},xtick={1,4,8,16,32,48}]
				\addplot[red,semithick,mark=triangle] 	[select coords between index={40}{49}] table[x=ThreadCount,y=GLB-GME,col sep=space]{Data/lonestar/threadSweep-branches-nu.csv};
				\addplot[blue,semithick,mark=asterisk] 	[select coords between index={40}{49}] table[x=ThreadCount,y=FS-GME,col sep=space]{Data/lonestar/threadSweep-branches-nu.csv};  
				\addplot[brown, semithick,mark=square] 	[select coords between index={40}{49}] table[x=ThreadCount,y=BH-GME,col sep=space]{Data/lonestar/threadSweep-branches-nu.csv};
		\nextgroupplot[title=64 sessions,xlabel={Number of Threads},xtick={1,4,8,16,32,48}]
				\addplot[red,semithick,mark=triangle] 	[select coords between index={50}{59}] table[x=ThreadCount,y=GLB-GME,col sep=space]{Data/lonestar/threadSweep-branches-nu.csv};
				\addplot[blue,semithick,mark=asterisk] 	[select coords between index={50}{59}] table[x=ThreadCount,y=FS-GME,col sep=space]{Data/lonestar/threadSweep-branches-nu.csv};  
				\addplot[brown, semithick,mark=square] 	[select coords between index={50}{59}] table[x=ThreadCount,y=BH-GME,col sep=space]{Data/lonestar/threadSweep-branches-nu.csv};
				\coordinate (bot) at (rel axis cs:1,0);
	\end{groupplot}
	\path (top-|current bounding box.west)-- node[anchor=south,rotate=90] {Number of Branch Instructions} (bot-|current bounding box.west);
	\path (top|-current bounding box.north)-- coordinate(legendpos) (bot|-current bounding box.north);
\end{tikzpicture}
\caption{Non-uniform session distribution}
\end{subfigure}
\caption{Comparison of number of branch instructions of different algorithms.}
\label{fig:lonestar:branches}
\end{figure} 

\pgfplotsset
{
	select coords between index/.style 2 args=
	{
    x filter/.code=
		{
        \ifnum\coordindex<#1\def\pgfmathresult{}\fi
        \ifnum\coordindex>#2\def\pgfmathresult{}\fi
    }
	}
}

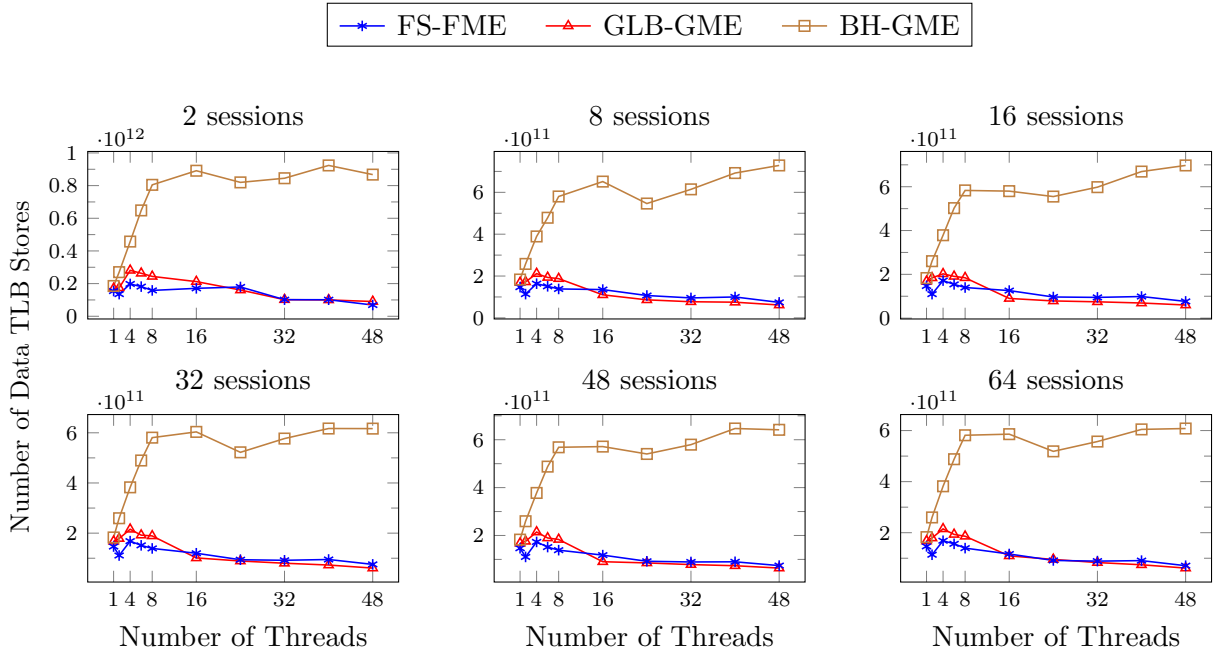
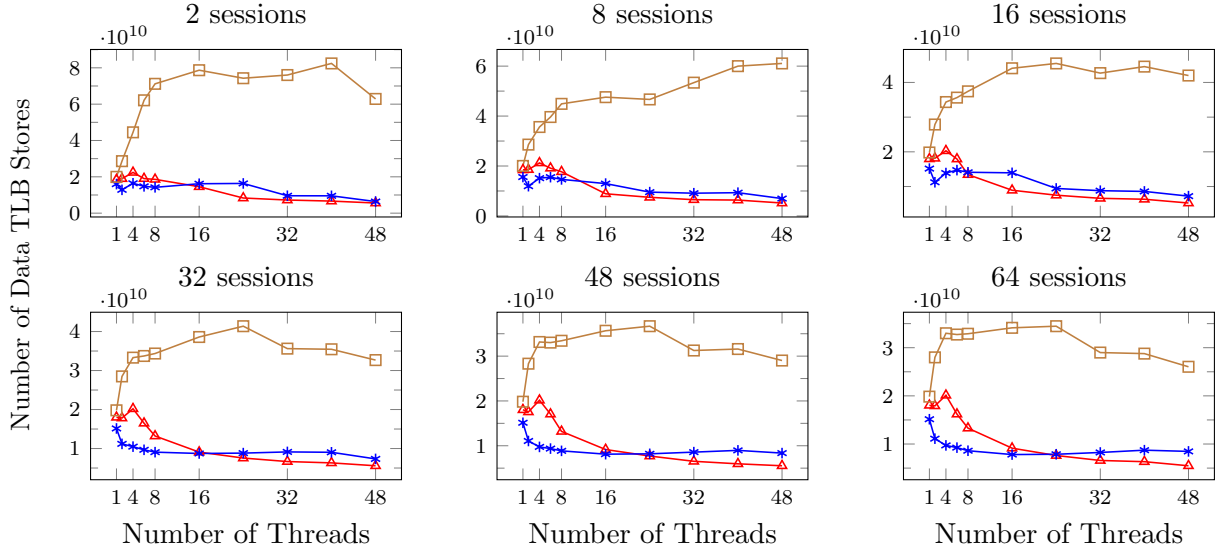
\begin{figure}[!t]
\centering
\begin{subfigure}[t]{\textwidth}
\begin{tikzpicture}
	\begin{groupplot}[group style={group size= 3 by 2, horizontal sep=0.5in, vertical sep=0.5in},height=1.5in,width=2.25in, max space between ticks=20, minor tick num=1,tick label style={font=\scriptsize}]
		\nextgroupplot[title=2 sessions,xtick={1,4,8,16,32,48}]
				\addplot[red,semithick,mark=triangle] 	[select coords between index={0}{9}] table[x=ThreadCount,y=GLB-GME,col sep=space]{Data/lonestar/threadSweep-dtlb-stores.csv};  \label{plots:GLB-GME:9}
				\addplot[blue,semithick,mark=asterisk] 	[select coords between index={0}{9}] table[x=ThreadCount,y=FS-GME,col sep=space]{Data/lonestar/threadSweep-dtlb-stores.csv};  \label{plots:FS-GME:9}
				\addplot[brown, semithick,mark=square] 	[select coords between index={0}{9}] table[x=ThreadCount,y=BH-GME,col sep=space]{Data/lonestar/threadSweep-dtlb-stores.csv};  \label{plots:BH-GME:9}
				\coordinate (top) at (rel axis cs:0,1);
		\nextgroupplot[title=8 sessions,xtick={1,4,8,16,32,48}]
				\addplot[red,semithick,mark=triangle] 	[select coords between index={10}{19}] table[x=ThreadCount,y=GLB-GME,col sep=space]{Data/lonestar/threadSweep-dtlb-stores.csv};
				\addplot[blue,semithick,mark=asterisk] 	[select coords between index={10}{19}] table[x=ThreadCount,y=FS-GME,col sep=space]{Data/lonestar/threadSweep-dtlb-stores.csv};  
				\addplot[brown, semithick,mark=square] 	[select coords between index={10}{19}] table[x=ThreadCount,y=BH-GME,col sep=space]{Data/lonestar/threadSweep-dtlb-stores.csv};
		\nextgroupplot[title=16 sessions,xtick={1,4,8,16,32,48}]
				\addplot[red,semithick,mark=triangle] 	[select coords between index={20}{29}] table[x=ThreadCount,y=GLB-GME,col sep=space]{Data/lonestar/threadSweep-dtlb-stores.csv};
				\addplot[blue,semithick,mark=asterisk] 	[select coords between index={20}{29}] table[x=ThreadCount,y=FS-GME,col sep=space]{Data/lonestar/threadSweep-dtlb-stores.csv};  
				\addplot[brown, semithick,mark=square] 	[select coords between index={20}{29}] table[x=ThreadCount,y=BH-GME,col sep=space]{Data/lonestar/threadSweep-dtlb-stores.csv};
		\nextgroupplot[title=32 sessions,xlabel={Number of Threads},xtick={1,4,8,16,32,48}]
				\addplot[red,semithick,mark=triangle] 	[select coords between index={30}{39}] table[x=ThreadCount,y=GLB-GME,col sep=space]{Data/lonestar/threadSweep-dtlb-stores.csv};
				\addplot[blue,semithick,mark=asterisk] 	[select coords between index={30}{39}] table[x=ThreadCount,y=FS-GME,col sep=space]{Data/lonestar/threadSweep-dtlb-stores.csv};  
				\addplot[brown, semithick,mark=square] 	[select coords between index={30}{39}] table[x=ThreadCount,y=BH-GME,col sep=space]{Data/lonestar/threadSweep-dtlb-stores.csv};
		\nextgroupplot[title=48 sessions,xlabel={Number of Threads},xtick={1,4,8,16,32,48}]
				\addplot[red,semithick,mark=triangle] 	[select coords between index={40}{49}] table[x=ThreadCount,y=GLB-GME,col sep=space]{Data/lonestar/threadSweep-dtlb-stores.csv};
				\addplot[blue,semithick,mark=asterisk] 	[select coords between index={40}{49}] table[x=ThreadCount,y=FS-GME,col sep=space]{Data/lonestar/threadSweep-dtlb-stores.csv};  
				\addplot[brown, semithick,mark=square] 	[select coords between index={40}{49}] table[x=ThreadCount,y=BH-GME,col sep=space]{Data/lonestar/threadSweep-dtlb-stores.csv};
		\nextgroupplot[title=64 sessions,xlabel={Number of Threads},xtick={1,4,8,16,32,48}]
				\addplot[red,semithick,mark=triangle] 	[select coords between index={50}{59}] table[x=ThreadCount,y=GLB-GME,col sep=space]{Data/lonestar/threadSweep-dtlb-stores.csv};
				\addplot[blue,semithick,mark=asterisk] 	[select coords between index={50}{59}] table[x=ThreadCount,y=FS-GME,col sep=space]{Data/lonestar/threadSweep-dtlb-stores.csv};  
				\addplot[brown, semithick,mark=square] 	[select coords between index={50}{59}] table[x=ThreadCount,y=BH-GME,col sep=space]{Data/lonestar/threadSweep-dtlb-stores.csv};
				\coordinate (bot) at (rel axis cs:1,0);
	\end{groupplot}
	\path (top-|current bounding box.west)-- node[anchor=south,rotate=90] {\small Number of Data TLB Stores} (bot-|current bounding box.west);
	\path (top|-current bounding box.north)-- coordinate(legendpos) (bot|-current bounding box.north);
	\matrix[matrix of nodes, anchor=south, draw, inner sep=0.2em] at ([yshift=4ex]legendpos)
    {
	   \ref{plots:FS-GME:9}   & \FSGME{} & [10pt]
	   \ref{plots:GLB-GME:9}    & \GLBGME{}   & [10pt]
	   \ref{plots:BH-GME:9}    & \BHGME{} \\
	};

\end{tikzpicture}
\caption{Uniform session distribution}
\end{subfigure}
\begin{subfigure}[t]{\textwidth}
\begin{tikzpicture}
	\begin{groupplot}[group style={group size= 3 by 2, horizontal sep=0.5in, vertical sep=0.5in},height=1.5in,width=2.25in, max space between ticks=20, minor tick num=1,tick label style={font=\scriptsize}]
		\nextgroupplot[title=2 sessions,xtick={1,4,8,16,32,48}]
				\addplot[red,semithick,mark=triangle] 	[select coords between index={0}{9}] table[x=ThreadCount,y=GLB-GME,col sep=space]{Data/lonestar/threadSweep-dtlb-stores-nu.csv};  \label{plots:GLB-GME:10}
				\addplot[blue,semithick,mark=asterisk] 	[select coords between index={0}{9}] table[x=ThreadCount,y=FS-GME,col sep=space]{Data/lonestar/threadSweep-dtlb-stores-nu.csv};  \label{plots:FS-GME:10}
				\addplot[brown, semithick,mark=square] 	[select coords between index={0}{9}] table[x=ThreadCount,y=BH-GME,col sep=space]{Data/lonestar/threadSweep-dtlb-stores-nu.csv};  \label{plots:BH-GME:10}
				\coordinate (top) at (rel axis cs:0,1);
		\nextgroupplot[title=8 sessions,xtick={1,4,8,16,32,48}]
				\addplot[red,semithick,mark=triangle] 	[select coords between index={10}{19}] table[x=ThreadCount,y=GLB-GME,col sep=space]{Data/lonestar/threadSweep-dtlb-stores-nu.csv};
				\addplot[blue,semithick,mark=asterisk] 	[select coords between index={10}{19}] table[x=ThreadCount,y=FS-GME,col sep=space]{Data/lonestar/threadSweep-dtlb-stores-nu.csv};  
				\addplot[brown, semithick,mark=square] 	[select coords between index={10}{19}] table[x=ThreadCount,y=BH-GME,col sep=space]{Data/lonestar/threadSweep-dtlb-stores-nu.csv};
		\nextgroupplot[title=16 sessions,xtick={1,4,8,16,32,48}]
				\addplot[red,semithick,mark=triangle] 	[select coords between index={20}{29}] table[x=ThreadCount,y=GLB-GME,col sep=space]{Data/lonestar/threadSweep-dtlb-stores-nu.csv};
				\addplot[blue,semithick,mark=asterisk] 	[select coords between index={20}{29}] table[x=ThreadCount,y=FS-GME,col sep=space]{Data/lonestar/threadSweep-dtlb-stores-nu.csv};  
				\addplot[brown, semithick,mark=square] 	[select coords between index={20}{29}] table[x=ThreadCount,y=BH-GME,col sep=space]{Data/lonestar/threadSweep-dtlb-stores-nu.csv};
		\nextgroupplot[title=32 sessions,xlabel={Number of Threads},xtick={1,4,8,16,32,48}]
				\addplot[red,semithick,mark=triangle] 	[select coords between index={30}{39}] table[x=ThreadCount,y=GLB-GME,col sep=space]{Data/lonestar/threadSweep-dtlb-stores-nu.csv};
				\addplot[blue,semithick,mark=asterisk] 	[select coords between index={30}{39}] table[x=ThreadCount,y=FS-GME,col sep=space]{Data/lonestar/threadSweep-dtlb-stores-nu.csv};  
				\addplot[brown, semithick,mark=square] 	[select coords between index={30}{39}] table[x=ThreadCount,y=BH-GME,col sep=space]{Data/lonestar/threadSweep-dtlb-stores-nu.csv};
		\nextgroupplot[title=48 sessions,xlabel={Number of Threads},xtick={1,4,8,16,32,48}]
				\addplot[red,semithick,mark=triangle] 	[select coords between index={40}{49}] table[x=ThreadCount,y=GLB-GME,col sep=space]{Data/lonestar/threadSweep-dtlb-stores-nu.csv};
				\addplot[blue,semithick,mark=asterisk] 	[select coords between index={40}{49}] table[x=ThreadCount,y=FS-GME,col sep=space]{Data/lonestar/threadSweep-dtlb-stores-nu.csv};  
				\addplot[brown, semithick,mark=square] 	[select coords between index={40}{49}] table[x=ThreadCount,y=BH-GME,col sep=space]{Data/lonestar/threadSweep-dtlb-stores-nu.csv};
		\nextgroupplot[title=64 sessions,xlabel={Number of Threads},xtick={1,4,8,16,32,48}]
				\addplot[red,semithick,mark=triangle] 	[select coords between index={50}{59}] table[x=ThreadCount,y=GLB-GME,col sep=space]{Data/lonestar/threadSweep-dtlb-stores-nu.csv};
				\addplot[blue,semithick,mark=asterisk] 	[select coords between index={50}{59}] table[x=ThreadCount,y=FS-GME,col sep=space]{Data/lonestar/threadSweep-dtlb-stores-nu.csv};  
				\addplot[brown, semithick,mark=square] 	[select coords between index={50}{59}] table[x=ThreadCount,y=BH-GME,col sep=space]{Data/lonestar/threadSweep-dtlb-stores-nu.csv};
				\coordinate (bot) at (rel axis cs:1,0);
	\end{groupplot}
	\path (top-|current bounding box.west)-- node[anchor=south,rotate=90,yshift=0.7em] {\small Number of Data TLB Stores} (bot-|current bounding box.west);
\end{tikzpicture}
\caption{Non-uniform session distribution}
\end{subfigure}
\caption{Comparison of number of data TLB (DTLB) stores  of different algorithms.}
\label{fig:lonestar:DTLB|stores}
\end{figure} 

\pgfplotsset
{
	select coords between index/.style 2 args=
	{
    x filter/.code=
		{
        \ifnum\coordindex<#1\def\pgfmathresult{}\fi
        \ifnum\coordindex>#2\def\pgfmathresult{}\fi
    }
	}
}

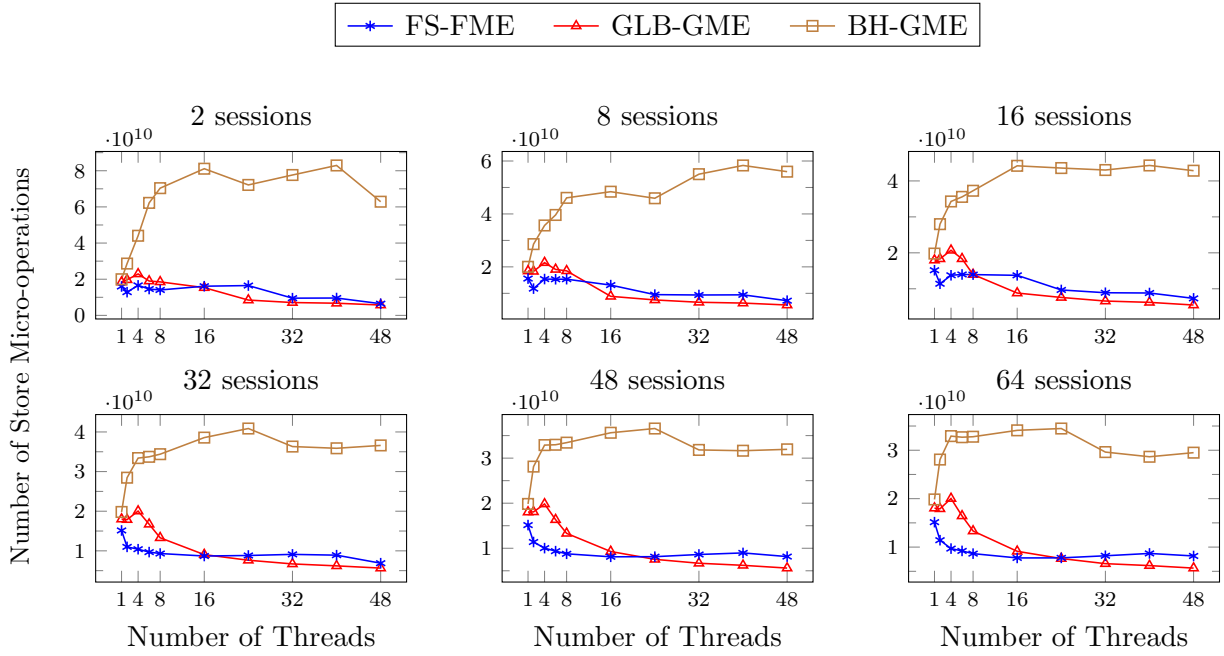
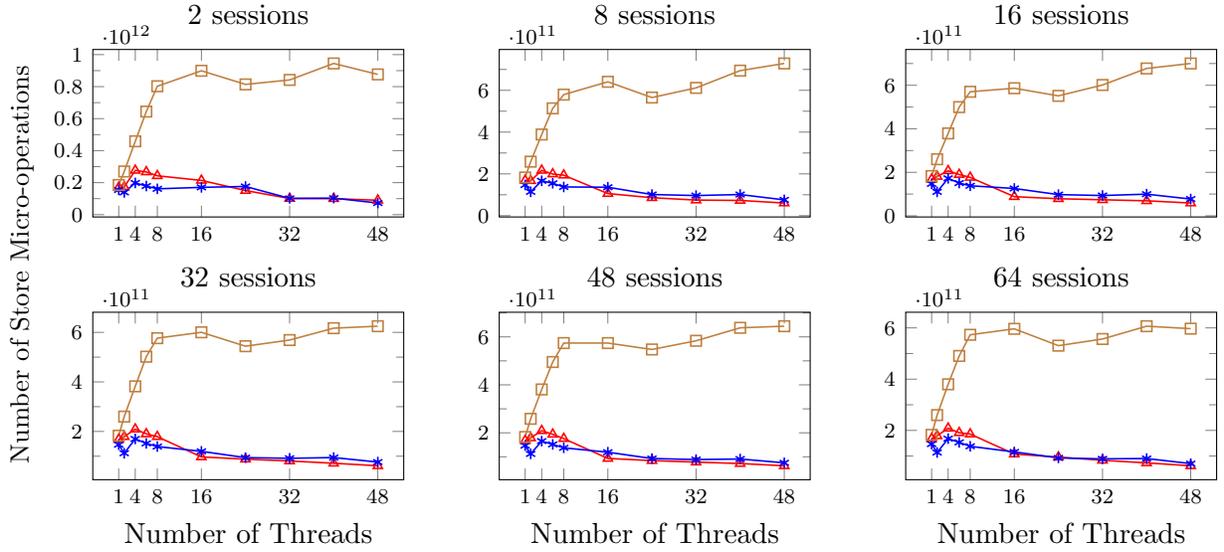
\begin{figure}[!t]
\centering
\begin{subfigure}{\textwidth}
\begin{tikzpicture}
	\begin{groupplot}[group style={group size= 3 by 2, horizontal sep=0.5in, vertical sep=0.5in},height=1.5in,width=2.25in, max space between ticks=20, minor tick num=1,tick label style={font=\scriptsize}]
		\nextgroupplot[title=2 sessions,xtick={1,4,8,16,32,48}]
				\addplot[red,semithick,mark=triangle] 	[select coords between index={0}{9}] table[x=ThreadCount,y=GLB-GME,col sep=space]{Data/lonestar/threadSweep-stores.csv};  \label{plots:GLB-GME:11}
				\addplot[blue,semithick,mark=asterisk] 	[select coords between index={0}{9}] table[x=ThreadCount,y=FS-GME,col sep=space]{Data/lonestar/threadSweep-stores.csv};  \label{plots:FS-GME:11}
				\addplot[brown, semithick,mark=square] 	[select coords between index={0}{9}] table[x=ThreadCount,y=BH-GME,col sep=space]{Data/lonestar/threadSweep-stores.csv};  \label{plots:BH-GME:11}
				\coordinate (top) at (rel axis cs:0,1);
		\nextgroupplot[title=8 sessions,xtick={1,4,8,16,32,48}]
				\addplot[red,semithick,mark=triangle] 	[select coords between index={10}{19}] table[x=ThreadCount,y=GLB-GME,col sep=space]{Data/lonestar/threadSweep-stores.csv};
				\addplot[blue,semithick,mark=asterisk] 	[select coords between index={10}{19}] table[x=ThreadCount,y=FS-GME,col sep=space]{Data/lonestar/threadSweep-stores.csv};  
				\addplot[brown, semithick,mark=square] 	[select coords between index={10}{19}] table[x=ThreadCount,y=BH-GME,col sep=space]{Data/lonestar/threadSweep-stores.csv};
		\nextgroupplot[title=16 sessions,xtick={1,4,8,16,32,48}]
				\addplot[red,semithick,mark=triangle] 	[select coords between index={20}{29}] table[x=ThreadCount,y=GLB-GME,col sep=space]{Data/lonestar/threadSweep-stores.csv};
				\addplot[blue,semithick,mark=asterisk] 	[select coords between index={20}{29}] table[x=ThreadCount,y=FS-GME,col sep=space]{Data/lonestar/threadSweep-stores.csv};  
				\addplot[brown, semithick,mark=square] 	[select coords between index={20}{29}] table[x=ThreadCount,y=BH-GME,col sep=space]{Data/lonestar/threadSweep-stores.csv};
		\nextgroupplot[title=32 sessions,xlabel={Number of Threads},xtick={1,4,8,16,32,48}]
				\addplot[red,semithick,mark=triangle] 	[select coords between index={30}{39}] table[x=ThreadCount,y=GLB-GME,col sep=space]{Data/lonestar/threadSweep-stores.csv};
				\addplot[blue,semithick,mark=asterisk] 	[select coords between index={30}{39}] table[x=ThreadCount,y=FS-GME,col sep=space]{Data/lonestar/threadSweep-stores.csv};  
				\addplot[brown, semithick,mark=square] 	[select coords between index={30}{39}] table[x=ThreadCount,y=BH-GME,col sep=space]{Data/lonestar/threadSweep-stores.csv};
		\nextgroupplot[title=48 sessions,xlabel={Number of Threads}, xtick={1,4,8,16,32,48}]
				\addplot[red,semithick,mark=triangle] 	[select coords between index={40}{49}] table[x=ThreadCount,y=GLB-GME,col sep=space]{Data/lonestar/threadSweep-stores.csv};
				\addplot[blue,semithick,mark=asterisk] 	[select coords between index={40}{49}] table[x=ThreadCount,y=FS-GME,col sep=space]{Data/lonestar/threadSweep-stores.csv};  
				\addplot[brown, semithick,mark=square] 	[select coords between index={40}{49}] table[x=ThreadCount,y=BH-GME,col sep=space]{Data/lonestar/threadSweep-stores.csv};
		\nextgroupplot[title=64 sessions,xlabel={Number of Threads},xtick={1,4,8,16,32,48}]
				\addplot[red,semithick,mark=triangle] 	[select coords between index={50}{59}] table[x=ThreadCount,y=GLB-GME,col sep=space]{Data/lonestar/threadSweep-stores.csv};
				\addplot[blue,semithick,mark=asterisk] 	[select coords between index={50}{59}] table[x=ThreadCount,y=FS-GME,col sep=space]{Data/lonestar/threadSweep-stores.csv};  
				\addplot[brown, semithick,mark=square] 	[select coords between index={50}{59}] table[x=ThreadCount,y=BH-GME,col sep=space]{Data/lonestar/threadSweep-stores.csv};
				\coordinate (bot) at (rel axis cs:1,0);
	\end{groupplot}
	\path (top-|current bounding box.west)-- node[anchor=south,rotate=90,yshift=0.7em] {\small  Number of Store Micro-operations} (bot-|current bounding box.west);
	\path (top|-current bounding box.north)-- coordinate(legendpos) (bot|-current bounding box.north);
	\matrix[matrix of nodes, anchor=south, draw, inner sep=0.2em] at ([yshift=4ex]legendpos)
    {
	   \ref{plots:FS-GME:11}   & \FSGME{} & [10pt]
	   \ref{plots:GLB-GME:11}    & \GLBGME{}   & [10pt]
	   \ref{plots:BH-GME:11}    & \BHGME{} \\
	};

\end{tikzpicture}
\caption{Uniform session distribution}
\end{subfigure}
\begin{subfigure}{\textwidth}
\begin{tikzpicture}
	\begin{groupplot}[group style={group size= 3 by 2, horizontal sep=0.5in, vertical sep=0.5in},height=1.5in,width=2.25in, max space between ticks=20, minor tick num=1,tick label style={font=\scriptsize}]
		\nextgroupplot[title=2 sessions,xtick={1,4,8,16,32,48}]
				\addplot[red,semithick,mark=triangle] 	[select coords between index={0}{9}] table[x=ThreadCount,y=GLB-GME,col sep=space]{Data/lonestar/threadSweep-stores-nu.csv};  \label{plots:GLB-GME:12}
				\addplot[blue,semithick,mark=asterisk] 	[select coords between index={0}{9}] table[x=ThreadCount,y=FS-GME,col sep=space]{Data/lonestar/threadSweep-stores-nu.csv};  \label{plots:FS-GME:12}
				\addplot[brown, semithick,mark=square] 	[select coords between index={0}{9}] table[x=ThreadCount,y=BH-GME,col sep=space]{Data/lonestar/threadSweep-stores-nu.csv};  \label{plots:BH-GME:12}
				\coordinate (top) at (rel axis cs:0,1);
		\nextgroupplot[title=8 sessions,xtick={1,4,8,16,32,48}]
				\addplot[red,semithick,mark=triangle] 	[select coords between index={10}{19}] table[x=ThreadCount,y=GLB-GME,col sep=space]{Data/lonestar/threadSweep-stores-nu.csv};
				\addplot[blue,semithick,mark=asterisk] 	[select coords between index={10}{19}] table[x=ThreadCount,y=FS-GME,col sep=space]{Data/lonestar/threadSweep-stores-nu.csv};  
				\addplot[brown, semithick,mark=square] 	[select coords between index={10}{19}] table[x=ThreadCount,y=BH-GME,col sep=space]{Data/lonestar/threadSweep-stores-nu.csv};
		\nextgroupplot[title=16 sessions,xtick={1,4,8,16,32,48}]
				\addplot[red,semithick,mark=triangle] 	[select coords between index={20}{29}] table[x=ThreadCount,y=GLB-GME,col sep=space]{Data/lonestar/threadSweep-stores-nu.csv};
				\addplot[blue,semithick,mark=asterisk] 	[select coords between index={20}{29}] table[x=ThreadCount,y=FS-GME,col sep=space]{Data/lonestar/threadSweep-stores-nu.csv};  
				\addplot[brown, semithick,mark=square] 	[select coords between index={20}{29}] table[x=ThreadCount,y=BH-GME,col sep=space]{Data/lonestar/threadSweep-stores-nu.csv};
		\nextgroupplot[title=32 sessions,xlabel={Number of Threads},xtick={1,4,8,16,32,48}]
				\addplot[red,semithick,mark=triangle] 	[select coords between index={30}{39}] table[x=ThreadCount,y=GLB-GME,col sep=space]{Data/lonestar/threadSweep-stores-nu.csv};
				\addplot[blue,semithick,mark=asterisk] 	[select coords between index={30}{39}] table[x=ThreadCount,y=FS-GME,col sep=space]{Data/lonestar/threadSweep-stores-nu.csv};  
				\addplot[brown, semithick,mark=square] 	[select coords between index={30}{39}] table[x=ThreadCount,y=BH-GME,col sep=space]{Data/lonestar/threadSweep-stores-nu.csv};
		\nextgroupplot[title=48 sessions,xlabel={Number of Threads},xtick={1,4,8,16,32,48}]
				\addplot[red,semithick,mark=triangle] 	[select coords between index={40}{49}] table[x=ThreadCount,y=GLB-GME,col sep=space]{Data/lonestar/threadSweep-stores-nu.csv};
				\addplot[blue,semithick,mark=asterisk] 	[select coords between index={40}{49}] table[x=ThreadCount,y=FS-GME,col sep=space]{Data/lonestar/threadSweep-stores-nu.csv};  
				\addplot[brown, semithick,mark=square] 	[select coords between index={40}{49}] table[x=ThreadCount,y=BH-GME,col sep=space]{Data/lonestar/threadSweep-stores-nu.csv};
		\nextgroupplot[title=64 sessions,xlabel={Number of Threads},xtick={1,4,8,16,32,48}]
				\addplot[red,semithick,mark=triangle] 	[select coords between index={50}{59}] table[x=ThreadCount,y=GLB-GME,col sep=space]{Data/lonestar/threadSweep-stores-nu.csv};
				\addplot[blue,semithick,mark=asterisk] 	[select coords between index={50}{59}] table[x=ThreadCount,y=FS-GME,col sep=space]{Data/lonestar/threadSweep-stores-nu.csv};  
				\addplot[brown, semithick,mark=square] 	[select coords between index={50}{59}] table[x=ThreadCount,y=BH-GME,col sep=space]{Data/lonestar/threadSweep-stores-nu.csv};
				\coordinate (bot) at (rel axis cs:1,0);
	\end{groupplot}
	\path (top-|current bounding box.west)-- node[anchor=south,rotate=90] {\small Number of Store Micro-operations} (bot-|current bounding box.west);
\end{tikzpicture}
\caption{Non-uniform session distribution}
\end{subfigure}
\caption{Comparison of number of store micro-operations of different algorithms.}
\label{fig:lonestar:store|uops}
\end{figure}

In this section, we present our experimental results of evaluating different GME algorithms.

\subsection{Different Group Mutual Exclusion Algorithms}

We compare the performance of the following implementations of GME algorithms:
\begin{enumerate}[label=(\alph*)]
\item the GME algorithm proposed by Bhatt and Huang~\cite{BhaHua:2010:PODC}, which is based on $f$-array data structure~\cite{Jay:2002:PODC}, denoted by \BHGME{}, 
\item the GME algorithm proposed by He {\em et al.}~\cite{HeGop+:2016:ICDCN}, which is a generalization of the classical Lamport's Bakery algorithm, denoted by \GLBGME{}, and
\item the GME algorithm presented in this work, denoted by \FSGME{}.
\end{enumerate}
		
We chose \GLBGME{} and \BHGME{} for comparison due to the following reasons. First, to our knowledge, \BHGME{} has the best RMR complexity among all existing GME algorithms, and \GLBGME{} is the most recently proposed GME algorithm. Second, both algorithms satisfy the First-Come-First-Serve (FCFS) property---relaxed in the case of \BHGME{} and strict in the case of \GLBGME{}. Additionally, \BHGME{} also satisfies the pulling property. Third, \BHGME{} uses load-linked and store-conditional (\LLSC{}) RMW instructions whereas \GLBGME{} does not use any RMW instruction. 

Since the system we used in our experiment did not support \LLSC{} instruction, 
we modified \BHGME{} to use \CAS{} instructions instead by using a timestamp (packed into the same word as the variable) to detect any writes to the variable. We used at least 32 bits for a timestamp, 
which never rolled over in our experiments. 

To our knowledge, no current implementations of \GLBGME{} and \BHGME{} exist (confirmed with the authors) so we implemented them ourselves. 
All implementations were written in C/C++.

\subsection{Experimental Setup}

\myparagraph{System used}

We conducted our experiments on a dual socket Intel Xeon E5-2690 v3  processor
consisting of 12 2.6~GHz cores per socket with hyper-threading enabled, yielding 48 logical cores in total, and 64~GB RAM.
The machine had 64~KB L1 cache (32~KB each of instruction and data) per core, 256~KB L2 cache per core and 30~MB L3 cache per socket. 
In addition, the machine had 128-entry instruction TLB (ITLB) and 64-entry data TLB (DTLB).
We used g++ compiler with optimization flags set to -O3.

\myparagraph{Experimental parameters}

To comparatively evaluate different implementations, we considered the following parameters:

\begin{enumerate}[leftmargin=*]
	
	\item \textbf{Number of Different Sessions:} We considered six different values of 2, 8, 16, 32, 48 and 64.
		
	\item \textbf{Distribution of Sessions:} We considered two different session distributions: 
	\begin{enumerate*}[label=(\alph*)]
	   \item \emph{uniform:} all session types are requested with the same probability.
	   \item \emph{non-uniform:} different session types are requested with different probabilities. In our experiments, we assumed that two session types are requested 90\% of the time and the remaining 10\% of the time (90/10 distribution)~\cite{Pla:2017:PhD}. 
	\end{enumerate*}

	\item \textbf{Maximum Degree of Contention:} This depends on number of threads that can concurrently request entry to their critical sections. We varied the number of threads from 1 to 48 in suitable increments.
	
\end{enumerate}

\myparagraph{Testing framework}

In each run of the experiment, every thread repeatedly generated requests for a (single) GME lock. Upon obtaining the lock, in its critical section, each thread executed an RMW instruction (\FAA) on one shared variable and a simple write instruction on a certain number of local variables (chosen randomly between 1 and 100 each time). The non-critical section was essentially empty.

\myparagraph{Run details}

For the uniform distribution, each experiment was run for eight seconds and the results were averaged over ten runs. For the non-uniform distribution, each experiment was run for two minutes and the results were averaged over five runs. Longer running time was required to conform to the desired probability distribution.  To generate random numbers, we used the Mersenne Twister pseudo-random number generator.
For both experiments, each run had a two second ``warm-up'' phase whose numbers were excluded from the calculations to minimize the effect of initial caching on the computed statistics.

\myparagraph{Evaluation metric}
We compared the performance of different implementations with respect to \emph{system throughput}, which is given by the number of critical section executions completed per unit time.


\subsection{Results}

\Cref{fig:lonestar:throughput}  depicts the system throughput of the three GME algorithms for the parameter values discussed above.
As the graphs clearly show, \FSGME{} outperformed the other two GME algorithms in \emph{almost all} the cases.
The difference was really stark at medium and larger thread count values when the throughput of \FSGME{} was sometimes as much as \gainnextbest{}  \emph{higher} than the next best performer. 
Even though, \BHGME{} has the lowest (worst-case) RMR complexity among the three algorithms, it had the worst performance. 

To understand the reasons for the differences in the performance, we used Linux performance analyzing tool \texttt{perf}. Specifically, we 
measured the following metrics for the three GME algorithms:
\begin{enumerate*}[label=(\alph*)]
\item number of L3 cache references,
\item number of data TLB stores, 
\item number of retired store micro-operations, and
\item number of branch instructions.
\end{enumerate*}
The third metric measures store micro-operations across the entire cache hierarchy.
The above four metrics for the three algorithms are shown in 
\cref{fig:lonestar:l3|references,fig:lonestar:DTLB|stores,fig:lonestar:store|uops,fig:lonestar:branches}, respectively.
(Many other options offered by \texttt{perf} tool were not supported by the system hardware.)

We believe that the reasons for the poor performance of \GLBGME{} compared to \FSGME{} are:
\begin{enumerate*}[label=(\roman*)]
\item \GLBGME{} has much higher RMR step complexity than \FSGME{}, and
\item \GLBGME{} satisfies strict FCFS property.
\end{enumerate*} 
Recall that \GLBGME{} has $\Omega(\n)$ RMR complexity.
In the \entry{} section of \GLBGME{}, a thread examines the request of every other thread and busy waits on that request to complete if it conflicts with its own and has a higher priority. 
As the graphs in \cref{fig:lonestar:l3|references} confirm, \GLBGME{} generates significantly larger number of L3 cache references  than the other two algorithms (implying worse performance with respect to L1 and L2 caches) and, moreover, the gap grows with the number of threads.
Further, with regard to FCFS property, as either the number of threads or the number of different sessions increases, the probability that requests of different threads conflict also increases. Joung proved analytically in \cite{Jou:2000:DC}  that, as the likelihood of conflicts increases, a GME algorithm that satisfies strict FCFS property will degenerate to a traditional ME algorithm in which only one thread is able to execute its critical section at a time.

We believe that the reason for the poor performance of \BHGME{} compared to \FSGME{} is its poor memory access pattern. 
In the \entry{} section of \BHGME{}, a thread has to perform many checks before it can enter its critical section. 
As the graphs in \cref{fig:lonestar:branches} show, the execution history of \BHGME{} exhibited higher branching compared to that of \FSGME{} and \GLBGME{}.
Excessive branching is undesirable and may adversely impact the performance of an algorithm significantly because branching inhibits many of the compiler and hardware optimizations.
Further, the graphs in \cref{fig:lonestar:DTLB|stores,fig:lonestar:store|uops} indicate that \BHGME{} has much higher number of store operations with respect to data TLB as well as cache hierarchy.
Finally, \texttt{perf-record} and \texttt{perf-annotate} tools also indicated that $f$-array based queue operations were the bottleneck and responsible for a large fraction of the execution time (of \BHGME{}). A more efficient implementation of a concurrent priority queue may help improve the performance of \BHGME{}.

For the non-uniform case, we conducted experiments using 80/20 and 70/30 session distributions as well. The gap between our GME algorithm and the other two GME algorithms narrowed by 10-15\%, but the trend was still the same.

We also conducted experiments in which threads were bound to cores using \linebreak \texttt{pthread\_setaffinity\_np()} 
function available in \texttt{sched.h} library. We observed that binding threads to cores had no significant impact on the performance and, thus, we have not included those results here.

\section{An Optimal GME Algorithm for DSM Model}

\begin{algorithm}[h]
\Additional \Shared \\
\Indp
$\readyArray$: \Array[$1\ldots\n$] of NodePtr\tcp*[r]{used for spinning - $\readyArray[i]$ is local to process $p_i$}
\label{line:ready|array}
\Indm
\BlankLine\BlankLine
\tcp{changes to \TryToEnter{} method - replace \crefrange{line:trytoenter:spin:adjourned|begin}{line:trytoenter:spin:adjourned|end} with \crefrange{line:tryotenter:ready|initialize}{line:tryotenter:ready:reset|end}}
$\readyArray[\myid]$ := $\current$\tcp*[r]{the node hosting the current session}
\label{line:tryotenter:ready|initialize}
\If{\IsAdjourned($\current \pointer \gate$)}
{
    $\readyArray[\myid]$ := \Null{}\tcp*[r]{session already \adjourned{} - no need to spin}
}
\While(\tcp*[f]{spin until the entry contains \Null pointer}){($\readyArray[\myid]$ $\neq$ \Null{})}          
{ 
   \label{line:tryotenter:ready:reset|begin}
   \tcp*[r]{do nothing}
}
\label{line:tryotenter:ready:reset|end}
\BlankLine\BlankLine
\tcp{notify a specific process to stop spinning}
\Notify(\Integer $i$, NodePtr $\node$) \\
\Begin{
\remove{
	$\current$ := $\announceArray[i]$\tcp*[r]{locate the request node of the process}
   	$\hpArray[\myid][2]$ := $\current$\tcp*[r]{declare it as a hazard pointer}
   	\If{($\current$ = \Null{}) \LOr{} ($\announceArray[i]$ $\neq$ $\current$)}
   	{
      \tcp{either process has no outstanding request or its request has already been fulfilled}
   	  \Continue\;
   	}
   	\BlankLine
   	\If{($\current \pointer \instance$ = $\node \pointer \instance$)}
   {
   	  \tcp{process has an outstanding request for the same GME object}
      \CAS($\readyArray[i]$, $\node$, \Null{})\tcp*[r]{signal the process to stop spinning}
   }
}
   \CAS($\readyArray[i]$, $\node$, \Null{})\tcp*[r]{signal the process to stop spinning}
   \label{line:notify:CAS}
}
\BlankLine\BlankLine
\tcp{notify all processes to stop spinning}
\NotifyAll(NodePtr $\node$) \\
\Begin{
   \lForEach{$i \in [1,\n]$}{\Notify($i$, $\node$)}
   \label{line:notifyall:for}
}
\BlankLine\BlankLine
\tcp{changes to the \AppendNode{} method - insert \cref{line:appendnode:notify} just after \cref{line:appendnode:advance|head}}
\Notify($\successor \pointer \owner$, $\current$)\;
\label{line:appendnode:notify}
\BlankLine\BlankLine
\tcp{changes to the \TryToEnter{} method - insert \cref{line:trytoenter:notify} just before \cref{line:trytoenter:retire|leader}}
\NotifyAll($\mynode \pointer \previousinlist$)\;
\label{line:trytoenter:notify}
\BlankLine\BlankLine
\tcp{changes to the \TryToLeave{} method - replace \cref{line:trytoleave:vacant} with \cref{line:trytoleave:vacant|notify}}
\lIf{\SetVacantFlag($\current$)}{\NotifyAll($\current$)}
\label{line:trytoleave:vacant|notify}
\caption{Changes for the DSM Model.}
\label{algo:dsm}
\end{algorithm}

In the DSM model, the lower bound on the RME step complexity of a request is $\Omega(\n)$. We show how to modify our GME algorithm to achieve this lower bound while maintaining all the other desirable properties 

The main idea is that, instead of busy waiting on session \state{} (until it \adjourn[s]), a process busy waits on a variable in its local memory (but still accessible to other processes); the local memory for process $p_i$ is denoted by $\readyArray[i]$. A process notifies a spinning process that the relevant session has \adjourned{} under the following conditions:
\begin{enumerate*}[label=(\arabic*)]
\item if it is the last process to leave the session provided it is also responsible for \adjourn[ing] the session,
\item if it is the leader of the next session, or
\item if it is trying to establish a new session and the spinning process is the leader of the new session.
\end{enumerate*}
To ensure that only relevant processes are notified, a process stores the address of the node hosting the session it is waiting to \adjourn[] in the location it will spin on (\emph{i.e.}. $\readyArray[i]$ for process $p_i$). A process notifies a spinning process that the session it is waiting to be \adjourned has indeed \adjourned by resetting the spin location to a \Null{} pointer using a \CAS{} instruction provided the location contains the address of the host node of the sesssion.

\section{Related Work}
\label{sec:related}

Several algorithms  have been proposed to solve the GME problem for shared-memory systems in the last two decades~\cite{Jou:2000:DC,KeaMoi:1999:PODC,Had:2001:PODC,TakIga:2003:COCOON,JayPet+:2003:PODC,DanHad:2004:DISC,BhaHua:2010:PODC,HeGop+:2016:ICDCN}.
Most of the earlier algorithms use only read and write instructions whereas many of the later algorithms use atomic instructions as well.
Different algorithms provide different fairness, concurrency and performance guarantees.

Many GME algorithms use a \emph{traditional} or an \emph{abortable} mutual exclusion (ME) algorithm as a subroutine.
The GME algorithm proposed by Keane and Moir in~\cite{KeaMoi:1999:PODC} uses a traditional ME algorithm as an exclusive lock to protect access to \entry and \exit sections of the algorithm. As such, this algorithm does not satisfy bounded exit and concurrent entering properties.
The GME algorithms presented in~\cite{DanHad:2004:DISC,BhaHua:2010:PODC} use an abortable ME algorithm as a subroutine. The main idea is that a process can enter its critical section using multiple pathways: 
\begin{enumerate*}[label=(\roman*)]
\item as a ``leader'' by establishing a new session, or
\item as a ``follower'' by joining an existing session. 
\end{enumerate*}
The first case occurs if the process is able to acquire the exclusive lock. 
The second case occurs if the process learns that a session ``compatible'' with its own request is already in progress in which case it aborts the ME algorithm and joins that session.
Both pathways are explored concurrently and, as soon as one of them allows the process enter its critical section, the other one is abandoned.

\subsection{Fairness and Concurrency Guarantees}

In many (group) mutual exclusion algorithms, the \entry section consists of two distinct subsections: a \emph{doorway} and a \emph{waiting-room}. A doorway is the wait-free portion of the \entry section that a process can complete within a bounded number of its own steps. 
A waiting-room of the \entry section is the portion where a process is blocked until it is its turn to execute its critical section. 

We say that two active processes are \emph{fellow} processes if they are requesting the same session (of the same GME object) and \emph{conflicting} processes if they are requesting different sessions (of the same GME object).

We say that an active process $p$ \emph{doorway-preceeds} another active process $q$ if $p$ completes the doorway before 
$q$ enters the doorway. 
Besides the four properties listed in \cref{sec:model|definition}, a GME algorithm may satisfy one or more of the  properties listed below.
These properties, which were defined in~\cite{Had:2001:PODC, JayPet+:2003:PODC, BhaHua:2010:PODC}, describe additional guarantees that a GME algorithm may provide.

\begin{description} 
\item[(P\propcount) Strong Concurrent Entering] If a process $p$ has completed its doorway, and $p$ doorway-precedes
every active conflicting  process, then $p$ enters its critical section within a bounded number of its own steps.
\label{pty:sce}
\item[(P\propcount) First-Come-First-Served (FCFS)] If  $p$ and $q$ are two conflicting processes such that $p$ doorway-preceeds $q$, then
$p$ enters its critical section before $q$. 
\item[(P\propcount) Relaxed FCFS] If  $p$ and $q$ are two conflicting processes such that $p$ doorway-preceeds $q$ but $q$ enter its critical section before $p$, then there exists another  process $r$ whose current attempt overlaps with that of $q$ such that $q$ and $r$ are fellow processes $p$ does not doorway-preceed $r$.
\item[(P\propcount) First-In-First-Enabled (FIFE)] If $p$ and $q$ are two fellow processes such that $p$ doorway-preceeds 
$q$ and $q$ enters its critical section before $p$, then $p$ can enter its critical section within a bounded number of its own steps.
\item[(P\propcount) Pulling] Suppose $p$ and $q$ are two fellow processes such that $p$ is currently in its critical section and doorway-preceeds all conflicting processes. If $q$ is currently in the waiting room, then $q$ can enter its critical section within a bounded number of its own steps.
\end{description}


\begin{table}[tp]
\centering
\tabulinesep=0.25em 
\resizebox{\textwidth}{!}{
\begin{tabu}{|X[12,l,m]|X[1,c,m]|X[1,c,m]|X[1,c,m]|X[1,c,m]|X[1,c,m]|X[1,c,m]|X[1,c,m]|X[1,c,m]|} \hline
\multicolumn{1}{|c|}{\textbf{Algorithm}} & \textbf{P2} & \textbf{P3} & \textbf{P4} & \textbf{P5} & \textbf{P6} & \textbf{P7} & \textbf{P8} & \textbf{P9} \\ 
\hline \hline
Joung \cite{Jou:2000:DC} & \cmark & \cmark & \cmark & \cmark & \xmark & \xmark & \xmark & \xmark  \\ 
\hline
Keane \& Moir \cite{KeaMoi:1999:PODC} & \cmark & \xmark & \xmark & \xmark & \xmark & \xmark & \xmark & \xmark \\
\hline
Hadzilacos \cite{Had:2001:PODC} & \cmark & \cmark & \cmark & \xmark & \cmark & \cmark & \xmark & \xmark \\
\hline
Takamura \& Igarashi \cite[Algorithm 1]{TakIga:2003:COCOON}  & \xmark & \cmark & \xmark & \xmark & \xmark & \xmark & \xmark & \xmark \\
\hline
Takamura \& Igarashi \cite[Algorithm 2]{TakIga:2003:COCOON}  &  \cmark & \xmark & \xmark & \xmark & \xmark & \xmark & \xmark & \xmark \\
\hline
Takamura \& Igarashi \cite[Algorithm 3]{TakIga:2003:COCOON}   &  \cmark & \xmark & \xmark & \xmark & \xmark & \xmark & \xmark & \xmark \\
\hline
Jayanti \emph{et al.} \cite[Algorithm 1]{JayPet+:2003:PODC}  &  \cmark & \cmark & \cmark & \xmark & \cmark & \cmark & \xmark & \xmark \\
\hline
Jayanti \emph{et al.} \cite[Algorithm 2]{JayPet+:2003:PODC}  &  \cmark & \cmark & \cmark & \cmark & \cmark & \xmark & \cmark & \xmark \\
\hline
Jayanti \emph{et al.} \cite[Algorithm 3]{JayPet+:2003:PODC}  &  \cmark & \cmark & \cmark & \cmark & \cmark & \xmark & \cmark & \xmark \\
\hline
Danek \& Hadzilacos \cite[Algorithm 1]{DanHad:2004:DISC}   &  \cmark & \cmark & \cmark & \cmark & \cmark & \cmark & \cmark & \xmark \\
\hline
Danek \& Hadzilacos \cite[Algorithm 2]{DanHad:2004:DISC}   &  \cmark & \cmark & \cmark & \xmark & \xmark & \cmark & \xmark & \xmark \\
\hline
Danek \& Hadzilacos \cite[Algorithm 3]{DanHad:2004:DISC}   &  \cmark & \cmark & \cmark & \xmark & \cmark & \cmark & \xmark & \xmark \\
\hline
Bhatt \& Huang \cite{BhaHua:2010:PODC} & \cmark & \cmark & \cmark & \xmark & \xmark & \cmark & \xmark & \cmark \\ 
\hline
He \emph{et al.} \cite[Algorithm 1]{HeGop+:2016:ICDCN}   &  \cmark & \cmark & \cmark & \xmark & \cmark & \xmark & \xmark & \xmark \\
\hline
He \emph{et al.} \cite[Algorithm 2]{HeGop+:2016:ICDCN}  &  \cmark & \cmark & \cmark & \xmark & \cmark & \xmark & \xmark & \xmark \\
\hline
Our Algorithm  [This Work] &  \cmark & \cmark & \cmark & \xmark & \xmark & \xmark & \xmark & \xmark \\
\hline
\end{tabu}
}
\caption{Fairness and concurrency properties satisfied by different algorithms. Note that all algorithms satisfy P1.}
\label{tab:fairness|concurrency}
\end{table}


\begin{table}[tp]
\centering
\tabulinesep=0.25em
\resizebox{\textwidth}{!}{
\begin{tabu} to 1.1\textwidth {|X[12,l,m]|X[5,c,m]|X[5,c,m]|X[5,c,m]|X[5,c,m]|X[6,c,m]|} \hline
\multicolumn{1}{|c|}{\textbf{Algorithm}} & \textbf{Space Complexity} & \textbf{Space Shareable Across Multiple Objects} & \textbf{Solitary Request Step Complexity} & \textbf{RMR Complexity} & \textbf{RMW Instructions} \\
\hline \hline
Yang \& Anderson's Algorithm~1 \cite{YanAnd:1995:DC}  & $O(\n)$ & \xmark & $O(\log \n)$ & $O(\log \n)$ & - \\ 
\hline
Mellor-Crummey \& Scott's Algorithm \cite{MelSco:1991:trcs}  & $O(1)$ & \cmark & $O(1)$ & $O(\n)$ & \FAS{} \\ 
\hline
\multicolumn{5}{l}{$\n$: number of processes} 
\end{tabu}
}
\caption{Complexity measures for ME algorithms used by some GME algorithms.}
\label{tab:complexity|ME}
\end{table}


\begin{table}[tp]
\centering
\tabulinesep=0.25em
\resizebox{\textwidth}{!}{
\begin{tabu} to 1.4\textwidth {|X[14,l,m]|X[7,c,m]|X[7,c,m]|X[9,c,m]|X[9,c,m]|X[6,c,m]|X[7,c,m]|} \hline
\multicolumn{1}{|c|}{\textbf{Algorithm}} & \textbf{Multi-Object Space Complexity} & \textbf{Solitary Request Step Complexity} & \textbf{Concurrent Entering  Step Complexity} & \textbf{RMR Complexity} & \textbf{Bounded Shared Variables} & \textbf{RMW Instructions} \\
\hline \hline
Joung \cite{Jou:2000:DC} &  $O(\m \n)$ & $\Omega(\n)$ & $\Omega(\n)$ & $\infty$ & \cmark & - \\ 
\hline
Keane \& Moir \cite{KeaMoi:1999:PODC} \newline  (with Yang \& Anderson's Algorithm~1)  & $O(\m \n)$ & $O(\log \n)$ & \na{} & $O(\log \n + \pc)$ & \cmark & - \\
\hline
Keane \& Moir \cite{KeaMoi:1999:PODC} \newline (with Mellor-Crummey and Scott's Algorithm) & $O(\m + \n)$  & $O(1)$ & \na{} & $O(\n)$ & \cmark & \FAS \\
\hline
Hadzilacos \cite{Had:2001:PODC} & $O(\m \n^2)$ & $\Omega(\n)$ & $\Omega(\n)$ & $O(\n + \pc^2)$ & \cmark & - \\ 
\hline
Takamura \& Igarashi \newline \cite[Algorithm 1]{TakIga:2003:COCOON}  &  $O(\m + \n)$ & $\Omega(\n)$ & \na{} & $\infty$ & \cmark & - \\ 
\hline
Takamura \& Igarashi \newline \cite[Algorithm 2]{TakIga:2003:COCOON}  &  $O(\m + \n)$ & $\Omega(\n)$ & \na{} & $O(\n)$ & \xmark & - \\
\hline
Takamura \& Igarashi \newline  \cite[Algorithm 3]{TakIga:2003:COCOON}   &  $O(\m + \n)$ & $\Omega(\n)$ & \na{} & $O(\n)$ & \xmark & - \\
\hline
Jayanti \emph{et al.} \newline \cite[Algorithm 1]{JayPet+:2003:PODC}  &  $O(\m \n)$ & $\Omega(\n)$ & $\Omega(\n)$ & $O(\n + \pc^2)$ & \cmark & - \\ \hline
He \emph{et al.} \newline \cite[Algorithm 1]{HeGop+:2016:ICDCN}  &  $O(\m + \n)$ & $\Omega(\n)$ & $\Omega(\n)$ & $O(\n)$ & \xmark & - \\
\hline
He \emph{et al.} \newline \cite[Algorithm 2]{HeGop+:2016:ICDCN}  &  $O(\m + \n)$ & $\Omega(\n)$ & $\Omega(\n)$ & $O(\n)$ & \cmark & - \\
\hline
Bhatt \& Huang \cite{BhaHua:2010:PODC} & $O(\m \n)$ & $O(1)$ & $O(\min\{\log \n, \pc\})$ & $O(\min\{\log \n, \pc\})$ & \xmark & \LL{}/\SC{} \\
\hline
Our Algorithm \newline [This Work] & $O(\m + \n^2)$ & $O(1)$ & $O(1)$ & $O(\pc)^\ast$ & \cmark & \CAS{} and \FAA{} \\
\hline
\multicolumn{6}{l}{\na{}: the algorithm does not satisfy P4} \\
\multicolumn{3}{l}{$\n$: number of processes} &
\multicolumn{3}{l}{$\m$: number of GME objects} \\
\multicolumn{3}{l}{$\s$: number of different types of sessions} & 
\multicolumn{3}{l}{$\pc$: point contention of the request} \\
\multicolumn{3}{l}{$\ast$: amortized case} 
\end{tabu}
}
\caption{Complexity measures of GME algorithms excluding those in~\cite{JayPet+:2003:PODC,DanHad:2004:DISC} that use an abortable ME algorithm as a subroutine.}
\label{tab:complexity}
\end{table}


\begin{table}[tp]
\centering
\tabulinesep=0.25em
\resizebox{\textwidth}{!}{
\begin{tabu} to 1.5\textwidth {|X[6,c,m]|X[10,l,m]|X[6,c,m]|X[6,c,m]|X[9,c,m]|X[9,c,m]|X[6,c,m]|X[7,c,m]|} \hline
\textbf{Abortable ME Algorithm} & \multicolumn{1}{|c|}{\textbf{Algorithm}} & \textbf{Space Complexity} & \textbf{Solitary Request Step Complexity} & \textbf{Concurrent Entering Step Complexity} & \textbf{RMR Complexity} & \textbf{Bounded Shared Variables} & \textbf{RMW Instructions} \\
\hline \hline
\multirow{5}{*}{$\n$-bit FCFS} & Jayanti \emph{et al.} \newline \cite[Algorithm~2]{JayPet+:2003:PODC} & $O(\m \n^2)$ & $\Omega(\n)$ & $\Omega(\n)$ & $O(\n)$ & \cmark  & -  \\ \cline{2-8}
& Jayanti \emph{et al.} \newline \cite[Algorithm~3]{JayPet+:2003:PODC} & $O(\m \n^2)$ & $\Omega(\n)$ & $\Omega(\n)$ & $O(\n)$ & \xmark & - \\ \cline{2-8}
 & Danek \& Hadzilacos \newline \cite[Algorithm~1]{DanHad:2004:DISC} & $O(\m \n^2)$ & $\Omega(\n)$ & $\Omega(\n)$ & $O(\n)$ & \cmark & - \\ 
\cline{2-8}
&  Danek \& Hadzilacos \newline \cite[Algorithm~2]{DanHad:2004:DISC} & $O(\m \n^2)$ & $\Omega(\n)$ & $\Omega(\n)$ & $O(\n)$ & \cmark & \CAS{} and \FAA{} \\ 
\cline{2-8}
&  Danek \& Hadzilacos \newline \cite[Algorithm~3]{DanHad:2004:DISC} & $O(\m \n^2 \s)$ & $O(\n \log \s)$ & $O(\n \log \s)$ & $O(\n \log \s)$ & \xmark & \CAS{} and \FAA{} \\ 
\hline
\multirow{5}{*}{\STAB{modified \\ Bakery \\  algorithm}} & Jayanti \emph{et al.} \newline \cite[Algorithm~2]{JayPet+:2003:PODC} & $O(\m \n^2)$ & $\Omega(\n)$ & $\Omega(\n)$ & $O(\n)$ & \xmark & - \\ \cline{2-8}
& Jayanti \emph{et al.} \newline \cite[Algorithm~3]{JayPet+:2003:PODC}  &  $O(\m \n)$ & $\Omega(\n)$ & $\Omega(\n)$ & $O(\n)$ & \xmark & -  \\ \cline{2-8}
&  Danek \& Hadzilacos \newline \cite[Algorithm~1]{DanHad:2004:DISC} & $O(\m + \n^2)$ & $\Omega(\n)$ & $\Omega(\n)$ & $O(\n)$ & \xmark & - \\ 
\cline{2-8}
&  Danek \& Hadzilacos \newline \cite[Algorithm~2]{DanHad:2004:DISC} & $O(\m + \n^2)$ & $\Omega(\n)$ & $\Omega(\n)$ & $O(\n)$ & \xmark & \CAS{} and \FAA{} \\ 
\cline{2-8}
&  Danek \& Hadzilacos \newline \cite[Algorithm~3]{DanHad:2004:DISC} & $O(\m \n \s)$ & $O(\n \log \s)$ & $O(\n \log \s)$ & $O(\n \log \s)$ & \xmark & \CAS{} and \FAA{} \\ 
\hline
\multirow{5}{*}{\STAB{Jayanti's \\ algorithm}} & Jayanti \emph{et al.} \newline \cite[Algorithm~2]{JayPet+:2003:PODC} & $O(\m \n^2)$ & $\Omega(\n)$ & $\Omega(\n)$ & $O(\n)$ & \xmark & - \\ \cline{2-8}
& Jayanti \emph{et al.} \newline \cite[Algorithm~3]{JayPet+:2003:PODC} & $O(\m \n)$ & $\Omega(\n)$ & $\Omega(\n)$ & $O(\n)$ & \xmark & - \\ \cline{2-8}
&  Danek \& Hadzilacos \newline \cite[Algorithm~1]{DanHad:2004:DISC} & $O(\m \n + \n^2)$ & $\Omega(\n)$ & $\Omega(\n)$ & $O(\n)$ & \xmark & - \\ 
\cline{2-8}
&  Danek \& Hadzilacos \newline \cite[Algorithm~2]{DanHad:2004:DISC} & $O(\m \n + \n^2)$ & $\Omega(\n)$ & $\Omega(\n)$ & $O(\n)$ & \xmark & \CAS{} and \FAA{} \\ 
\cline{2-8}
&  Danek \& Hadzilacos \newline \cite[Algorithm~3]{DanHad:2004:DISC} & $O(\m \n \s)$ & $O(\log \s)$ & $O\left(
\begin{array}{@{}c@{}}\log \s  \  \times \\ \min\{\log \n, \pc\} \end{array}
\right)$ 
& $O\left(
\begin{array}{@{}c@{}}\log \s  \  \times \\ \min\{\log \n, \pc\} \end{array}
\right)$ 
& \xmark & \CAS{} and \FAA{}  \\
\hline
\multicolumn{3}{l}{$\n$: number of processes} &
\multicolumn{3}{l}{$\m$: number of GME objects} \\
\multicolumn{3}{l}{$\s$: number of different types of sessions} & 
\multicolumn{3}{l}{$\pc$: point contention of the request} 
\end{tabu}
}
\caption{Complexity measures of the GME algorithms in~\cite{JayPet+:2003:PODC,DanHad:2004:DISC} using the three abortable mutex algorithms.}
\label{tab:complexity|abortable}
\end{table}

\Cref{tab:fairness|concurrency} compares all algorithms with respect to the properties they satisfy.
Takamura and Igarashi's GME algorithms~\cite{TakIga:2003:COCOON} do not satisfy \cref{pty:sce} because it is possible that processes requesting the same session can delay one another while executing the \entry section.

\subsection{Synchronization Instructions}

In addition to simple read and write instructions, GME algorithms may use one or more of the following RMW instructions: compare-and-swap (\CAS), fetch-and-add (\FAA), load-linked and store conditional (\LL{}/\SC{}) and fetch-and-store (\FAS) instructions. Compare-and-swap and fetch-and-add instructions are as defined in \cref{sec:model|definition}.

Load-linked and store-conditional  instructions are performed in pairs and behave in a similar manner as their simpler load and store counterparts, but with some additional features. A load-linked instruction takes a shared variable $x$ as input, and returns the current value of $x$ as output. 
A store-conditional instruction takes a shared variable $x$ and a value $v$ as inputs. If the value of $x$ has not been modified by any process since the \emph{associated} load-linked instruction was performed on $x$, it  overwrites the current value of $x$ with $v$,  and returns true as output. Otherwise, it leaves $x$ unchanged, and returns false as output.

A fetch-and-store instruction takes a shared variable $x$  and a value $v$ as inputs, returns the current value of $x$ as output and, at the same time, stores $v$ in $x$.

\subsection{Complexity Measures}

\Cref{tab:complexity,tab:complexity|abortable} shows the complexity measures of different GME algorithms with respect to the metrics described in \cref{sec:model|definition}.

Note that complexity measures for the algorithm in~\cite{KeaMoi:1999:PODC}, which uses a traditional ME algorithm 
as a subroutine,
and the algorithms in~\cite{JayPet+:2003:PODC,DanHad:2004:DISC,BhaHua:2010:PODC}, which use an abortabe ME algorithm  as a subroutine, depend on which ME algorithm is used.
Also, the work in~\cite{JayPet+:2003:PODC,DanHad:2004:DISC} describes multiple GME algorithms, each of which can be combined with one of the several abortable ME algorithms, to yield GME algorithms with different complexity measures.

To analyze the performance of the Keane and Moir's GME algorithm, we consider two different traditional ME algorithms, namely
\begin{enumerate*}[label=(\roman*)]
\item a tree-based algorithm by Yang and Anderson~\cite[Algorithm 1]{YanAnd:1995:DC}, which has small RMR complexity, and
\item a queue-based algorithm by Mellor-Crummey and Scott~\cite{MelSco:1991:trcs}, which has small space complexity as well as small solitary request step complexity.
\end{enumerate*}
Both use bounded space variables. 
\Cref{tab:complexity|ME} displays the performance of the two ME algorithms with respect to various complexity measures.

To analyze the performance of a GME algorithm that uses an abortable ME algorithm as a subroutine, we consider the following abortable ME algorithms, namely
\begin{enumerate*}[label=(\roman*)]
\item Bakery algorithm by Lamport\cite{Lam:1974:CACM} modified to support aborts~\cite{JayPet+:2003:PODC},
\item $\n$-bit FCFS algorithm by Lamport~\cite{Lam:1986:JACM}, and
\item an algorithm by Jayanti~\cite{Jay:2003:PODC}.
\end{enumerate*}
The second (middle) algorithm uses bounded space variables whereas the other two do not.
Note that the six GME algorithms in~\cite{JayPet+:2003:PODC,DanHad:2004:DISC} can be combined with each of the three abortable ME algorithm to yield eighteen GME algorithms with potentially different complexity measures. 
For clarity, the complexity measures of these eighteen GME algorithms are given in \cref{tab:complexity|abortable} and those of the remaining GME algorithms are given in \cref{tab:complexity}.

\section{Conclusion and Future Work}
\label{sec:conclusion|future}

In this work, we have presented a suite of GME algorithms for an asynchronous shared memory system, each successively building on and addressing the limitations of the previous algorithm.  Specifically, the final version uses bounded space variables and satisfies the four most important properties of the GME problem, namely group mutual exclusion, lockout freedom, bounded exit and concurrent entering. At the same, it has $O(1)$ step-complexity in the absence of any conflicting requests, and $O(1)$ space-complexity per GME object when the system contains $\Omega(\n)$ GME objects.
To the best of our knowledge, our algorithm is the \emph{first} GME algorithm that has constant complexity for both metrics. Finally, the RMR complexity of our GME algorithm in the general case depends on the contention encountered by a request.
In our experimental results, our GME algorithm vastly outperformed two of the well-known existing GME algorithms especially for higher thread counts.

As future work, we plan to extend our GME algorithm so that 
it provides stronger fairness or concurrency guarantees such as some combination of first-come-first-served (FCFS)~\cite{Had:2001:PODC}, first-in-first-enabled (FIFE)~\cite{JayPet+:2003:PODC}, strong concurrent entry~\cite{JayPet+:2003:PODC} and pulling~\cite{BhaHua:2010:PODC} among others. 
We also plan to investigate the \emph{trade-off} between the RMR complexity of a GME algorithm (in the presence of conflicting requests) and its space complexity with large number of GME objects under the CC model. 
At this point, it is not clear to us if we can design a GME algorithm that has $O(1)$ complexity for both the metrics.
Finally, we plan to extend our GME algorithm so that it has good RMR complexity under the DSM model.

\bibliography{Bibliography/citations}

\end{document}